\documentclass[letterpaper,journal]{IEEEtran}
\usepackage{amsmath}
\usepackage{amsfonts}
\usepackage{subcaption}
\usepackage{amssymb}
\usepackage{float}
\usepackage{enumerate}
\restylefloat{table}
\usepackage{tikz} 
\usetikzlibrary{arrows, positioning, automata}
\usetikzlibrary{shapes.geometric, arrows}
\newcommand*{\Comb}[2]{{}^{#1}C_{#2}}

\usepackage{tabu}

\def\SC{{\rm SC}}
\def\SCN{{\rm SCN}}
\def\LP{{\rm LP}}
\def\DPI{{\rm DPI}}
\def\DP{{\rm DP}}

\usepackage{amsmath, amssymb, bbm, xspace}
\usepackage{epsfig}
\usepackage{longtable}
\usepackage{color}
\usepackage{mathrsfs}
\usepackage{comment}

\usepackage{courier}

\newtheorem{theorem}{Theorem}[section]

\newtheorem{lemma}[theorem]{Lemma}

\newtheorem{proposition}[theorem]{Proposition}

\newtheorem{corollary}[theorem]{Corollary}

\newtheorem{definition}{Definition}[section]

\def\bkE{{\rm I\kern-.17em E}}
\def\bk1{{\rm 1\kern-.17em l}}
\def\bkD{{\rm I\kern-.17em D}}
\def\bkR{{\rm I\kern-.17em R}}
\def\bkP{{\rm I\kern-.17em P}}

\def\Fbb{{\mathbb{F}}}

\def\Vbb{{\mathbb{V}}}

\def\bkZ{{\bf{Z}}}

\def\bkE{{\rm I\kern-.17em E}}
\def\bk1{{\rm 1\kern-.17em l}}
\def\bkD{{\rm I\kern-.17em D}}
\def\bkR{{\rm I\kern-.17em R}}
\def\bkP{{\rm I\kern-.17em P}}

\makeatletter
\newcommand{\pushright}[1]{\ifmeasuring@#1\else\omit\hfill$\displaystyle#1$\fi\ignorespaces}
\newcommand{\pushleft}[1]{\ifmeasuring@#1\else\omit$\displaystyle#1$\hfill\fi\ignorespaces}
\makeatother

\def\bkZ{{\bf{Z}}}
\def\b12{(\beta_1,\beta_2)}

\newenvironment{proofarg}[1][]{\noindent\hspace{2em}{\itshape Proof #1: }}{\hspace*{\fill}~\qed\par\endtrivlist\unskip}
\newenvironment{example}{{\noindent \bf Example}}{\hfill $\square$\hspace{-4.5pt}\vspace{6pt}}
\newcounter{example}
\renewcommand{\theexample}{\thesection.\arabic{example}}
\newenvironment{examplec}[1][]{\refstepcounter{example}
\par\medskip \noindent%
   \textbf{Example~\theexample. #1} \rmfamily}{\hfill $\square$   \hspace{-4.5pt} \vspace{6pt}}

\newcounter{remark}
\renewcommand{\theremark}{\thesection.\arabic{remark}}
\newenvironment{remarkc}[1][]{\refstepcounter{remark}
\noindent{\itshape Remark~\theremark. #1} \rmfamily}{\hspace*{\fill}~$\square$\vspace{0pt}}

\def\t{^\top}
\def\Bscr{\mathscr{B}}
\def\Xscr{\mathcal{X}}
\def\Yscr{\mathcal{Y}}

\def\Ebb{\mathbb{E}}
\newlength{\noteWidth}
\long\def\notes#1{\ifinner
{\tiny #1}
\else
\marginpar{\parbox[t]{\noteWidth}{\raggedright\tiny #1}}
\fi\typeout{#1}}

 \def\notes#1{\typeout{read notes: #1}} 

\newcommand{\I}[1]{\mathbb{I}_{\{#1\}}}

\newcommand{\ie}{i.e.\@\xspace} 
\newcommand{\eg}{e.g.\@\xspace} 
\newcommand{\etal}{et al.\@\xspace} 

\newcommand{\Real}{\ensuremath{\mathbb{R}}}

\newcommand{\minimize}[1]{\displaystyle\minim_{#1}}
\newcommand{\minim}{\mathop{\hbox{\rm min}}}

\DeclareMathOperator*{\st}{subject\;to}

\def\OPT{{\rm OPT}}
\def\FEA{{\rm FEA}}

\def\Ebb{\mathbb{E}}

\def\Pbb{{\mathbb{P}}}
\def\Nbb{{\mathbb{N}}}
\def\Qbb{{\mathbb{Q}}}
\def\Ibb{{\mathbb{I}}}

\def\Dbf{{\bf D}}

\def\dbf{{\bf d}}
\def\fbf{{\bf f}}

\def\ast{\buildrel{t}\over\rightarrow}

\def\ext{{\rm ext}}

\def\exp{\mathop{\hbox{\rm exp}}}

\def\half  {{\textstyle{1\over 2}}}

\def\conv{\textrm{conv}\:}

\def\spose#1{\hbox to 0pt{#1\hss}}
\def\sub#1{^{\null}_{#1}}
\def\text #1{\hbox{\quad#1\quad}}

\def \dtilde{\tilde{d}}

\def\nthinsp{\mskip -2   mu}

\def\S{_{\scriptscriptstyle S}}

\def\superstar{^{\raise 0.5pt\hbox{$\nthinsp *$}}}
\def\SUPERSTAR{^{\raise 0.5pt\hbox{$*$}}}

\def\lamstarT {\lambda^{\raise 0.5pt\hbox{$\nthinsp *$}T}}

\def\Ascr{{\cal A}}
\def\Bscr{{\cal B}}

\def\Iscr{{\cal I}}
\def\Jscr{{\cal J}}

\def\Mscr{{\cal M}}

\def\Oscr{{\cal O}}
\def\Pscr{{\cal P}}

\def\Sscr{{\cal S}}
\def\Sscrhat{\widehat{\mathcal{S}}}

\def\Vscr{{\cal V}}

\def\Mscr{{\cal M}}

\def\Rscr{{\cal R}}

\def\Zscr{{\cal Z}}
\def\Xscr{{\cal X}}
\def\Yscr{{\cal Y}}

\def\dbar{\bar d}

\def\dtilde{\widetilde d}

\def\Nbar{\skew{4.4}\bar N}

\def\Pbar{\skew5\bar P}

\def\Qbar{\bar Q}
\def\Qhat{\widehat Q}

\def\sbar{\bar s}
\def\shat{\widehat s}

\def\Shat{\widehat S}

\def\Ybar{\skew2\bar Y}

\def\Dbf{{\bf D}}

\def\Pbf{{\bf P}}

\def\xbf{{\bf x}}
\def\shatbf{{\bf \shat}}
\def\ybf{{\bf y}}

\def\aur{\;\textrm{and}\;}

\def\non{\nonumber}

\let\forallnew\forall
\renewcommand{\forall}{\forallnew\ }
\let\forall\forallnew

		\def\bkE{{\rm I\kern-.17em E}}
		\def\bk1{{\rm 1\kern-.17em l}}
		\def\bkD{{\rm I\kern-.17em D}}
		\def\bkR{{\rm I\kern-.17em R}}
		\def\bkP{{\rm I\kern-.17em P}}
		\def\bkY{{\bf \kern-.17em Y}}
		\def\bkZ{{\bf \kern-.17em Z}}
		\def\bkC{{\bf  \kern-.17em C}}

{\begin{list}{}%
         {\setlength{\leftmargin}{#1}}%
         \item[]%
}
{\end{list}}

		\def\bsp{\begin{split}}
		\def\beq{\begin{eqnarray}}
		\def\bal{\begin{align*}}
		\def\bc{\begin{center}}
		\def\be{\begin{enumerate}}
		\def\bi{\begin{itemize}}
		\def\bs{\begin{small}}
		\def\bS{\begin{slide}}
		\def\ec{\end{center}}
		\def\ee{\end{enumerate}}
		\def\ei{\end{itemize}}
		\def\es{\end{small}}
		\def\eS{\end{slide}}
		\def\eeq{\end{eqnarray}}
		\def\eal{\end{align*}}
		\def\esp{\end{split}}
		\def\qed{ \vrule height7.5pt width7.5pt depth0pt}  

	\def\cp2problem#1#2#3#4{\fbox
		 {\begin{tabular*}{0.9\textwidth}
			{@{}l@{\extracolsep{\fill}}l@{\extracolsep{6pt}}l@{\extracolsep{\fill}}c@{}}
				#1 & & $#4 $ 
			\end{tabular*}}}

		\def\bkE{{\rm I\kern-.17em E}}
		\def\bk1{{\rm 1\kern-.17em l}}
		\def\bkD{{\rm I\kern-.17em D}}
		\def\bkR{{\rm I\kern-.17em R}}
		\def\bkP{{\rm I\kern-.17em P}}
		
		\def\bkZ{{\bf{Z}}}

\newcommand {\beeq}[1]{\begin{equation}\label{#1}}
\newcommand {\eeeq}{\end{equation}}
\newcommand {\bea}{\begin{eqnarray}}
\newcommand {\eea}{\end{eqnarray}}

\def\texitem#1{\par\smallskip\noindent\hangindent 25pt
               \hbox to 25pt {\hss #1 ~}\ignorespaces}

\def\st{\mbox{subject to}}

\def\bsp{\begin{split}}
		\def\beq{\begin{eqnarray}}
		\def\bal{\begin{align*}}
		\def\bc{\begin{center}}
		\def\be{\begin{enumerate}}
		\def\bi{\begin{itemize}}
		\def\bs{\begin{small}}
		\def\bS{\begin{slide}}
		\def\ec{\end{center}}
		\def\ee{\end{enumerate}}
		\def\ei{\end{itemize}}
		\def\es{\end{small}}
		\def\eS{\end{slide}}
		\def\eeq{\end{eqnarray}}
		\def\eal{\end{align*}}
		\def\esp{\end{split}}
		\def\qed{ \vrule height7.5pt width7.5pt depth0pt}  

\usepackage{amsmath, amssymb, xspace}
\usepackage{epsfig}
\usepackage{longtable}
\usepackage{color}
\usepackage{mathrsfs}
\newenvironment{proof}[1][]{{\noindent \emph {Proof} #1: }}{\hfill \qed \vspace{3pt}\\ }
\def\ext{{\rm ext}}

\def\sub{\hbox{\rm s.t}}

\def\conv{{\rm conv}}
			\def\problemsmalla#1#2#3#4{\fbox
		 {\begin{tabular*}{0.47\textwidth}
			{@{}l@{\extracolsep{\fill}}l@{\extracolsep{-4pt}}l@{\extracolsep{\fill}}c@{}}
				#1 &  & $\minimize{#2}$  $#3$ & $ $ \\[4pt]
					  $\sub \ $  &   & $#4$ &  $ $
			\end{tabular*}}
			}

			\renewcommand{\I}[1]{\mathbb{I}\{#1\}}
\ifCLASSINFOpdf
  \else
 
\fi
\author{Sharu Theresa Jose 
\qquad 
 Ankur A. Kulkarni \thanks{Sharu and Ankur are with the Systems and Control Engineering group at the  
Indian Institute of Technology Bombay in Mumbai, India 400076. They can be reached at sharutheresa@sc.iitb.ac.in, kulkarni.ankur@iitb.ac.in.}
\thanks{The results in this paper were presented in part at the IEEE Conference on Decision and Control, held in Osaka, Japan in 2015 \cite{jose2015linear} and at the National Conference on Communications in Chennai, India in 2017~\cite{jose2017linear}.}}
\title{Linear Programming based Converses for Finite Blocklength Lossy Joint Source-Channel Coding}
\begin{document}
\maketitle
\begin{abstract}
A linear programming (LP) based framework is presented for obtaining converses for finite blocklength lossy joint source-channel coding problems. The framework applies for any loss criterion, generalizes certain previously known converses, and also extends to multi-terminal settings. The finite blocklength problem is posed equivalently as a nonconvex optimization problem and using a lift-and-project-like method, a close but tractable LP relaxation of this problem is derived.
 Lower bounds on the original problem are obtained by the construction of feasible points for the dual of the LP relaxation. 
A particular application of this approach leads to new converses which recover and improve on the converses of Kostina and Verd{\'u} for finite blocklength lossy joint source-channel coding and lossy source coding. For finite blocklength channel coding, the LP relaxation recovers the converse of Polyanskiy, Poor and Verd\'{u} and leads to a new improvement on the converse of Wolfowitz, showing thereby that our LP relaxation is asymptotically tight with increasing blocklengths for channel coding, lossless source coding and joint source-channel coding with the excess distortion probability as the loss criterion. Using a duality based argument, a new converse is derived for finite blocklength joint source-channel coding for a class of source-channel pairs. Employing this converse, the LP relaxation is also shown to be tight for all blocklengths for the minimization of the expected average symbol-wise Hamming distortion of a $q$-ary uniform source over a $q$-ary symmetric memoryless channel for any $q\in \Nbb$. The optimization formulation and the lift-and-project method are extended to networked settings and demonstrated by obtaining an improvement on a converse of Zhou \etal for the successive refinement problem for successively refinable source-distortion measure triplets.
\end{abstract}
\begin{IEEEkeywords}
Converses, lossy joint source-channel coding, finite blocklength regime, linear programming relaxation, lift-and-project, strong duality.
\end{IEEEkeywords}
\section{Introduction}
A general problem of finite blocklength lossy joint source-channel coding can be framed as 
the following optimization problem, denoted SC:
$$
\problemsmalla{SC}
	{f,g}
	{\displaystyle \mathbb{E}[\kappa(S,X,Y,\Shat)]}
				 {\begin{array}{r@{\ }c@{\ }l}
				 				 X&=&f(S),\\
								 \widehat{S}&=&g(Y).
	\end{array}}
	$$ Here $S,X,Y,\Shat$ are random variables taking values in fixed spaces $\Sscr, \Xscr,\Yscr$ and $\Sscrhat$ respectively and $\kappa: \Sscr \times \Xscr \times \Yscr \times \Sscrhat \rightarrow \Real$ is a given loss function.  $S$ is a source signal distributed according to a given probability distribution $P_S$. An encoder $f:\Sscr \rightarrow \Xscr$  maps $S$ to an  encoded signal $X$ (see Fig~\ref{fig:sc}). The encoded signal is sent through a channel which given $X$ produces an output signal $Y$ according to a known channel law, denoted by $P_{Y|X}$, following which a decoder, $g:\Yscr \rightarrow \Sscrhat$, maps the channel output signal to a destination signal $\Shat$. Each pair $f,g$ induces a joint distribution on $\Sscr \times \Xscr \times \Yscr \times \Sscrhat$ and the expectation $\Ebb$ is with respect to this joint distribution. 
Problem SC seeks to minimize the expectation of the loss function $\kappa$ over all \textit{codes}, \ie, over all encoder-decoder pairs $(f,g)$.

If $\kappa(S,X,Y,\Shat) = d(S,\Shat)$ for a distortion function $d:\Sscr \times \Sscrhat \rightarrow \Real$, SC gives the  encoder-decoder pair $(f,g)$ yielding minimum expected distortion between $S$ and $\Shat$. When $\Sscrhat=\Sscr$ and $d(S,\Shat) = \I{S\neq \Shat}$, SC\footnote{$\I{\bullet}$ denotes the indicator function of `$\bullet$'.}  finds the code $(f,g)$ that minimizes the probability of error in the reproduction of a message $S$. Note that the spaces $\Sscr,\Xscr,\Yscr,\Sscrhat$ are taken as fixed, whereby problem SC corresponds to a fixed  blocklength setting (unit blocklength, if the alphabet is defined appropriately). 
In the infinite blocklength setting one has a sequence of problems SC parameterized by the blocklength and the spaces $\Sscr,\Xscr,\Yscr,\Sscrhat$ are structured as Cartesian products of smaller fixed spaces. Our interest in this paper is in the finite blocklength problem, and our main contribution is a new framework for obtaining \textit{lower bounds} or \textit{converses} for this problem.

   \begin{figure}[t]
\begin{center}
\includegraphics[scale=0.3,clip=true, trim = 0in 4.5in 0in 2.5in]{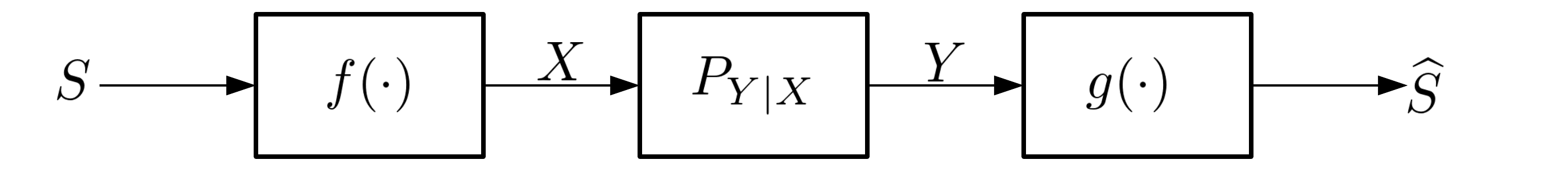} 
\caption{Source-Channel Communication System} \vspace{-.3cm}
\label{fig:sc}
\end{center}
\end{figure}

 Finite blocklength problems have great significance in practical communication systems like multimedia communication which cannot afford to have large delays. 
 However, with the exception of a few cases, such as a Gaussian source with mean square distortion measure across an additive white Gaussian noise channel with a power constraint \cite{goblick1965theoretical},
ascertaining the optimal performance under finite blocklength joint source-channel coding is a challenging problem. Consequently, a natural approach to these problems is to find tight bounds. 

Over the years there has been significant amount of work in deriving such bounds using a variety of tools and arguments. Classically, converses for lossy source coding have been found using $\fbf$-divergences~\cite{ziv73functionals}. 
For the case of channel coding, the state of the art is perhaps the work of Polyanskiy, Poor and Verd\'{u}~\cite{polyanskiy2010channel} (and its numerous follow ups) that employed hypothesis testing to derive converses. For joint source-channel coding with the probability of excess distortion as the loss criterion, 
Kostina and Verd\'{u} in \cite{kostina2013lossy} derived asymptotically tight lower and upper bounds using \textit{tilted information}. While many sharp converses have been discovered for specific loss criteria, what is perhaps unsatisfying is the absence of a common framework using which converses can be found for any loss criterion. Our central contribution is such a framework.

We introduce a linear programming (LP) based approach for obtaining lower bounds on SC that applies for any loss criterion and is also extendable to multi-terminal settings.  
A fundamental difficulty in the finite blocklength problem is the lack of a sufficient condition for establishing an expression as a lower bound on $\SC$. 
The linear programming approach we introduce remedies this. It asks for the construction of functions on subspaces of $\Sscr\times \Xscr\times \Yscr \times \Sscrhat$ such that they satisfy certain pointwise inequalities; any such construction leads to a lower bound.  As such the approach provides a general purpose framework from which specific converses can be derived.

The crux of the approach lies in the derivation of a close but tractable \textit{LP relaxation} of problem SC. Applying the duality theorem of linear programming we then get that the objective value of \textit{any} feasible point of the \textit{dual} of this LP relaxation is a lower bound on the optimal value of SC. The variables and constraints of the dual program are precisely the functions and the pointwise inequalities   mentioned above. 

\subsection{Main Results}
In general, relaxation of an optimization problem may lead to large and persistent departures from the true optimal. Our main results show that quite the opposite is the case with our LP relaxation while considering canonical loss criteria.
 With one dual feasible construction we recover the tilted information based lower bound derived by Kostina and Verd\'{u} \cite{kostina2013lossy} on the minimum excess distortion probability for a finite blocklength lossy joint source-channel code. By variants of this construction we then obtain two levels of improvements on the above converse of Kostina and Verd\'{u}. For finite blocklength lossy source coding, these results imply new lower bounds that improve on the tilted information based converse of Kostina and Verd\'{u} \cite{kostina2012fixed}. For finite blocklength channel coding, our LP relaxation  recovers the converse derived by Polyanskiy, Poor and Verd\'{u}~\cite{polyanskiy2010channel} using hypothesis testing and yields a new improvement on the converse of Wolfowitz~\cite{wolfowitz1968notes}.  
  It follows that the LP relaxation we derive is \textit{asymptotically tight with increasing blocklengths} for channel coding, lossless source coding and joint source-channel coding with the excess distortion probability as the loss criterion -- problems for which the Polyanskiy-Poor-Verd\'{u} and Kostina-Verd\'{u} converses are also tight. Moreover,
   our LP relaxation also implies the strong converse for channel coding.

Kostina and Verd\'{u} have also derived hypothesis testing based converses in \cite{kostina2013lossy} and \cite{kostina2012fixed}  that are known to be better than their tilted information based converses for the case of lossy transmission of a binary uniform source over a binary symmetric channel, and for the case of lossy source coding of  a binary memoryless source. When particularized to these cases our improved converses imply the above converses obtained by Kostina and Verd\'{u}  using hypothesis testing. 
 
While the Kostina-Verd\'{u} converse pertains only to the excess distortion criterion, the LP relaxation framework can be applied to any loss criterion. 
Using a duality based argument, we derive a new lower bound on the expected distortion of a finite blocklength lossy joint source-channel code for a class of channels which includes the binary symmetric channel and the binary erasure channel, amongst others.
Employing this bound, we show that our LP relaxation is \textit{tight for all blocklengths} for the problem of minimizing the expected average symbol-wise Hamming distortion of a $q$-ary uniform source over a $q$-ary symmetric, memoryless channel for any $q\in \Nbb$. The above source-channel pair is \textit{probabilistically matched} in the sense of~\cite{gastpar03tocode}. It is known that the finite blocklength optimal performance of this system is the same as that in the infinite blocklength setting; indeed it is $\epsilon$ where $\frac{\epsilon}{q-1} <\frac{1}{q}$ is the channel crossover  probability. We obtain the same result by showing there is a dual feasible point with objective value $\epsilon$, independently of the blocklength. 

Our method of obtaining LP relaxations can be mechanistically extended to find relaxations for networked settings. Analyzing the dual of the resulting LP relaxation could provide converses for very general problems in network information theory. As a first step in this regard, we consider the successive refinement source coding problem and improve on the converse of Zhou \etal  on joint excess distortion probability of successively refinable source-distortion measure triplets \cite[Lemma 9]{zhou2016successive} which generalizes the Kostina-Verd\'{u} lossy source coding converse to a network setting.
\subsection{The LP Relaxation}
To find the LP relaxation,  we first pose SC equivalently as a \textit{continuous optimization problem} over joint probability distributions, an approach studied in \cite{kulkarni2014optimizer} for stochastic control problems by the second author of the present paper. The resulting optimization problem has a linear objective but a necessarily \textit{nonconvex} feasible region thereby making the problem intractable in general. A natural `optimizer's' approach proposed in \cite{kulkarni2014optimizer} is to seek a convex relaxation of this problem. Since solutions of the original problem and of the relaxation lie on extreme points of the respective feasible regions, a desirable property of any good relaxation is that all extreme points of the feasible region of the original nonconvex problem be retained as extreme points of the relaxation. Our LP relaxation of SC has precisely this property. We argue that classical bounding arguments using the data processing inequality induced by $\fbf$-divergences, such as in~\cite{ziv73functionals}, amount to performing a relaxation that \textit{does not} enjoy this property.

A distinct feature of our LP relaxation is that it is obtained by appealing exclusively to convex analytic principles and does not rely on mutual information or similar other constructs of information theory. We first identify the underlying cause of nonconvexity as the presence of multiple bilinear terms which are coupled in the constraints of the continuous optimization formulation of SC. To obtain the relaxation we replace these terms by their respective convex hulls. Equivalently, we \textit{lift} this problem to a higher dimensional space by introducing new variables replacing the bilinear  terms, and \textit{project} it back on the original space by introducing new implied affine inequalities in the enlarged space. These additional constraints together with the affine constraints present in the original nonconvex problem give our LP relaxation. Since the relaxation only affects the feasible region and not the objective of SC, the relaxation applies to any choice of $\kappa$ and thereby for any loss criterion. 

We also note that the above manner of constructing relaxations extends readily to multi-terminal settings. A finite blocklength joint source-channel coding problem on a network admits a continuous optimization formulation as a particularly structured \textit{polynomial optimization problem}; in the point-to-point setting this polynomial is a bilinear function.

Coincidentally, the recent past has seen a spurt of results using LP duality for obtaining outer bounds for problems in combinatorial coding theory. In~\cite{kulkarni12nonasymptotic} nonasymptotic upper bounds for zero-error deletion correcting codes were derived  using LP duality by the second author of the present paper. This argument was soon refined and extended to other combinatorial channels, see \eg, \cite{fazeli2015generalized}, \cite{kashyap2014upper}, \cite{gabrys2015correcting}, \cite{cullina2016restricted}. Closely related to our work is the work of Matthews \cite{matthews2012linear} wherein he posed the hypothesis testing based channel coding converse of Polyanskiy, Poor and Verd\'{u} \cite{polyanskiy2010channel} as the optimal solution of a linear program obtained by relaxing the problem to non-signaling codes.  The LP relaxation proposed in this paper may be viewed as a logical extension of the above lines of work to the problem of joint source-channel coding. Consequently, the results in this paper conceptually situate distinct converses -- of Kostina-Verd\'{u} and of Polyanskiy-Poor-Verd\'{u} (thereby lossless source and channel coding), the probabilistically matched case of a $q$-ary uniform source and $q$-ary symmetric channel with Hamming distortion, the converse of Zhou \etal for successive refinement problem and several recent converses from combinatorial coding theory -- within a larger unified class of  convex analytic or duality-based converses. This is attractive from the point of view of understanding these converses
and for the promise it holds for a general unified theory of converses for problems in information theory. 

The holy grail in joint source-channel coding is probably a formal understanding of how the complex geometry of the combinatorial finite blocklength problem transforms into smooth and convex characterizations in the large blocklength limit. The LP relaxation, being an approximation of this problem via supporting hyperplanes, serves as a modest tool for understanding this geometry. 
On a related note, the asymptotic tightness of the LP relaxation could be useful as an analytical lemma, and may thereby be of independent interest. 



 \subsection{Outline}

This paper is organized as follows. Section~\ref{sec:prelims} consists of some preliminaries, including some notation and the optimization concepts we need. Section~\ref{sec:optimisiation} consists of the optimization formulation, the LP relaxation and a discussion of the properties of the relaxation. Section~\ref{sec:channelBSC} explains the construction of dual variables to derive converses through the example of channel coding of a binary symmetric channel. Section~\ref{lowerbounds} consists of our main results on the application of duality for obtaining converses. Section~\ref{sec:BMS-BSC} focusses on numerical examples and certain new converses for lossy transmission of a binary memoryless source over a binary symmetric channel. Section~\ref{sec:bsc} discusses a new general duality based converse for finite blocklength lossy joint source-channel coding. The extension to a networked setting is discussed in Section~\ref{sec:networked} and we conclude in Section~\ref{conclusion}.
Appendices are included in Section~\ref{sec:appendices}.

\section{Preliminaries}\label{sec:prelims}
\subsection{Notation}
All random variables in this paper are discrete. Let $\mathcal{P}(\cdot)$ represent the set of (joint) probability distributions on `$\cdot$' and $P \in \mathcal{P}(\cdot)$ or $Q \in \mathcal{P}(\cdot)$ represent specific (joint) distributions. These distributions are interpreted as column vectors in a finite dimensional Euclidean space. 
If $Q$ is a joint probability distribution, let $Q_{\bullet}$ denote the marginal distribution of `$\bullet$'. For example, $Q_{X|S}$ represents the vector with $Q_{X|S}(x|s)$ for  $x \in \mathcal{X},s\in \mathcal{S}$ as its components. 
In general, we use capital letters $A,B,C$ to represent random variables, the corresponding calligraphic letters $\mathcal{A},\mathcal{B},\mathcal{C}$ represent the space or alphabet of these random variables and small letters $a,b,c$ to denote their specific values. 
We use $\mathcal{Z}$ to denote $\mathcal{Z} := \mathcal{S} \times \mathcal{X} \times \mathcal{Y} \times \mathcal{\widehat{S}}$ and $z :=(s,x,y,\shat)\in \mathcal{Z}$. 
For any $z\in \Zscr$, we use $P_SQ_{X|S} P_{Y|X}Q_{\widehat{S}|Y}(z) $ to represent the  product, $P_S(s)Q_{X|S}(x|s)P_{Y|X}(y|x)Q_{\widehat{S}|Y}(\shat|y)$.

 $I(A;B)$ represents the mutual information between random variables $A$ and $B$ and $\I{\bullet}$ represents the indicator function of the event `$\bullet$' which is equal to one when `$\bullet$' is true and is zero otherwise. 
A \textit{string} or \textit{sequence} is a vector of symbols from a given alphabet. We use $\Fbb_q:=\{0,\hdots,q-1\}$ to represent the $q$-ary alphabet and $\Fbb_q^n$ to represent the set of all $q$-ary strings of length $n$. 
For any strings $u,v$ of the same length, 
  we use  $d_{u,v}$ to represent the Hamming distance between $u$ and $v$ (\ie number of positions at which corresponding symbols differ in $u$ and $v$). If $u$ is a binary string, we use $w_{u}$ to represent the Hamming weight of $u$ (\ie number of ones in the binary string $u$). The abbreviations LHS and RHS stand for Left Hand Side and Right Hand Side, respectively. LP stands for linear program or linear program\textit{ming}, based on the context.

\subsection{Convex hull, valid inequalities and duality}	  
This paper relies on some concepts of optimization; although most of them are found in standard literature (\eg,~\cite{rockafellar97convex} and \cite{boyd04convex}), we recount them here in the context of the challenges encountered in this paper. 

A set $K \subseteq \Real^n$ is said to be convex if for any $x,y \in K$ and  $\alpha \in (0,1)$, the convex combination $ \alpha x+(1-\alpha) y \in K$. The \textit{convex hull} of a set $K$, denoted $\conv(K)$ is the intersection of all convex sets containing $K$.  
A \textit{halfspace} is a set of the form $\{x \in \Real^n | a\t x \leq b\}$ where $a$ is a vector in $ \Real^n$ and $b$ is a scalar. $K$ is a \textit{polyhedron} if it is the intersection of finitely many halfspaces; these halfspaces constitute its \textit{halfspace representation}. A bounded polyhedron, called a  \textit{polytope}, also admits another equivalent representation. By the Minkowski-Weyl theorem~\cite{rockafellar97convex}, a set is a polytope if and only if it is the convex hull of finitely many points;  these points may be taken as its \textit{extreme points} and constitute its \textit{vertex representation}. A point $x$ of a set $K$ is an extreme point if it cannot be written as  a convex combination of two distinct elements of $K$, \ie,  if for any $y,z \in K$ and $\alpha \in (0,1)$ we have 
$x= \alpha y +(1-\alpha) z,$ 
then we must have $x=y=z.$ We use $\ext(K)$ to denote the set of extreme points of $K$. 

For an optimization problem $P$ involving the minimization of a continuous function $f_0$ over a closed set  $K$,
\begin{align}
\min_x \qquad & f_0(x) \label{eq:opti} \tag{$P$}\\
\st \qquad & x \in K, \non
\end{align}
$f_0$ is referred to as the \textit{objective function}, points in $K$ are called \textit{feasible} and the set $K$ is called the \textit{feasible region} (denoted $\FEA(P)$). $\OPT(P)$ denotes its optimal value. $K$ is often expressed as $K=\{x \in \Real^n | f_1(x)\leq 0, f_2(x)=0\}$, where the vector-valued functions $f_1,f_2$ are referred to as \textit{constraints}. 
Problem \eqref{eq:opti} is a \textit{convex optimization problem} if $f_0,f_1$ are convex and $f_2$ is affine; in this case $K$ is a convex set. 
Problem \eqref{eq:opti} is a \textit{linear program} if $f_0,f_1,f_2$ are all affine; in this case $K$ is a polyhedron. By introducing a new variable, say $y$, to represent $f_0(x)$ we may equivalently write \eqref{eq:opti} as 
\begin{align*}
\min_{x,y} \qquad & y, \\
\st \qquad & y\geq f_0(x),\\
& x \in K.
\end{align*}
Thus, from now on,  without loss of generality, we consider $f_0$ in \eqref{eq:opti} to be a linear function. For such a problem, a solution lies on an extreme point of the feasible region (if there exists an extreme point).

 \eqref{eq:opti} is a \textit{nonconvex} optimization problem if $\FEA \eqref{eq:opti}$ is not convex. 
Nonconvex optimization problems lack an easily verifiable characterization of optimality whereby these  problems are, in general, extremely difficult to solve both analytically and computationally. Problem SC is of this kind. 
A \textit{convex relaxation} of  \eqref{eq:opti} is the problem, 
\begin{align}
\min_x \qquad & f_0(x) \label{eq:relax} \tag{$P'$}\\
\st \qquad & x\in K', \non
\end{align}
where $K'$ is a {convex} set that contains $K.$ If $K'$ is a polyehdron, then \eqref{eq:relax} is a linear program, and hence a \textit{linear programming relaxation}  of \eqref{eq:opti}. 

LP relaxations together with mathematical programming duality provide a clean framework for obtaining bounds on optimization problems. 
Corresponding to any minimization problem  there exists a related \textit{maximization} problem called the \textit{dual problem} whose optimal value is a \textit{lower bound} on the optimal value of the original problem (referred to as the \textit{primal}). 
LPs are particularly attractive because they satisfy \textit{strong duality}, \ie, 
\begin{theorem}[Strong Duality] \label{thm:strongdual} 
If either the primal LP or its dual problem has a finite optimal value, then so does the other and their optimal values are equal.
\end{theorem}
And, moreover, the dual of an LP is itself an LP and it  is known in an explicit form. For example, if the primal has the following form, 
\begin{align*} (\textbf{P}) \quad \quad \qquad  \min_x \quad &c\t x \\
\mbox{subject to} \quad \quad   Ax &= b, \\ 
x &\geq 0,
\end{align*}where $c \in \mathbb{R}^n, b \in \mathbb{R}^m,A \in \mathbb{R}^{m \times n}$, 
its dual problem is,
\begin{align*} ( \textbf{D}) \quad \quad \qquad   \max_y \quad &b\t y\\
\mbox{subject to} \quad \quad    A\t y &\leq c.
\end{align*} 
 Consequently, if an LP relaxation of \eqref{eq:opti} is found, a systematic way to obtain a lower bound on \eqref{eq:opti} is to find a point $y$ that is feasible for the dual of this LP. 
Specifically, if $(\Pbf)$ is an LP relaxation of a nonconvex problem \eqref{eq:opti} and $\OPT(\Pbf)$ is finite, then  Theorem~\ref{thm:strongdual} gives,
\[\OPT\eqref{eq:opti} \geq \OPT(\Pbf) =\OPT(\Dbf) \geq b\t y,\]
for any $y \in \FEA(\Dbf)$ (\ie, $y$ such that $A\t y\leq c $).

While the gap in the second inequality above can be made to vanish via the right choice of $y$ (thanks to Theorem~\ref{thm:strongdual}), the gap in the first inequality is fundamental and can only be improved by obtaining a tighter relaxation $(\Pbf)$. This puts the onus on discovering an LP relaxation that closely approximates \eqref{eq:opti}. 

If $K$ is compact and $\conv(K)$ happens to be a polyhedron, the relaxation \eqref{eq:relax} with $K'=\conv(K)$ is guaranteed to be exact. But finding the halfspace representation of the convex hull is in general hard\footnote{For binary integer programs with $n$ variables, the number of halfspaces required seem to be to the tune of $n^n$. The reader may see~\cite{conforti2014integer} for more.} which makes it hard to express \eqref{eq:relax} as a LP. An alternative to this situation is to seek \textit{valid inequalities}.   
An inequality ``$a\t x \leq b$", where $0\neq a\in \Real^n,b\in \Real$ is said to be \textit{valid} for $K$ if $K \subseteq \{x | a\t x \leq b\}.$ Although finding nontrivial valid inequalities for an arbitrary set is also not straightforward, it is often possible to exploit the algebraic nature of the constraints of $K$ to infer valid inequalities. For example, consider the nonconvex set,
\[K=\{(w,x_1,x_2) \in \Real^3| w=x_1x_2, x_1 \in [l_1 ,u_1], x_2 \in [l_2 ,u_2]\}. \]
Then one can show that the following inequalities are valid for $K$,
  \begin{alignat}{2}
w &\leq u_2x_1+l_1x_2-l_1u_2, \quad& &w \leq l_2x_1+u_1y_1-u_1l_2, \label{eq:mccormick1} \\
 w &\geq u_2x_1+u_1x_2-u_1u_2, \quad& &w \geq l_2x_1+l_1x_2-l_1l_2. \label{eq:mccormick2} 
\end{alignat} 
To see \eqref{eq:mccormick1}, observe that if $(w,x_1,x_2) \in K$ then, $(x_1-l_1)(u_2-x_2) \geq 0.$ Likewise, 
$(x_2-l_2)(u_1-x_1) \geq 0$. Expanding and substituting $w=x_1x_2$ shows \eqref{eq:mccormick1} are valid. Similarly, one can use that $(u_1-x_1)(u_2-x_2) \geq 0$ and $(x_1-l_1)(x_2-l_2)\geq 0$ to check the validity of \eqref{eq:mccormick2}. It follows that  
$K'=\{(w,x_1,x_2) \in \Real^3 |\ \eqref{eq:mccormick1} \aur \eqref{eq:mccormick2} \ {\rm hold}\},$
is a polyhedron containing $K$. These simple observations are in fact quite powerful. It is known that $K'$ is in fact equal to $\conv(K)$ \cite{al1983jointly}. Inequalities \eqref{eq:mccormick1}-\eqref{eq:mccormick2} are called the \textit{McCormick inequalities} --  \eqref{eq:mccormick2} are the \textit{convex under-estimating inequalities} and \eqref{eq:mccormick1} are the \textit{concave over-estimating inequalities}. The bilinear product $w$ is sandwiched between the two sets of inequalities.  These arguments when systematically generalized lead to the so-called \textit{lift and project} method~\cite{conforti2014integer} or \textit{reformulation linearization technique}~\cite{sherali91bilinear}, \cite{sherali1992global} which discover valid inequalities for polynomial optimization problems by multiplying constraints.


Above story is quick account of the challenges encountered in SC. The continuous optimization formulation of SC has a linear objective and a feasible region $\FEA(\SC)$ that is nonconvex but with finitely many extreme points. Thus $\conv (\FEA(\SC))$ is a polyhedron. However, we know its convex hull only in an abstract form -- specifically, we only know its vertex representation. We derive valid inequalities for the feasible region by exploiting the structure of the problem via a lift-and-project like argument. The resulting LP relaxation of SC has the property that all extreme points of $\FEA(\SC)$ are extreme points of the LP relaxation. This is indicative of the relaxation being a close approximation of SC. We find this is indeed the case -- the LP relaxation implies several known converses. Moreover, it leads to new  converses.

\section{Optimization-based Formulation and LP relaxation}\label{sec:optimisiation}
This section presents the optimization based formulation of SC. We then derive the LP relaxation, discuss its properties and formulate and discuss the dual of the LP relaxation.
\subsection{A continuous optimization formulation for SC}
Consider a joint probability distribution $Q:\Zscr \rightarrow [0,1]$ factored as:
\begin{align}
Q(s,x,y,\shat)  \equiv P_S(s)Q_{X|S}(x|s)P_{Y|X}(y|x)Q_{\widehat{S}|Y}(\shat|y) \label{eq:Q} 
\end{align}where recall that $\Zscr:=\Sscr \times \Xscr \times \Yscr \times \Sscrhat$.
Any $Q$ that satisfies \eqref{eq:Q} is defined by $Q_{X|S}$ and $Q_{\widehat{S}|Y}$ lying in the space of conditional probability distributions $\mathcal{P}(\mathcal{X}|\mathcal{S})$ and $\mathcal{P}(\mathcal{\widehat{S}}|\mathcal{Y})$ respectively. The kernels $Q_{X|S}$ and $Q_{\Shat|Y}$ represent a \textit{randomized encoder} and \textit{randomized decoder} respectively and 
 together they constitute a `randomized code'. A randomized encoder (resp., a randomized decoder) is said to be \textit{deterministic} if there exists a function $f$ (resp., $g$) such that $Q_{X|S}(x|s)=\I{x=f(s)},$ for all $x,s$  (resp., $Q_{\Shat|Y}(\shat|y)=\I{\shat=g(y)},$ for all $\shat,y$). A deterministic encoder-decoder pair constitute a deterministic code. Recall that SC as posed in Section I is  an optimization problem over deterministic codes.

Now, consider the following optimization problem over joint probability distributions, 
 $$
\problemsmalla{$\SC{'}$}
	{Q,Q_{X|S},Q_{\widehat{S}|Y}}
	{\displaystyle \sum_{z}\kappa(z) Q(z)}
				 {\begin{array}{r@{\ }c@{\ }l}
				 Q(z)&\equiv &P_{S} Q_{X|S} P_{Y|X}  Q_{\widehat{S}|Y}(z),\\
				 				 \sum_{x} Q_{X|S}(x|s)&=& 1 \quad \forall s \in \mathcal{S},\\
\sum_{\shat} Q_{\widehat{S}|Y}(\shat|y)&=&1 \quad\forall y \in \mathcal{Y},\\
 Q_{X|S}(x|s) &\geq& 0\quad \forall s \in \mathcal{S},x \in \mathcal{X},\\
  Q_{\widehat{S}|Y}(\shat|y) &\geq& 0 \quad \forall \shat \in \mathcal{\widehat{S}},y \in \mathcal{Y},
	\end{array}}
	$$  which amounts to minimizing the same objective over randomized codes. We first note the set of extreme points of the feasible region of $\SC{'}$.\begin{theorem}\label{thm:extremepointsSC}
 The extreme points of the feasible region of $\SC{'}$ are given as,  \begin{align*}
&\ext(\FEA({\SC{'}}))=\left\{ (Q,Q_{X|S},Q_{\widehat{S}|Y}) \mid \exists \hspace{0.05cm}(f,g) \hspace{0.1cm} \mbox{such that} \hspace{0.1cm} \right. \\& \left.   Q_{X|S} \equiv \I{x=f(s)},Q_{\widehat{S}|Y} \equiv \I{\shat=g(y)}
  , Q \in \mathcal{P}(\Zscr) \right.\\&\left.\mbox{satisfies} \hspace{0.1cm}Q(z)\equiv P_S(s)Q_{X|S}(x|s)P_{Y|X}(y|x)Q_{\Shat|Y}(\shat|y)\right\}. \end{align*} 
 \end{theorem}The proof is included in Appendix~\ref{app:sec3}.
 
	By replacing $Q$ in the objective function of $\SC{'}$ with the RHS of first constraint, the resulting $\SC{'}$ has a bilinear objective function due to the presence of product terms $Q_{X|S}Q_{\Shat|Y}$ and a seperable feasible region given as $ \Pscr(\Xscr|\Sscr)\times \Pscr(\Sscrhat|\Yscr)$. Hence, $\SC{'}$ becomes equivalent to a seperably constrained bilinear programming problem. It is well known \cite[Exercise 4.25]{bazaraa06nonlinear} that such a problem admits an optimal solution at an extreme point of the feasible region. This implies that there exists an optimal solution of $\SC{'}$ at the extreme point of $ \Pscr(\Xscr|\Sscr)\times \Pscr(\Sscrhat|\Yscr)$, which in turn corresponds to a deterministic code. 
 Hence, the above optimization formulation $\SC{'}$  is equivalent to SC in the sense that they have the same optimal value.

We also note that in the setting of SC where $S,X,Y$ and $\Shat$ are discrete random variables taking values in finite spaces, 
there exist finitely many choices for functions $f: \Sscr \rightarrow \Xscr$ and $g:\Yscr \rightarrow \Sscrhat$. 
Consequently, SC is a discrete optimization problem. 
However, within the framework of $\SC{'}$, where optimization is done over probability distributions ($Q,Q_{X|S},Q_{\Shat|Y}$), we obtain a \textit{continuous optimization formulation} of SC. It is this continuous formulation which further along the way aids in the derivation of an LP relaxation. Since $\SC{'}$
is equivalent to SC, henceforth, we use SC to denote $\SC{'}$.

An important characteristic of $\SC$ is that $\FEA(\SC)$ is in fact nonconvex (see~\cite{kulkarni2014optimizer}). Our approach to lower-bounding $\SC$ will be via LP relaxation, which will be introduced in the following section. 
Presently, we first motivate the properties we desire from a relaxation. Consider the problem $\SC$ of lower bounding $\Ebb[d(S,\Shat)]$ where $d:\Sscr \times \Sscrhat \rightarrow \mathbb{R}$ is a distortion function. A classical approach \cite{ziv73functionals} to derive this lower bound is to employ the concept of $\fbf$-mutual information. Using the data processing inequality, this argument results in the following inequality, 
\begin{align}
R_{\fbf}(\dtilde) \leq C_{\fbf},\label{eq:fratedistoriton}
\end{align} 
where recall that 
 \begin{align}R_{\fbf}(\dtilde)=\min_{P_{\Shat|S}: \Ebb[d(S,
 \Shat)]\leq \dtilde}I_{\fbf}(S;\Shat), \label{eq:rf} 
 \end{align} where the minimum is over $P_{\Shat|S} \in \Pscr(\Sscrhat|\Sscr)$ and 
 \begin{align}C_{\fbf}=\max_{P_X}I_{\fbf}(X;Y), \label{eq:cf} 
 \end{align} where the maximum is over $P_X \in \Pscr(\Xscr)$. The $\mathbf{f}$-mutual information between discrete random variables $A \in \Ascr$, $B \in \mathcal{B}$ is defined as $$I_{\mathbf{f}}(A;B)=\sum_{a \in \Ascr}\sum_{b \in \mathcal{B}}P_{A,B}(a,b) \mathbf{f} \left( \frac{P_B(b)P_A(a)}{P_{A,B}(a,b)} \right),$$ where $\mathbf{f}: \mathbb{R} \rightarrow \mathbb{R}$ is any convex function such that $\mathbf{f}(1)=0$.
When $\fbf(x) \equiv -\log(x)$, $I_{\mathbf{f}}(A;B)=I(A;B),$ the mutual information between random variables $A$ and $B$. Since $R_{\fbf}(\dtilde)$ is a non-increasing function of $\dtilde$, a lower bound on $\dtilde$ follows from \eqref{eq:fratedistoriton}.

 
Observe that the above approach is equivalent to considering the following 
convex relaxation of SC with $\kappa(s,x,y,\shat)\equiv d(s,\shat)$,
$$
\problemsmalla{$ \rm{DPI_{\fbf}}$}
	{Q \in \Pscr(\Zscr)}
	{\displaystyle \Ebb_{Q} [ \kappa(S,X,Y\Shat)]}
				 {\begin{array}{r@{\ }c@{\ }l}
				 Q_S(s)&\equiv& P_S(s),\\
				 Q_{Y|X}(y|x) &\equiv & P_{Y|X}(y|x),\\
				 I_{\fbf}(Q_{S,\Shat}) &\leq & C_{\fbf},\\
	\end{array}}
	$$
	 where $I_{\fbf}(Q_{S,\Shat})$ is the $\fbf$-mutual information of $S,\Shat$ under the distribution $Q_{S,\Shat}$.
We explain the equivalence and convexity of ${\rm DPI_{\fbf}}$ in detail in Theorem~\ref{thm:fDPI} in Appendix~\ref{app:sec3}.
Convex analytically speaking,  this relaxation has a crucial shortcoming.
There are extreme points of $\FEA(\SC)$ which are not on the boundary of the relaxation $\FEA(\DPI_{\fbf})$, and thereby are not extreme points of $\FEA(\DPI_{\fbf})$. One can verify this through the following example.
\begin{examplec} Employing Theorem~\ref{thm:extremepointsSC}, consider an extreme point of $\FEA(\SC)$ given by the deterministic code $Q_{X|S}(x|s) \equiv \I{x=f(s)}$, where $f$ is an invertible function, and $Q_{\widehat{S}|Y}(\shat|y) \equiv \I{\shat= \shatbf}$ for some $\shatbf \in \mathcal{\widehat{S}}$ and where $Q(s,x,y,\shat)$ satisfies \eqref{eq:Q}. We see that $C_{\fbf} \geq I_{\mathbf{f}}(Q_{X,Y})= \sum_{x,y}P_{Y|X}(y|x)P_S(f^{-1}(x)) \mathbf{f} \left( \frac{\sum_x P_{Y|X}(y|x)P_S(f^{-1}(x))}{P_{Y|X}(y|x)}\right) > I_{\mathbf{f}}(Q_{S,\widehat{S}})=0$.  Thus, this point lies in the (relative) interior of $\FEA(\DPI_{\fbf})$ and cannot be an extreme point $\FEA(\DPI_{\fbf}).$ \end{examplec} 

Consequently, for problem SC, there exist loss functions $\kappa$ for which the convex relaxation ${\rm DPI}_{\fbf}$ is not tight. When $\kappa$ takes the form, $\kappa(s,x,y,\shat) \equiv d(s,\shat)$, the relaxation is tight only if $R_\fbf(\OPT(\SC)) =C_\fbf.$ When $\fbf(x) \equiv -\log(x)$, this corresponds to the rare, probablistically matched case~\cite{gastpar03tocode} wherein single-letter codes are optimal over arbitrary blocklengths. 

The above example highlights what one may ask for from a good relaxation. Recall from Section~\ref{sec:prelims} that for a problem like SC with a linear objective, the ideal relaxation of the nonconvex set $\FEA(\SC)$ is its convex hull. However, obtaining a half-space representation of the convex hull of nonconvex sets of the form of $\FEA(\SC)$ is still an open problem. A desirable property of a relaxation is that all extreme points of $\FEA(\SC)$ be retained as extreme points of the relaxation. 
In the next section we present an LP relaxation of SC with this property.

\subsection{Linear programming relaxation}
We apply the following lift-and-project-like idea (see Section~\ref{sec:prelims}) to derive the relaxation. We lift SC to a higher dimensional space by introducing new variables $W(s,x,y,\shat)$ to represent the product $Q_{X|S}(x|s)Q_{\Shat|Y}(\shat|y),$ for all $s,x,y,\shat$. Then, we obtain valid inequalities using these newly introduced variables $W(s,x,y,\shat)$ in this enlarged space.
To do so, for each $s \in \Sscr$, we multiply the constraint $\sum_x Q_{X|S}(x|s)=1,$ with the variables $Q_{\Shat|Y}(\shat|y)$ for all $\shat \in \Sscrhat,y\in \Yscr$, and likewise for each $y\in \Yscr$ we multiply the constraint $\sum_{\shat}Q_{\Shat|Y}(\shat|y)=1,$ by $Q_{X|S}(x|s)$, for all $x\in \Xscr,s\in \Sscr$. We further obtain additional constraints by multiplying the variable bounds with each other, \ie $(1-Q_{X|S}(x|s))(1-Q_{\Shat|Y}(\shat|y))\geq 0,$ for all $s,x,y,\shat$ and $Q_{X|S}(x|s)Q_{\Shat|Y}(\shat|y) \geq 0,$ for all $s,x,y,\shat$. Subsequently, we replace the bilinear product terms $Q_{X|S}(x|s)Q_{\Shat|Y}(\shat|y)$ in the constraints with $W(s,x,y,\shat)$ to obtain valid inequalities in the space of $(Q_{X|S},Q_{\Shat|Y},W)$. 
Clearly, these constraints are implied by the constraints of SC.  To obtain the LP relaxation, the constraint $W(s,x,y,\shat)=Q_{X|S}(x|s)Q_{\Shat|Y}(\shat|y)$ for all $s,x,y,\shat$ is dropped.
 


Thus, following is our LP relaxation.
$$ \begin{small}\problemsmalla{LP}
	{ Q_{X|S},Q_{\widehat{S}|Y},W}
	{\displaystyle \sum_{z}\kappa(z)P_{S}(s)P_{Y|X}(y|x)W(z)}
				 {\begin{array}{r@{\ }c@{\ }l}
				 				 \sum_{x} Q_{X|S}(x|s)&=& 1 \hspace{0.05cm} : \gamma^a(s) \qquad  \hspace{0.08cm} \forall s\\
\sum_{\shat} Q_{\Shat|Y}(\shat|y)&=&1 \hspace{0.05cm} :\gamma^b(y) \qquad \hspace{0.08cm}\forall y\\
 \sum_{x} W(z)-Q_{\widehat{S}|Y}(\shat|y)&=&0  \hspace{0.05cm} :\lambda^a(s,\shat,y)
  \hspace{0.1cm}  \forall s,\shat,y\\ 
 \sum_{\shat}W(z)-Q_{X|S}(x|s)&=&0  \hspace{0.05cm} :\lambda^b(x,s,y)\hspace{0.1cm} \forall x,s,y\\
Q_{X|S}(x|s)+Q_{\widehat{S}|Y}(\shat|y)-W(z)&\leq &1  \hspace{0.05cm} :\mu(z)\qquad \hspace{0.25cm} \forall z\\ Q_{X|S}(x|s) &\geq& 0\hspace{0.05cm} : \phi^a(x|s) \quad \hspace{0.15cm} \forall s,x\\
 Q_{\widehat{S}|Y}(\shat|y) &\geq& 0 \hspace{0.05cm} : \phi^b(\shat|y) \quad \hspace{0.2cm} \forall \shat,y\\
 W(z)&\geq& 0 \hspace{0.05cm}: \nu(z) \qquad  \hspace{0.3cm} \forall z.
	\end{array}} \end{small} $$ 
	Here $\gamma^a,\gamma^b,\lambda^a,\lambda^b,\mu,\phi^a,\phi^b$ and $\nu$ are Lagrange multipliers corresponding to the respective constraints. The following theorem proves that the feasible region of LP contains the feasible region of SC. Let the product $Q_{X|S}(x|s)Q_{\widehat{S}|Y}(\shat|y)$ be represented as $Q_{X|S} Q_{\widehat{S}|Y}(z)$ for all $z$. Similarly, let $P_SP_{Y|X}W(z)$ represent the product $P_S(s)P_{Y|X}(y|x)W(s,x,y,\shat) $ for all $z$.
\begin{figure}
\begin{center}
\includegraphics[scale=0.40, clip=true, trim = 0.55in 4.5in 0in 2.75in]{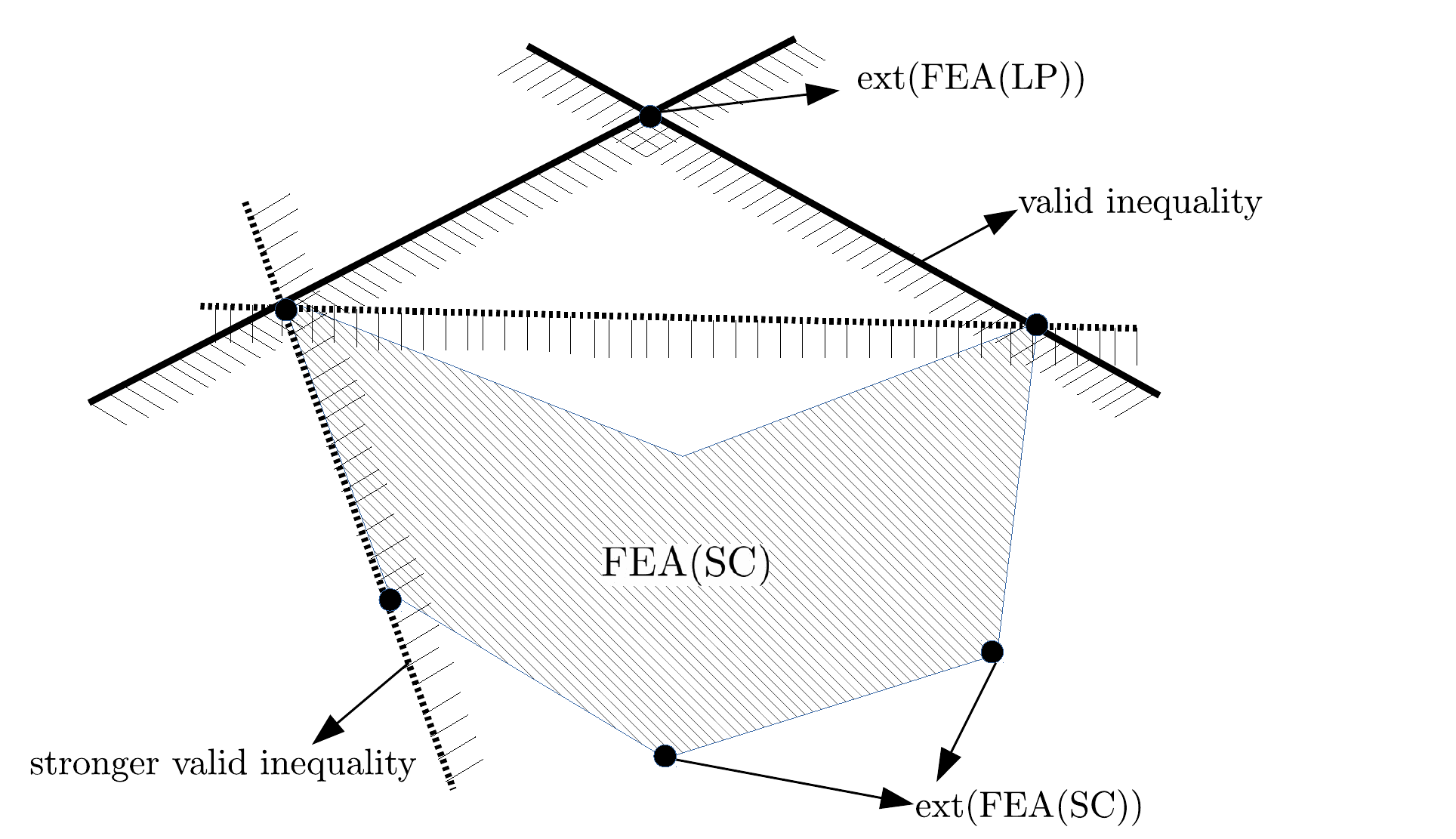} 
\caption{$\FEA(\SC)$ is depicted as a nonconvex set with finitely many extreme points. Also shown are valid inequalities that pass through some of these extreme points. \vspace{-.4cm}} 
\label{Fig3}
\end{center} 
\end{figure}

\begin{theorem} \label{thm:lpisrelaxation} 
LP is a convex relaxation of SC. \ie,
\begin{align*}
\FEA({\rm LP})&\supseteq  \biggl \lbrace (Q_{X|S},Q_{\widehat{S}|Y},W)\mid
(Q_{X|S},Q_{\widehat{S}|Y},Q) \in\\& \FEA({\SC}), W(z)\equiv Q_{X|S} Q_{\widehat{S}|Y}(z) \biggr \rbrace.
\end{align*}
\end{theorem}
The formal proof is in Appendix~\ref{app:sec3}.

\subsection{Extreme points of the LP relaxation}

We now discuss an 
%
  important property of the LP relaxation.
\begin{lemma}\label{lem:extremepointinclusion}
The extreme points of the feasible region of SC 
are included in the extreme points of the feasible region of LP. \ie,  
\begin{align*}
\ext(\FEA({\rm LP}))\supseteq \left\{ (Q_{X|S}^{*},Q_{\widehat{S}|Y}^{*},W^{*}) \mid (Q_{X|S}^{*},Q_{\widehat{S}|Y}^{*},Q^{*})  \right.\\  \in \ext(\FEA({\SC})), \left.  W^{*}(z)\equiv Q_{X|S}^{*} Q_{\widehat{S}|Y}^{*} (z)
 \right\}.
\end{align*}
\end{lemma} The proof is included in Appendix~\ref{app:sec3}.
\begin{figure}
\begin{center}
\includegraphics[scale=0.33,clip=true, trim = 1in 3in 1.7in 1.3in]{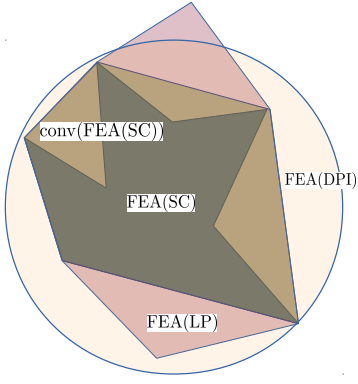} 
\caption{Comparison of $\rm{DPI}$  relaxation and LP relaxation: The 8-sided polygonal set represents the nonconvex $\FEA({\SC})$. This is contained in the convex hull of $\FEA({\SC})$. The convex hull is further contained in a polytope which represents the $\FEA({\rm LP})$. The outer elliptical set represents $\FEA({\rm DPI})$. 
}
\label{Fig2}
\end{center} 
\vspace{-.8cm}
\end{figure}

Valid inequalites and extreme points of $\FEA(\SC)$ are depicted in Fig~\ref{Fig3}. Fig \ref{Fig2} illustrates the LP relaxation in comparison with the convex hull of $\FEA({\SC})$.  As shown in the figure, the LP relaxation retains all the extreme points of $\FEA({\SC})$ in its set of extreme points. However, there may be additional extreme points of $\FEA(\rm LP)$ that are outside the convex hull of $\FEA(\SC)$. Also shown in the figure is the set $\FEA({\rm DPI}) :=\{Q \in \Pscr(\Zscr) | Q_S\equiv P_S, Q_{Y|X} \equiv P_{Y|X},I(Q_{S,\Shat}) \leq I(Q_{X,Y})\}$, where recall that $I(Q_{S,\Shat})$ is the mutual information of random variables $S$ and $\Shat$ under the distribution $Q_{S,\Shat}$ (and similarly $I(Q_{X,Y})$). This is clearly a convex set that contains $\FEA(\SC).$ However, as shown in the figure, there may be extreme points of $\FEA(\SC)$ that are not extreme points of $\FEA({\rm DPI})$. 

A natural question is whether $ \FEA({\rm LP}) \subseteq \FEA({\rm DPI})$. 
We do not have precise answers to this as yet.
However notice that the LP relaxation obtained could be further tightened by incorporating the data processing inequality into the LP relaxation. This yields a relaxation which is strictly tighter than $\FEA(\rm{DPI})$ (though it is no more a LP relaxation).

Consider the following reduced version of LP obtained by removing the inequality, $-1+Q_{X|S}(x|s)+Q_{\Shat|Y}(\shat|y)-W(x,s,\shat,y) \leq 0$, for all $s,x,y,\shat$ from the problem LP. 
$$ \problemsmalla{$\LP'$}
	{ Q_{X|S},Q_{\Shat|Y},W}
	{\displaystyle \sum_{z}\kappa(z)P_{S}(s)P_{Y|X}(y|x)W(z)}
				 {\begin{array}{r@{\ }c@{\ }l}
				 				 \sum_{x} Q_{X|S}(x|s)&=& 1  \qquad   \forall s\\
\sum_{\shat} Q_{\Shat|Y}(\shat|y)&=&1  \qquad \forall y\\
 \sum_{x} W(z)-Q_{\Shat|Y}(\shat|y)&=&0 
  \qquad  \forall s,\shat,y\\ 
 \sum_{\shat}W(z)-Q_{X|S}(x|s)&=&0\qquad \forall x,s,y\\
 Q_{X|S}(x|s) &\geq& 0 \qquad  \forall s,x\\
 Q_{\Shat|Y}(\shat|y) &\geq& 0  \qquad  \forall \shat,y\\
 W(z)&\geq& 0  \qquad   \forall z.
	\end{array}} $$
$\LP'$ is simpler than $\LP$, but also enjoys the property that all extreme points of $\FEA(\SC)$ are extreme points of $\FEA(\LP')$.
\begin{lemma}
Consider the linear program $\LP'$. $\LP'$ is a relaxation of SC such that $\FEA(\LP) \subseteq \FEA(\LP')$ and the extreme points of the feasible region of SC are included in the set of extreme points of $\LP'$, \ie, \begin{align*}
\ext(\FEA(\LP'))\supseteq \left\{ (Q_{X|S}^{*},Q_{\Shat|Y}^{*},W^{*}) \mid (Q_{X|S}^{*},Q_{\Shat|Y}^{*},Q^{*})  \right.\\  \in \ext(\FEA(\SC)), \left.  W^{*}(z)\equiv Q_{X|S}^{*} Q_{\Shat|Y}^{*} (z)
 \right\}.
\end{align*}
\end{lemma}The proof is similar to the proof of Lemma~\ref{lem:extremepointinclusion} and we skip the proof here.
\begin{figure*}[!t]
\normalsize
\begin{equation}
{\rm DP}\quad \qquad \max_{\gamma^a,\gamma^b,\lambda^a,\lambda^b,\mu} \quad \qquad \sum_{s \in \Sscr}\gamma^a(s)+\sum_{y \in \Yscr} \gamma^b(y)-\sum_{s,x,y,\shat}\mu(s,x,y,\shat)\nonumber
\end{equation}
\begin{align}\mbox{subject to}  \qquad  \gamma^a(s)- \sum_y \lambda^b(x,s,y)-\sum_{\shat,y}\mu(s,x,y,\shat) &\leq 0 \qquad \quad \quad \hspace{0.3cm} \qquad \quad \qquad \quad \quad \hspace{0.1cm} \forall x,s \tag{D1} \\
\gamma^b(y)- \sum_s \lambda^a(s,\shat,y)-\sum_{x,s}\mu(s,x,y,\shat) &\leq 0 \qquad \quad \quad \quad \hspace{0.3cm} \hspace{0.1cm} \qquad \quad \qquad \quad \forall \shat,y \tag{D2} \\
\lambda^a(s,\shat,y)+\lambda^b(x,s,y)+ \mu(s,x,y,\shat) &\leq \kappa(s,x,y,\shat) P_S(s) P_{Y|X}(y|x)\hspace{0.3cm}\forall s,x,y,\shat \tag{D3} \\
\mu(s,x,y,\shat) &\geq 0 \qquad \qquad \quad \qquad \quad \qquad \quad \quad  \forall s,x,y,\shat \nonumber
\end{align}
\hrulefill
\vspace*{4pt}
\end{figure*}
\subsection{Relation to existing approaches to obtain relaxations}
We now discuss the relation of our LP relaxation to  already known approaches for obtaining an LP relaxation.
\subsubsection{McCormick inequalities}
Recall from Section~\ref{sec:prelims} that McCormick inequalities provide the convex hull of a set  $K \triangleq \{ (w,x_1,x_2) \in \Real \times [l_1 ,u_1]\times [l_2 ,u_2]  \mid w=x_1x_2\}$ by means of convex under-estimating and concave over-estimating inequalities. 
  However, note that the constraints on $x_1$ and $x_2$ (\ie the bounds $[l_1,u_1],[l_2,u_2]$) are not coupled. As such, McCormick inequalities provide the convex hull of simple bilinear product terms which are uncoupled in their constraints. However, the bilinear products in SC are coupled in their constraints. For each $s$, $Q_{X|S}(x|s)$ for all $x$ are coupled through the constraint $\sum_{x}Q_{X|S}(x|s)=1$. Similarly, for each $y$,  $Q_{\Shat|Y}(\shat|y)$ for all $\shat$ are linked through the constraint $\sum_{\shat}Q_{\Shat|Y}(\shat|y)=1,$ for all $y$. 

One could also arrive at the LP relaxation using McCormick inequalities as follows. Employing the convex underestimating inequalities in SC with $W(z) \equiv Q_{X|S}Q_{\widehat{S}|Y}(z)$, $0 \leq Q_{X|S}(x|s) \leq 1,$ for all $x,s$ and $0 \leq Q_{\Shat|Y}(\shat|y) \leq  1,$ for all $\shat,y$, we get the constraints $
Q_{X|S}(x|s)+Q_{\widehat{S}|Y}(\shat|y)-W(z) \leq 1, $ and  $W(z) \geq 0,$ for all $z$.
 However, this leaves an obvious gap: the equations $\sum_x Q_{X|S}(x|s) \equiv 1$, $\sum_{\shat}Q_{\widehat{S}|Y}(\shat|y) \equiv 1,\sum_s W(s,x,y,\shat) \equiv Q_{\Shat|Y}(\shat|y),\sum_{\shat}W(s,x,y,\shat) \equiv Q_{X|S}(x|s)$ (the last two equations in fact imply the concave overestimating McCormick inequalities) must hold for LP but are not implied by McCormick inequalities. Hence, these additional constraints have to be included along with the the convex under-estimating inequalities to arrive at our relaxation. 

\subsubsection{The Reformulation-linearization technique (RLT) \cite{sherali91bilinear}}
The reformulation-linearization technique proposed by Sherali and Alameddine can also be employed to arrive at our LP relaxation.
RLT is a linear programming relaxation approach for bilinear programming problems with a bilinear objective function subjected to linear constraints. However, applying RLT to SC results in a linear program with more number of variables and  constraints than our LP relaxation. For example, in the reformulation phase, each equality constraint is multiplied with each of the variables, resulting in products of the form of $Q_{X|S}(x|s)Q_{X|S}(\bar{x}|\bar{s})$, $x,\bar{x} \in \Xscr, s,\bar{s} \in \Sscr$ and $Q_{\Shat|Y}(\shat|y)Q_{\Shat|Y}(\shat'|y')$, $\shat,\shat' \in \Sscrhat, y,y' \in \Yscr$. Similar products also arise when inequalities are multiplied with each other, say, $(1-Q_{X|S}(x|s))(1-Q_{X|S}(\bar{x}|\bar{s}))$. In the linearization phase, products of this form are replaced by new variables, which lead to additional variables.  Dropping these additional constraints leads to our LP relaxation. By selectively multiplying the constraints with certain variables, we limit the number of newly introduced variables, thereby making the LP relaxation more tractable.

\subsection{Dual program and lower bounds}\label{subsec:DP}
It is easy to see that the dual of problem LP  is the problem DP, given via Lagrange multipliers of LP. Henceforth, we call these Lagrange multipliers as the `dual variables'. Notice that these variables are restricted in their domain and in the case of $\mu$, also their range. Specifically,  $\gamma^a:\Sscr \rightarrow \mathbb{R}$, $\gamma^b:\Yscr \rightarrow \mathbb{R}$, $\lambda^b: \Sscr \times \Xscr \times \Yscr\rightarrow \mathbb{R}$, $\lambda^a:\Sscr \times \Sscrhat \times \Yscr \rightarrow \mathbb{R}$ and $\mu: \Zscr \rightarrow \mathbb{R}_{+}$, whereby these variables are functions on \textit{subspaces} of $\Zscr.$    Notice that the dual of $\LP'$, denoted $\DP'$, is same as DP with the dual variable $\mu(s,x,y,\shat) $ set identically to be $0$.

From the discussion in Section~\ref{sec:prelims}, the following lemma formalizes our framework for obtaining lower bounds on SC.
\begin{lemma}\label{lem:dualfeasiblebound}
The objective value of any feasible point of $\DP$ is a lower bound on the optimal value of SC, \ie, if $\left(\gamma^a(s),\gamma^b(y),\lambda^a(s,\shat,y),\lambda^b(x,s,y),\mu(s,x,y,\shat)\right)_{z\in \Zscr}$ is feasible for DP, then
\begin{align}
\hspace{-.3cm}\OPT({\SC})& \geq \OPT({\rm LP})= \OPT({\rm DP})\nonumber\\
&\geq \sum_s\gamma^a(s)+\sum_y \gamma^b(y)-\sum_{s,x,y,\shat}\mu(s,x,y,\shat).\label{eq:dualfeasiblebound}
\end{align} 
\end{lemma}
\begin{proof} It follows from the constraints of LP that $\rm{FEA}(\rm LP)$ is a bounded nonempty polyhedron, and hence $\OPT(\LP)$ is finite. The lower bound in the RHS of \eqref{eq:dualfeasiblebound} then follows from the strong duality of linear programming (Theorem~\ref{thm:strongdual}) and the fact that LP is a convex relaxation of SC (Theorem~\ref{thm:lpisrelaxation}).
\end{proof}
Consequently, DP may itself be considered as an abstract converse.

\begin{remarkc}[Interpretation of DP:] LP relaxations for combinatorial channel coding problems, such as those in~\cite{kulkarni12nonasymptotic} and \cite{fazeli2015generalized}, can be interpreted as \textit{generalized} sphere-packing. The duals of these relaxations represent  \textit{generalized covering}. Since our LP relaxation has not made any use of the structure of the loss function $\kappa,$ and it was derived from purely algebraic arguments, there does not seem to be any elegant interpretation for LP or DP. However, note that when $\mu\equiv 0$, the objective function of DP seems to reflect a tension between $\gamma^a$, which is a function of $S$, and $\gamma^b$, which is a function of $Y$, and such that sum of $\sum_s \gamma^a(s)$ and $\sum_y \gamma^b(y)$ is restricted via (D3). This suggests that the rate $|\Sscr|/|\Yscr|$ will play a role in determining $\OPT(\DP)$. For channel coding, this is indeed the case, as we shall see in Section~\ref{lowerbounds}, where DP will be shown to imply the channel coding strong converse.
\end{remarkc}

We end by noting that $\DP$ always yields a nontrivial bound on $\SC$ when considering natural problems from joint source-channel coding.
\begin{proposition}\label{prop:positiveoptimal}
Consider problem SC with $\Sscr=\Sscrhat$, $P_S(\cdot)>0$, $P_{Y|X}(\cdot|\cdot)>0$ and $\kappa(s,x,y,\shat) \equiv d(s,\shat)$ where $d:\Sscr \times \Sscrhat \rightarrow \Real$ is such that $d(s,\shat) \geq 0$ for all $s \in \Sscr$, $\shat \in \Sscrhat$ and $d(s,\shat)=0$ if and only if $s=\shat$. Then, the optimal value of DP is strictly positive, \ie, 
 \[\OPT(\DP) >0.\]
\end{proposition}The proof is included in Appendix~\ref{app:sec3}.
\section{Linear Programming Based Finite Blocklength Converses: An Illustrative Example}\label{sec:channelBSC}
In this section, we illustrate the selection of variables of the DP to derive converses through the example of finite blocklength channel coding of a Binary Symmetric Channel (BSC). 
  Through a simple selection of values of dual variables, we first obtain a `naive converse' which gives a lower bound of $-\infty$ in the limit of large blocklengths for rates slightly greater than the capacity of the BSC. We then illustrate how these dual variables are appropriately selected such that they yield a lower bound which implies the strong converse. In fact, a similar line of construction of feasible points of DP results in new and improved converses for lossy joint source-channel coding problem as explained in Section~\ref{lowerbounds}.

Consider a discrete memoryless BSC $(\Xscr,P_{Y|X},\Yscr)$ with $\Xscr=\Yscr=\{0,1\}^n, n \in \Nbb$ and channel conditional probability given as,
 \begin{align}P_{Y|X}(y|x)&=\prod_{i=1}^nP_{Y_i|X_i}(y_i|x_i),\quad \mbox{where}\non\\ 
P_{Y_i|X_i}(y_i|x_i)&=\epsilon \Ibb\{x_i\neq y_i\}+(1-\epsilon)\Ibb\{x_i=y_i\} \label{eq:BSC},
\end{align} where $x=(x_1,\hdots,x_n) \in \Xscr, y=(y_1,\hdots,y_n)\in \Yscr$ and $\epsilon <0.5$. If $d_{x,y}$ represents the Hamming distance between $x \in \Xscr$ and $y \in \Yscr$, then the channel conditional probability can be equivalently expressed as, \begin{align}P_{Y|X}(y|x)\equiv \epsilon^{d_{x,y}}(1-\epsilon)^{n-d_{x,y}}.\label{eq:BSC1} \end{align} We show that there exists a feasible solution of DP which implies the strong converse for BSC. To motivate the construction of such a feasible solution of DP, we first consider the following ``naive'' converse.
\begin{lemma}[A naive converse for the BSC]\label{thm:BSCconverse1}
Consider problem SC with $\Sscr=\Sscrhat=\{1,\hdots,M\}$, $P_S(s) \equiv \frac{1}{M}$ and $\Xscr=\Yscr=\{0,1\}^n$. Let the channel be the discrete memoryless BSC as given in \eqref{eq:BSC} with $\epsilon<0.5$. Then, for any code, the following is a lower bound on the probability of error,
\begin{align}
\Ebb[\Ibb\{S\neq \Shat\}]& \geq \OPT(\SC)\geq \OPT(\LP)=\OPT(\DP)\non\\& \geq 1-(1-\epsilon)^n\frac{2^n}{M}.\label{eq:BSCbound1}\end{align}
\end{lemma}
\begin{proof}
We construct a feasible solution of DP that gives the required bound. To do so, consider $\mu(s,x,y,\shat)\equiv 0$. Since $\lambda^b$ is a function of $x,y$ and $s$, we take $$\lambda^b(x,s,y)\equiv P_S(s)P_{Y|X}(y|x) \equiv \frac{1}{M} \epsilon^{d_{x,y}}(1-\epsilon)^{n-d_{x,y}},$$ such that $\lambda^b$ is a function of $d_{x,y}$.
We now choose $\lambda^a(s,\shat,y)$ such that (D3) is satisfied. The above choice of $\lambda^b$ implies that $$\lambda^a(s,\shat,y)\hspace{-0.1cm} \leq \frac{1}{M}\epsilon^{d_{x,y}}(1-\epsilon)^{n-d_{x,y}}\biggl[\Ibb\{s \neq \shat\}-1\biggr],$$ for all $s,x,y,\shat$. Thus, we choose $\lambda^a(s,\shat,y) =$
 \begin{align*}\hspace{-0.05cm}\min_x \frac{-1}{M}\epsilon^{d_{x,y}}(1-\epsilon)^{n-d_{x,y}}\Ibb\{s =\shat\}\hspace{-0.1cm}=\frac{-(1-\epsilon)^n\Ibb\{s =\shat\}}{M},\end{align*} for all $s,\shat,y$.
   We consider the following values for the remaining dual variables,
\begin{align*}
\gamma^a(s) &\equiv \min_x \sum_y \hspace{-0.05cm}\lambda^b(x,s,y)\hspace{-0.1cm}\stackrel{(a)}{=}\min_x \sum_{k=0}^n \hspace{-0.095cm}\Comb{n}{k}\frac{\epsilon^{k}(1-\epsilon)^{n-k}}{M} \hspace{-0.07cm}
=\hspace{-0.05cm}\frac{1}{M},\non\\
\gamma^b(y) &\equiv \min_{\shat}\sum_s \lambda^a(s,\shat,y)\equiv \frac{-(1-\epsilon)^n}{M},
\end{align*} which ensures that (D1) and (D2) are satisfied. The equality in (a) results since for any $x\in \Xscr$, the number of $y$'s in $\Yscr$ such that $d_{x,y}=k$ is equal to $\Comb{n}{k}$.  Thus, we have a feasible set of variables of DP. 

With the considered choice of dual variables, the dual cost evaluates to\begin{align*}
\sum_s \gamma^a(s)+\sum_y \gamma^b(y)
&=1-(1-\epsilon)^n \frac{2^n}{M}.
\end{align*} This dual cost together with \eqref{eq:dualfeasiblebound} gives the required bound.
\end{proof}
We now analyze if the lower bound in \eqref{eq:BSCbound1} implies the channel coding converse.
 Denote the capacity of the BSC as $\mathcal{C}_{BSC}=1-H_2(\epsilon)$, where $H_2(\epsilon)=-\epsilon \log_2(\epsilon)-(1-\epsilon)\log_2(1-\epsilon)$.
Notice that if $M=2^{nR}$, where $R=\mathcal{C}_{BSC}+\alpha \epsilon$, for $\alpha>0$  bound \eqref{eq:BSCbound1} becomes, 
\begin{align}
1-(1-\epsilon)^n \frac{2^n}{M}=1-\frac{2^n(1-\epsilon)^n}{ 2^{n\mathcal{C}_{BSC}+n\alpha \epsilon}}=1-\biggl(\frac{1-\epsilon}{2^{\alpha}\epsilon}\biggr)^{n\epsilon},\hspace{-0.2cm}\label{eq:BSCbound3}
\end{align}which goes to $-\infty$ as $n$ increases to $\infty$ if $\alpha$ is small (justifying why we call it a naive converse). Consequently, the lower bound in \eqref{eq:BSCbound1} does not imply the converse of channel coding.
To get a lower bound which also implies the strong converse from \eqref{eq:BSCbound3}, we modify the construction of dual variables. One way to accomplish this is by incorporating the additional term $\bigl(\frac{\epsilon}{1-\epsilon}\bigr)^{n\epsilon-n\delta}$ in $\lambda^a(s,\shat,y)$, \ie
   \begin{align}\lambda^a(s,\shat,y)\equiv -\biggl(\frac{\epsilon}{1-\epsilon}\biggr)^{n\epsilon-n\delta} \frac{(1-\epsilon)^n}{M}\Ibb\{s=\shat\}, \label{eq:newlambdaa}\end{align} where $\delta\in (0,\epsilon)$ is chosen suitably later.
    Modification in $\lambda^a(s,\shat,y)$ implies $\lambda^b(x,s,y)$ has to be modified such that (D3) is satisfied. For (D3) to hold, we need,
$$\lambda^b(x,s,y)\leq \frac{(1-\epsilon)^n}{M}\min \biggl[\frac{\epsilon^{d_{x,y}}}{(1-\epsilon)^{d_{x,y}}} ,\frac{\epsilon^{n\epsilon-n\delta}}{(1-\epsilon)^{n\epsilon-n\delta}} \biggr],$$  for all $s,x,y$.
Thus, we choose $\lambda^b(x,s,y)$ such that the above inequality holds with equality. The RHS of the above inequality can be expressed as,
\begin{align}&\lambda^b(x,s,y)\equiv\frac{(1-\epsilon)^n}{M}\Bigg[\left (\frac{\epsilon}{1-\epsilon}\right )^{d_{x,y}}\Ibb\{d_{x,y} > n\epsilon-n\delta\}\non\\&\qquad+\left (\frac{\epsilon}{1-\epsilon}\right )^{n\epsilon-n\delta}\Ibb\{d_{x,y} \leq n\epsilon-n\delta\}\Bigg]. \label{eq:newlambdab}
\end{align}

In the following theorem, we employ the above values of $\lambda^a$  and $\lambda^b$  (in \eqref{eq:newlambdaa} and \eqref{eq:newlambdab}) to obtain a lower bound which implies the strong converse for BSC.
\begin{theorem}\label{thm:BSCstrongconverse}
Consider problem SC with $\Sscr=\Sscrhat=\{1,\hdots,M\}$, $P_S(s) \equiv \frac{1}{M}$ and $\Xscr=\Yscr=\{0,1\}^n$. Let the channel be the discrete memoryless BSC as given in \eqref{eq:BSC} with $\epsilon<0.5$. Then, for any code, the following is a lower bound on the probability of error,
\begin{align}
&\Ebb[\Ibb\{S\neq \Shat\}]\geq \OPT(\SC)\geq \OPT(\LP)=\OPT(\DP)\non\\&\geq \sup_{\epsilon > \delta>0} 
\biggl[1+2^{-n\bigl(H_2(\epsilon)-H_2(\epsilon-\delta)-\delta \log_2 \frac{1-\epsilon}{\epsilon}+\frac{1}{n}\log_2 l(n,\epsilon-\delta)\bigr)}\non \\& -\frac{1}{M}2^{n\bigl(1-H_2(\epsilon)+\delta \log_2 \frac{1-\epsilon}{\epsilon}\bigr)}-(1-\epsilon)^{n-nH_{\frac{1}{1-\epsilon}}(\epsilon-\delta)}\biggr],\hspace{-0.2cm}\label{eq:BSCbound2}
\end{align}
\begin{align} &\mbox{where}\hspace{0.2cm}l(n,\alpha)=\frac{\exp\bigl(\lambda_1(n)-\lambda_2(n\alpha)-\lambda_2(n(1-\alpha))\bigr)}{\sqrt{2\pi \alpha(1-\alpha)n}}\non\\& \mbox{with}\quad\lambda_1(x)=\frac{1}{12x+1}, \quad \lambda_2(x)=\frac{1}{12x}.\label{eq:ln}\end{align}
\end{theorem}
\begin{proof}
Let $0<\delta<\epsilon$. To get the required bound, we take $\mu(s,x,y,\shat)\equiv 0$, $\lambda^b(x,s,y)$ as in \eqref{eq:newlambdab} and $\lambda^a(s,\shat,y)$ as in \eqref{eq:newlambdaa}. It is clear from our discussion above that the choice of values of $\lambda^a$ and $\lambda^b$ are feasible with respect to (D3). By setting $\gamma^a$ and $\gamma^b$ as,
\begin{align*}
\gamma^b(y)\equiv \min_{\shat} \sum_s \lambda^a(s,\shat,y), \quad 
\gamma^a(s) \equiv \min_x \sum_y \lambda^b(x,s,y),
\end{align*} constraints (D1) and (D2) of DP are also satisfied.

Now, for any $s \in \Sscr$,
\begin{align}
\gamma^a(s)&\stackrel{(a)}{=}\frac{1}{M}\biggl[\sum_{k=n\epsilon-n\delta+1}^n\Comb{n}{k}\epsilon^k(1-\epsilon)^{n-k} \non\\&\qquad +\epsilon^{n\epsilon-n\delta}(1-\epsilon)^{n-n\epsilon+n\delta}\sum_{k=0}^{n\epsilon-n\delta}\Comb{n}{k}\biggr]\label{eq:gammaforBSC} \\
&\stackrel{(b)}{\geq}\frac{1}{M}\biggl[1-(1-\epsilon)^{n-nH_{\frac{1}{1-\epsilon}}(\epsilon-\delta)}\non\\
+2&^{-n\bigl(H_2(\epsilon)-H_2(\epsilon-\delta)-\delta \log_2 \frac{1-\epsilon}{\epsilon}+\frac{1}{n}\log_2 l(n,\epsilon-\delta)\bigr)}\biggr],
\label{eq:BSCgamma}
\end{align}
where $H_q(x)=x\log_q(q-1)-x\log_q(x)-(1-x)\log_q(1-x)$. The equality in ($a$) arises as for any $x \in \Xscr$, the number of $y$'s in $\Yscr$ such that $d_{x,y}=k$ is equal to $\Comb{n}{k}$. 
To get to inequality ($b$), we upper bound $\sum_{k=0}^{n\epsilon-n\delta} \Comb{n}{k}\epsilon^k(1-\epsilon)^{n-k}$ using that
 \begin{align}\sum_{i=0}^{n\alpha}\Comb{n}{i}(q-1)^i\leq q^{nH_q(\alpha)}\hspace{0.1cm} \mbox{for} \hspace{0.1cm}q>\frac{1}{1-\alpha}, \alpha \in (0,1),\label{eq:upperbound}
 \end{align} 
and we lower bound $\sum_{k=0}^{n\epsilon-n\delta}\Comb{n}{k}$ using \begin{align}\sum_{i=0}^{n\alpha}\Comb{n}{i}(q-1)^i \geq q^{nH_q(\alpha)}l(n,\alpha),
\label{eq:lowerbound}\end{align} where $l(n,\alpha)$ is as defined in \eqref{eq:ln}.
Moreover, for any $y \in \Yscr$,
\begin{align}
\gamma^b(y)& =\frac{-1}{M}\epsilon^{n\epsilon-n\delta}(1-\epsilon)^{n-n\epsilon+n\delta}=\frac {-1}{M}2^{-nH_2(\epsilon)+n\delta \log_2 \frac{1-\epsilon}{\epsilon}}. \label{eq:gammabforBSC}
\end{align} 
Consequently, the dual cost $\sum_s \gamma^a(s)+\sum_y \gamma^b(y)$ evaluates to the term in the bracket in \eqref{eq:BSCbound2}.
Taking the supremum over $\delta \in (0,\epsilon)$ and then applying 
\eqref{eq:dualfeasiblebound} gives the required bound.
\end{proof}

We now show that \eqref{eq:BSCbound2} implies the strong converse for the BSC.
\begin{corollary}[Strong Converse for the BSC]
Consider the problem SC with $\Sscr=\Sscrhat=\{1,\hdots,M\}$, $P_S(s) \equiv \frac{1}{M}$ and $\Xscr=\Yscr=\{0,1\}^n$. Let the channel be the discrete memoryless BSC as given in \eqref{eq:BSC} with $\epsilon<0.5$. If $M=2^{nR}$,  where $R>\mathcal{C}_{BSC}=1-H_2(\epsilon)$, then the lower bound in \eqref{eq:BSCbound2} implies that $$\lim_{n \rightarrow \infty} \Ebb[\Ibb\{S\neq \Shat\}]=1.$$
\end{corollary}
\begin{proof}
In \eqref{eq:BSCbound2}, we fix $\delta$ such that $0<\delta<\min\biggl(\epsilon,\frac{R-\mathcal{C}_{BSC}}{\log_2 \bigl(\frac{1-\epsilon}{\epsilon}\bigr)}\biggr)$.
Further, $H_2(x)$ being concave, $H_2(\epsilon)-H_2(\epsilon-\delta)-\delta \log_2 \frac{1-\epsilon}{\epsilon}\geq 0$. Also, $H_{\frac{1}{1-\epsilon}}(\epsilon-\delta)<1$  and $1-H_2(\epsilon)-R+\delta \log_2 \bigl(\frac{1-\epsilon}{\epsilon}\bigr)<0$.
Consequently, as $n \rightarrow \infty$, RHS of \eqref{eq:BSCbound2} goes to 1, thereby implying the strong converse. 
\end{proof}
In a similar line of construction of dual variables, a strong converse for the finite blocklength channel coding of a discrete memoryless binary erasure channel is derived in \cite{jose2017linear}.\\
\begin{remarkc}[Selection of Dual Variables:]
It now becomes evident that the selection of values of dual variables $\lambda^a(s,\shat,y)$ and $\lambda^b(x,s,y)$ plays a crucial role in the quality of the converse. Notice that in \eqref{eq:newlambdab}, $\lambda^b(x,s,y)$ has been modified such that it takes $P_S(s)P_{Y|X}(y|x)$ when $d_{x,y}>n\epsilon-n\delta$ and when $d_{x,y}\leq n\epsilon-n\delta$, it takes $P_S(s)P_{Y|X}(y|x)$ where $P_{Y|X}(y|x)$ is evaluated at ${d_{x,y}=n\epsilon-n\delta}$.
 Further, in \eqref{eq:newlambdaa}, $\lambda^a(s,\shat,y)$ is chosen such that it takes a non-zero value when the compliment of the cost function (here, cost function $\kappa
 (z)\equiv \Ibb\{s\neq\shat\}$) is true and the corresponding value is in fact the negative of the term in $\lambda^b(x,s,y)$ corresponding to $\Ibb\{d_{x,y}\leq n\epsilon-n\delta\}$. Also, notice that $\Ibb\{d_{x,y}>n\epsilon-n\delta\}=\Ibb\{P_{Y|X}(y|x)< \epsilon^{n\epsilon-n\delta}(1-\epsilon)^{n-n\epsilon+n\delta}\}$. These observations can be  extended to derive new converses for finite blocklength joint source-channel coding problems as explained in the next section. 
\end{remarkc}\\
\section{Lower Bounds on Finite Blocklength Joint Source - Channel Coding Problems}\label{lowerbounds}
In this section, by a logical extension of the construction of dual variables from the previous section, we derive lower bounds for various instances of the finite blocklength lossy joint source-channel coding problem, thereby making the case that LP relaxation and duality serve as a common framework from which converses for various cases of joint source-channel coding can be derived. 
We first consider the problem of obtaining a lower bound on the \textit{minimum excess distortion probability} of a finite blocklength lossy joint source-channel code and then take up the lossy source coding and channel coding problems.
 
For the lossy joint source-channel coding problem, we consider problem SC with $S$ having the distribution $P_S$ and channel conditional probability distribution $P_{Y|X}$. The cost function is given as $\kappa(s,x,y,\shat) \equiv \I{d(s,\shat)>\dbf}$, where $d: \Sscr \times \Sscrhat \rightarrow [0,+\infty]$ represents the distortion function and $\dbf \in [0, \infty)$ is the distortion level. The objective is to obtain a lower bound on the minimum value of $\Ebb[\I{d(S,\Shat)>\dbf}]=\Pbb[d(S,\Shat) >\dbf]$ (which is called the excess distortion probability) achieved by a joint source-channel code $(f,g)$. We will use DP to derive a lower bound on this problem.

 Kostina and Verd\'{u} in \cite{kostina2013lossy} obtained general converses on the minimum excess distortion probability achieved by a finite blocklength joint source-channel code. The converse for lossy source coding \cite{kostina2012fixed} and the converse for channel coding without cost constraints proposed by Wolfowitz \cite{wolfowitz1968notes} follow as a particular case of the converse for joint source-channel coding. Further, it has been shown that these finite blocklength converses for channel coding, lossy source coding and joint source-channel coding with excess distortion probability as the loss criterion are asymptotically tight.
 
In this section, by constructing a feasible point of the dual program DP, we recover the converse of Kostina and Verd\'{u} \cite[Theorem 3]{kostina2013lossy} on the minimum excess distortion probability achieved by a finite blocklength lossy joint source-channel code. In fact, by tweaking this construction of feasible point, we derive a new converse which \textit{improves} on the Kostina-Verd\'{u} converse. 
We then derive another converse which \textit{further improves} on the Kostina-Verd\'{u} converse. 
%
For lossy source coding and channel coding without cost constraints, new converses which improve on the converses of Kostina and Verd\'{u} \cite[Theorem 7]{kostina2012fixed} and  Wolfowitz \cite{wolfowitz1968notes} respectively, follow from these new results. It thus follows that our LP relaxation is asymptotically tight for channel coding, lossy source coding and lossy joint source-channel coding with probability of excess distortion as the loss criterion.

Kostina and Verd\'{u} leverage the concept of $\dbf$-tilted information for  deriving the converse for joint source-channel coding \cite{kostina2013lossy}. 
For a source $S$ with distribution $P_S$, distortion function $d:\Sscr \times \Sscrhat \rightarrow [0,+\infty]$ and distortion level $\dbf$, the rate-distortion function is defined as \begin{align}R_S(\dbf)=\inf_{P_{\Shat|S}:\Ebb[d(S,\Shat)]\leq \dbf}I(S;\Shat),\label{eq:ratedistortion}
\end{align}where the infimum is over $P_{\Shat|S}\in \Pscr(\Sscrhat|\Sscr)$. Assume that the infimum in \eqref{eq:ratedistortion} is achieved by a unique $P_{\Shat^{*}|S}$ and $\dbf_{\rm min}=\inf\{\dbf:R_S(\dbf)< \infty\}$. 
\begin{definition}[$\dbf$-tilted information \cite{kostina2013lossy}]
For $\dbf>\dbf_{\rm min}$,
  the $\dbf$-tilted information in $S$ is defined as \begin{align}
j_S(s,\dbf)=\log \frac{1}{\Ebb[\exp(\lambda^{*}\dbf-\lambda^{*}d(s,\Shat^{*}))]},\label{eq:dtilted}
\end{align}where the expectation is with respect to the unconditional probability distribution $P_{\Shat^{*}}$ on $\Sscrhat$ which achieves the infimum in \eqref{eq:ratedistortion}. When $\dbf=0$, the $0$-tilted information is defined as
\[j_{S}(s,0)=i_S(s), \quad \mbox{ where} \quad i_S(s)=\log \frac{1}{P_S(s)}.\] \end{definition} (We refer the readers to \cite{kostina2013lossy} for more details). Following is an important property of the $\dbf$-tilted information, which comes to our aid in constructing a dual feasible point.
\begin{align}
\Ebb[\exp(j_S(S,\dbf)+\lambda^{*}\dbf-\lambda^{*}d(S,\shat))] \leq 1, \quad \forall \shat \in \Sscrhat, \label{eq:dtiltedproperty}
\end{align} where the expectation is with respect to $P_S$ and $\lambda^{*}=-R_S'(\dbf)>0$ is the negative of the slope of the rate-distortion function.

Following is the converse for joint source-channel coding
shown by Kostina and Verd\'{u} in \cite[Theorem 1]{kostina2013lossy}. \begin{theorem}[Kostina-Verd\'{u} bound]\label{thm:kostinaVerdu}
For any source-channel pair $(S,P_{Y|X})$, the existence of a finite blocklength joint source-channel code $(f,g)$ which satisfies  $\Pbb[d(S,\Shat)>\dbf]\leq \epsilon$, requires that
\begin{align}
&\epsilon \geq \inf_{P_{X|S}} \sup_{\gamma>0} \biggl \lbrace \sup_{P_{\bar{Y}}}\Pbb[j_S(S,\dbf)-i_{X;\bar{Y}}(X;Y) \geq \gamma]\nonumber \\&  \qquad \qquad -\exp(-\gamma) \biggr \rbrace\label{eq:kostinaVerdu1}\\
& \geq \sup_{\gamma>0} \biggl \lbrace \sup_{P_{\bar{Y}}}  \Ebb\left [ \inf_x \Pbb[j_S(S,\dbf)-i_{X;\bar{Y}}(x;Y)\geq \gamma \mid S \right] \nonumber\\&\qquad \qquad -\exp(-\gamma)\biggr \rbrace \label{eq:kostinaVerdu2},
\end{align} where, \begin{align}i_{X;\bar{Y}}(x;y)=\log \frac{dP_{Y|X=x}}{dP_{\bar{Y}}}(y)\label{eq:informationdensity},\end{align}
 $P_{\bar{Y}}$ is an arbitrary probability distribution on $\Yscr$ , $S \rightarrow X \rightarrow Y$ holds in \eqref{eq:kostinaVerdu1} and $\Pbb$ in \eqref{eq:kostinaVerdu2} is with respect to $Y$ distributed according to $P_{Y|X=x}$.
\end{theorem}

 Kostina and Verd\'{u} further generalize the above lower bounds to take into account the type of the channel input block and the number of channel input types. For this, an auxiliary random variable $V$ that takes values on $1,\hdots ,T$ is introduced where $T$ is a positive integer that represents the number of channel input types and $V$ represents the type of the channel input block. The following theorem gives the generalized converse of Kostina and Verd\'{u} \cite[Theorem 3]{kostina2013lossy}.
\begin{theorem}[Generalized Kostina-Verd\'{u} bound]\label{thm:generalkostina}
For any source-channel pair $(S,P_{Y|X})$, the existence of a finite blocklength joint source-channel code $(f,g)$ which satisfies  $\Pbb[d(S,\Shat)>\dbf]\leq \epsilon$, requires that
\begin{align}
\epsilon &\geq \inf_{P_{X|S}} \max_{\gamma>0,T} \biggl \lbrace -T\exp(-\gamma)\nonumber\\&\qquad + \sup_{\substack{\bar{Y},V:\\S\rightarrow (X,V)\rightarrow Y}} \Pbb[j_S(S,\dbf)-i_{X;\bar{Y}|V}(X;Y|V) \geq \gamma]\biggr \rbrace \nonumber\\
&\geq \max_{\gamma>0,T}\biggl \lbrace-T\exp(-\gamma) \nonumber \\ & \qquad + \sup_{\bar{Y},V} \Ebb\left[\inf_x \Pbb[j_S(S,\dbf)-i_{X;\bar{Y}|V}(x;Y|V)\geq \gamma |S]\right] \biggr \rbrace, \label{eq:generalkostina}
\end{align} where $T$ is a positive integer, $V$ is a random variable that takes values in $\{1,2, \hdots, T\}$, \begin{align}i_{X;\bar{Y}|V}(x;y|t)=\log \frac{P_{Y|X=x,V=t}}{P_{\bar{Y}|V=t}}(y)\label{eq:generalinfdensity}\end{align} and the probability measure $\Pbb$ in \eqref{eq:generalkostina} is generated by $P_SP_{V|X=x}P_{Y|X=x,V}$.
\end{theorem}
In particular, for a code that gives the minimum excess distortion probability of $\epsilon$, \eqref{eq:kostinaVerdu2}, \eqref{eq:generalkostina} give  lower bounds on this probability. 

\subsection{Recovery and Improvement of the Kostina-Verd\'{u} Converses}
Below we show that DP recovers and improves on the generalized Kostina-Verd\'{u} converse \eqref{eq:generalkostina}. Since the improvement on this converse is immediate, we directly present the proof of the improvement of this converse and remark on recovering the generalized Kostina-Verd\'{u} converse from this.
 Further, a lower bound  which improves on \eqref{eq:kostinaVerdu2} follows from the improvement on \eqref{eq:generalkostina}. Recall that we have assumed that all random variables are discrete.
\begin{theorem}\textit{(DP improves Generalized Kostina-Verd{\'u} bound)}\label{thm:DPIimpliesgeneralkostina}
Consider problem SC with $S$ having distribution $P_S$ and channel conditional probability distribution given by $P_{Y|X}$. Let $\kappa(s,x,y,\shat)\equiv \Ibb{\{ d(s,\shat) > \dbf\}}$ be the loss function where $d:\Sscr\times\Sscrhat\rightarrow[0,+\infty]$ is a distortion function and $\dbf\in [0,\infty)$ is the distortion level. Then, there exists a  feasible point of the problem DP with an objective value which improves on the RHS of \eqref{eq:generalkostina}. Specifically, for any code, we have the following lower bound on excess distortion probability,
\begin{align}
&\Ebb[\Ibb{\{d(S,\Shat)>\dbf\}}]\geq \OPT(\SC) \geq \OPT(\LP) =\OPT(\DP)\nonumber\\&\geq \max_{\gamma,T} \biggl \lbrace \sup_{\Ybar,V} \Ebb \biggl[\inf_x \biggl \lbrace \Pbb[j_S(S,\dbf)-i_{X;\bar{Y}|V}(x;Y|V)\geq \gamma |S]
\non\\&+\exp(j_S(S,\dbf)-\gamma)\sum_{y \in \Yscr}\sum_{t=1}^TP_{V|X}(t|x)P_{\Ybar|V}(y|t)\times \non \\&\Ibb{\{j_S(S,\dbf)-i_{X;\Ybar|V}(x;y|t)<\gamma\}}\biggr \rbrace \biggr] -T\exp(-\gamma) \biggr \rbrace, \label{eq:DPIimpliesgeneralkostina}
\end{align} where $T$ is a positive integer, $V$ is a random variable that takes values on $\{1,2, \hdots, T\}$, $i_{X;\Ybar|V}(x;y|t)$ is as defined in \eqref{eq:generalinfdensity}, $\Pbb$ is generated by $P_{V|X=x}P_{Y|X=x,V}$. Note that the expectation $\Ebb$ on the RHS of \eqref{eq:DPIimpliesgeneralkostina} is with respect to $P_S$.
\end{theorem}
\begin{proof}
Let $V$ be an auxiliary random variable that takes values on $\Vscr:=\{1,\hdots,, T\}$ such that
\begin{align}P_{Y|X}(y|x)=\sum_{t=1}^TP_{Y|X,V}(y|x,t)P_{V|X}(t|x)\label{eq:channelconditional}.
\end{align}
To obtain the required lower bound, consider the following values of dual variables,
\begin{align}
\lambda^b(x,s,y)&\equiv P_S(s)\sum_{t=1}^T\biggl[P_{Y|X,V}(y|x,t)P_{V|X}(t|x)\times\nonumber\\& 
\Ibb{\{P_{Y|X,V}(y|x,t) \leq P_{\Ybar|V}(y|t)\exp(j_S(s,\dbf)-\gamma)\}}\nonumber \\&\quad +P_{\Ybar|V}(y|t)\exp(j_S(s,\dbf)-\gamma)P_{V|X}(t|x)\times \non\\&\Ibb{\{P_{Y|X,V}(y|x,t)> P_{\Ybar|V}(y|t)\exp(j_S(s,\dbf)-\gamma)\}}\biggr] \nonumber,\\
\lambda^a(s,\shat,y)&\equiv\hspace{-0.1cm} -P_S(s)\exp(j_S(s,\dbf)-\gamma)\Ibb{\{d(s,\shat)\leq \dbf\}}\times\non\\&\hspace{1cm} \sum_{t=1}^TP_{\Ybar|V}(y|t), \non \\
\gamma^a(s)&\equiv \inf_{x \in \Xscr}\sum_y \lambda^b(x,s,y),\label{eq:DPimpliesgeneraldual}\\
\gamma^b(y) &\equiv -\exp(-\gamma)\sum_{t=1}^TP_{\Ybar|V}(y|t)\nonumber,
\end{align} and $\mu(s,x,y,\shat)\equiv 0$. Above, $P_{\Ybar|V=t}$ is any probability distribution on $\Yscr$ such that $i_{X;\Ybar|V}(x;y|t)=\log\frac{P_{Y|X=x,V=t}}{P_{\Ybar|V=t}}(y)$ where for any $y\in \Yscr$ and $t\in \Vscr$, $P_{\Ybar|V=t}(y)=0$ implies $P_{Y|X=x,V=t}(y)=0$. Consequently, for any $x \in \Xscr$ and $s \in \Sscr$, $\sum_{y\in \Yscr}\lambda^b(x,s,y)$ becomes
\begin{align}
&P_S(s)\sum_{y \in \Yscr,t \in \Vscr \mid P_{\Ybar|V}(y|t)>0}\biggl[P_{Y|X,V}(y|x,t)P_{V|X}(t|x)\times \nonumber\\&\hspace{-0.1cm}\Ibb\biggl \lbrace\hspace{-0.1cm}\frac{P_{Y|X,V}(y|x,t)}{P_{\Ybar|V}(y|t)}\hspace{-0.1cm} \leq \exp(j_S(s,\dbf)-\gamma)\hspace{-0.1cm}\biggr \rbrace\hspace{-0.1cm}+\hspace{-0.1cm}P_{\Ybar|V}(y|t)P_{V|X}(t|x)\times\nonumber\\&\exp(j_S(s,\dbf)-\gamma) \Ibb\biggl \lbrace\frac{P_{Y|X,V}(y|x,t)}{P_{\Ybar|V}(y|t)}> \exp(j_S(s,\dbf)-\gamma)\biggr \rbrace\biggr]\nonumber\\
&=P_S(s)\sum_{y \in \Yscr,t \in \Vscr \mid P_{\Ybar|V}(y|t)>0}P_{Y|X,V}(y|x,t)P_{V|X}(t|x)\times\nonumber \\&\Ibb\{i_{X;\Ybar|V}(x;y|t) \leq j_S(s,\dbf)-\gamma\}+P_{\Ybar|V}(y|t)P_{V|X}(t|x)\times\nonumber\\&\exp(j_S(s,\dbf)-\gamma) \Ibb\{i_{X;\Ybar|V}(x;y|t)> j_S(s,\dbf)-\gamma\}\biggr]\nonumber\\
&=P_S(s)\sum_{y \in \Yscr}\sum_{t=1}^T\biggl[P_{Y|X,V}(y|x,t)P_{V|X}(t|x)\times\nonumber \\&\Ibb{\{i_{X;\Ybar|V}(x;y|t) \leq j_S(s,\dbf)-\gamma\}}+P_{\Ybar|V}(y|t)P_{V|X}(t|x)\times \non\\& \exp(j_S(s,\dbf)-\gamma)\Ibb{\{i_{X;\Ybar|V}(x;y|t)>j_S(s,\dbf)-\gamma\}}\biggr],\label{eq:sumlambdabgeneral}
\end{align}where the last equality follows since for any $y\in \Yscr$ and $t\in N$, if $P_{\Ybar|V}(y|t)=0$, then $P_{Y|X,V}(y|x,t)=0$ for any $x \in \Xscr$.

It is now sufficient to show that the values of dual variables considered in \eqref{eq:DPimpliesgeneraldual} are in fact feasible for DP. We first check feasibility with respect to constraint (D1). For this, we need to show that
\begin{align*}
&\sum_{y \in \Yscr} \lambda^b(x,s,y)\geq \gamma^a(s), \quad \forall x,s,
\end{align*}which is trivially true by construction. 

To check for the feasibility of dual variables with respect to constraint (D2), we have to show that \begin{align}&-\sum_s P_S(s)\Ibb{\{d(s,\shat)\leq \dbf\}}\exp(j_S(s,\dbf)-\gamma) \sum_{t=1}^TP_{\Ybar|V}(y|t)\nonumber\\& \hspace{1.5cm} \geq  -\exp(-\gamma)\sum_{t=1}^TP_{\Ybar|V}(y|t)\quad \forall  \shat,y.\label{eq:modifiedD2general}
\end{align} To show this,
we start with the LHS of \eqref{eq:modifiedD2general}. For any $\shat \in \Sscrhat$ and $y \in \Yscr$,
\begin{align*}
&-\sum_s P_S(s)\Ibb{\{d(s,\shat)\leq \dbf\}}\exp(j_S(s,\dbf)-\gamma)\sum_{t=1}^TP_{\Ybar|V}(y|t)\\
&\stackrel{(a)}{\geq} -\exp(-\gamma)\hspace{-0.1cm}\sum_s \hspace{-0.1cm}P_S(s)\exp(j_S(s,\dbf))\exp \left(\lambda^{*} \left(\dbf-d(s,\shat)\right)\right)\times\\&\qquad \qquad \qquad \sum_{t=1}^TP_{\Ybar|V}(y|t)\\
&= -\exp(-\gamma)\Ebb\left[\exp \left( j_S(s,\dbf)+\lambda^{*}\dbf-\lambda^{*} d(S,\shat)\right)\right]\times\\&\qquad \qquad \qquad \sum_{t=1}^TP_{\Ybar|V}(y|t)\\
&\stackrel{(b)}{\geq} -\exp(-\gamma)\sum_{t=1}^TP_{\Ybar|V}(y|t),
\end{align*}
where the inequality in $(a)$ is due to 
$\Ibb{\{d(s,\shat)\leq \dbf\}} \leq \exp(\lambda^{*}(\dbf-d(s,\shat)))$ (recall that $\lambda^{*}=-R_S'(\dbf)>0$) and the inequality in $(b)$ follows from \eqref{eq:dtiltedproperty}. Thus (D2) holds.

 To check the feasibility of dual variables in \eqref{eq:DPimpliesgeneraldual} with respect to constraint (D3), we consider the following cases.\\
Case 1: $d(s,\shat)> \dbf$.\\
In this case, $\lambda^a(s,\shat,y) \equiv 0$ and the LHS of (D3) becomes $\lambda^b(x,s,y)$. Further, for any $t \in \Vscr$, two sub-cases arise.\\
Case 1a: \begin{small}$\Ibb{\{P_{Y|X,V}(y|x,t) \leq P_{\Ybar|V}(y|t)\exp(j_S(s,\dbf)-\gamma)\}}=1$\end{small}.\\
 In this case, the terms inside the square bracket in the definition of $\lambda^b(x,s,y)$ become $P_{Y|X,V}(y|x,t)P_{V|X}(t|x)$.\\
 Case 1b: \begin{small}$\Ibb{\{P_{Y|X,V}(y|x,t) \leq P_{\Ybar|V}(y|t)\exp(j_S(s,\dbf)-\gamma)\}}=0$\end{small}.\\
 In the considered case, $ P_{\Ybar|V}(y|t)\exp(j_S(s,\dbf)-\gamma)<P_{Y|X,V}(y|x,t)$ and consequently, the terms inside the square bracket in $\lambda^b(x,s,y)$ become less than $P_{Y|X,V}(y|x,t)P_{V|X}(t|x).$
 
  Thus, in each case, for any $t \in \Vscr$, the term inside the square bracket in $\lambda^b(x,s,y)$ is upper bounded by $P_{Y|X,V}(y|x,t)P_{V|X}(t|x)$. Consequently, $\lambda^b(x,s,y)$ is upper bounded as,
\begin{align*}
P_S(s)\sum_{t=1}^TP_{Y|X,V}(y|x,t)P_{V|X}(t|x)\stackrel{(a)}{=}P_S(s)P_{Y|X}(y|x),
\end{align*}which is the RHS of (D3). The equality in $(a)$ follows from \eqref{eq:channelconditional}.

Case 2: $d(s,\shat)\leq \dbf$\\ In this case, the RHS of (D3) is zero.
The LHS can
  be upper bounded employing $P_{V|X}(t|x) \leq 1$ for all $t,x$ in the expression for $\lambda^b(x,s,y)$ to get that the LHS is at most\begin{align}
&  P_S(s)\sum_{t=1}^T\biggl[P_{Y|X,V}(y|x,t) 
 \Ibb\{P_{Y|X,V}(y|x,t) \leq \nonumber\\&P_{\Ybar|V}(y|t)\exp(j_S(s,\dbf)-\gamma)\}+P_{\Ybar|V}(y|t)\exp(j_S(s,\dbf)-\gamma)\times \non\\& \qquad\Ibb{\{P_{Y|X,V}(y|x,t)> P_{\Ybar|V}(y|t)\exp(j_S(s,\dbf)-\gamma)\}}\nonumber\\&\qquad-P_{\Ybar|V}(y|t)\exp(j_S(s,\dbf)-\gamma)\biggr]
 \label{eq:LHSD3}.
\end{align}Just as in Case 1, for any $t \in \Vscr$, two sub-cases arise in \eqref{eq:LHSD3}.\\
Case 2a:\begin{small} $\Ibb{\{P_{Y|X,V}(y|x,t) \leq P_{\Ybar|V}(y|t)\exp(j_S(s,\dbf)-\gamma)\}}=1$\end{small}.\\
In this case, $P_{Y|X,V}(y|x,t)$ is upper bounded by $P_{\Ybar|V}(y|t)\exp(j_S(s,\dbf)-\gamma)$ showing that the term in the square bracket of \eqref{eq:LHSD3} is nonpositive. \\
Case 2b: \begin{small}$\Ibb{\{P_{Y|X,V}(y|x,t) \leq P_{\Ybar|V}(y|t)\exp(j_S(s,\dbf)-\gamma)\}}=0$\end{small}.\\ In this case, it can be  seen that the term inside the square bracket of \eqref{eq:LHSD3} becomes zero. 

Thus, for any $t \in N$, the terms inside the square bracket of \eqref{eq:LHSD3} evaluate to a non-positive quantity. Consequently, the LHS of (D3) is upper bounded by a nonpositive quantity,
thereby satisfying the constraint (D3).
 Thus, the dual variables considered in \eqref{eq:DPimpliesgeneraldual} are feasible for DP. 


Employing \eqref{eq:sumlambdabgeneral}, the dual cost is then obtained as,
\begin{align}
&\sum_{s \in \Sscr} \gamma^a(s)+\sum_{y \in \Yscr} \gamma^b(y)\nonumber\\&= \sum_sP_S(s)\inf_x \sum_{y \in \Yscr}\sum_{t=1}^T\biggl[P_{Y|X,V}(y|x,t)P_{V|X}(t|x)\times\nonumber \\
&\Ibb{\{i_{X;\Ybar|V}(x;y|t) \leq j_S(s,\dbf)-\gamma\}}+P_{\Ybar|V}(y|t)P_{V|X}(t|x)\times \non\\
& \exp(j_S(s,\dbf)-\gamma)\Ibb{\{i_{X;\Ybar|V}(x;y|t)>j_S(s,\dbf)-\gamma\}}\biggr]\non \\
&-\exp(-\gamma)\sum_{y \in \Yscr}\sum_{t=1}^TP_{\Ybar|V}(y|t).
\end{align} Since this lower bound is valid for any $\gamma$, $T$, $P_{\Ybar}$,
we get the rightmost inequality in \eqref{eq:DPIimpliesgeneralkostina}. Using Lemma~\ref{lem:dualfeasiblebound}, the proof is complete.
\end{proof}
\begin{remarkc}[Recovering the Kostina-Verd\'{u} Converse \eqref{eq:generalkostina}:]
To recover the Kostina-Verd\'{u} converse \eqref{eq:generalkostina} from our DP, consider the dual variables in \eqref{eq:DPimpliesgeneraldual} with $\lambda^b$ chosen as
\begin{align*}\lambda^b(x,s,y)&\equiv P_S(s)\sum_{t=1}^T\biggl[P_{Y|X,V}(y|x,t)P_{V|X}(t|x)\times\nonumber\\& 
\Ibb{\{P_{Y|X,V}(y|x,t) \leq P_{\Ybar|V}(y|t)\exp(j_S(s,\dbf)-\gamma)\}}\biggr] \nonumber.
\end{align*}It is easy to see that the dual cost of this feasible point gives the converse \eqref{eq:generalkostina}. 

The presence of an additional non-negative term corresponding to $\Ibb{\{P_{Y|X,V}(y|x,t)> P_{\Ybar|V}(y|t)\exp(j_S(s,\dbf)-\gamma)\}}$ in the value of $\lambda^b(x,s,y)$ in \eqref{eq:DPimpliesgeneraldual} results in the improvement on the Kostina-Verd\'{u} converse in \eqref{eq:generalkostina}. Numerical comparisons for the same are given in Section~\ref{sec:BMS-BSC}. 
%
\end{remarkc}\\
\begin{remarkc}[On the Construction of Feasible Points of DP:]\label{rem:constfeas}
Notice that the constructions of dual variables in \eqref{eq:DPimpliesgeneraldual} employ $\mu\equiv 0.$ Consequently, these variables are feasible for $\DP'$, the dual of the reduced LP relaxation $\LP'.$ Tighter bounds could be obtained by making a better use of $\mu.$ However, since $\mu$ is required to be nonnegative and it is present in all constraints (D1), (D2), (D3), and in the objective of $\DP$, using it in the construction of dual variables becomes challenging. Throughout this paper, we have not used a nonzero value for $\mu$. 

Note that of the three dual constraints in $\DP'$,  (D3) is the hardest. We have to find functions $\lambda^a$, $\lambda^b$ on strict subpsaces of $\Zscr$ such that they are pointwise dominated by the RHS of (D3), which is defined over the full space $\Zscr.$ Moreover, when $\kappa(z) \equiv \kappa(s,\shat)$, the RHS of (D3) becomes a \textit{product} of two terms, $P_S(s)\kappa(s,\shat)$, which is a source coding-like term and $P_{Y|X}(y|x)$, which is a channel coding-like term. On the other hand in the LHS, we have a \textit{sum} of $\lambda^a(s,\shat,y)$ and $\lambda^b(s,x,y)$. This is probably indicative that a logarithm would be involved in the construction of dual feasible points. 
Also, notice that terms in the LHS do not have a clean ``source-channel'' separation: $\lambda^a$ depends on $s,\shat$ as well as $y$. Of course,  one may choose to take $\lambda^a$ to be constant over $y$, but it is not clear that this is optimal. Indeed,  \eqref{eq:DPimpliesgeneraldual} sets $\lambda^a$ to be nontrivially dependent on $y.$
\end{remarkc}

\begin{remarkc}[Choice of $\gamma$:]
Notice that in the converse in Theorem~\ref{thm:DPIimpliesgeneralkostina}, the supremum is taken over any $\gamma$. To get the tightest bound, it is sufficient to take the supremum over $\gamma>0$. Although letting $\gamma$ take nonpositive values does not improve our bound, $\gamma<0$ comes handy in deriving  new converses as will be seen in Section~\ref{sec:BMS-BSC}.
\end{remarkc}\\


Finally, if $T=1$ in the converse in Theorem~\ref{thm:DPIimpliesgeneralkostina}, we get the following improvement on converse \eqref{eq:kostinaVerdu2}.
\begin{corollary}[DP improves Kostina-Verd{\'u} bound]\label{thm:DPimplieskostina}
Consider the lossy joint source-channel coding setting of Theorem~\ref{thm:DPIimpliesgeneralkostina}.
Then for any code, we have
\begin{align}
&\Ebb[\Ibb{\{d(S,\Shat)>\dbf\}}]\geq \OPT(\SC)\geq \OPT(\DP)\nonumber\\ & \hspace{-0.1cm}\geq  \sup_{\gamma} \biggl \lbrace \sup_{P_{\Ybar}}\Ebb \biggl[\inf_x\biggl\lbrace
\Pbb[j_S(S,\dbf)-i_{X;\bar{Y}}(x;Y)\geq \gamma |S]+\non\\&
\hspace{-0.1cm}\sum_{y\in \Yscr}P_{\Ybar}(y)\exp(j_S(S,\dbf)-\gamma)\Ibb\{j_S(S,\dbf)-i_{X;\Ybar}(x;y) < \gamma\}\biggr \rbrace \biggr]\non\\&\qquad \qquad -\exp(-\gamma)\biggr \rbrace \label{eq:DPIimplieskostina},
\end{align}where $P_{\Ybar}$ is any probability measure on $\Yscr$ such that $i_{X,\Ybar}(x;y)$ is as defined in \eqref{eq:informationdensity}, $\Pbb$ is with respect to $Y$ distributed acording to $P_{Y|X=x}$. Note that the expectation $\Ebb$ in the RHS of \eqref{eq:DPIimplieskostina} is with respect to $P_S$.
\end{corollary}

Coming to the asymptotics for the joint-source channel coding problem, lower bound the converse in Theorem~\ref{thm:DPIimpliesgeneralkostina} to get the Kostina-Verd\'{u} bound in \eqref{eq:generalkostina}. Let $\Mscr$, $\widehat{\Mscr}$, $\Ascr$ and $\Bscr$ represent the source alphabet, destination alphabet, channel input and output alphabets respectively. Then, we take $\Sscr=\prod_{i=1}^k\Sscr_i$, $\Sscrhat=\prod_{i=1}^k\Sscrhat_i$ where $\Sscr_i =\Mscr$ and $\Sscrhat_i =\widehat{\Mscr}$ for all $i$, $\Xscr=\prod_{i=1}^n\Xscr_i,$ where $\Xscr_i=\Ascr$ for all $i$ and $\Yscr=\prod_{i=1}^n\Yscr_i,$ where $\Yscr_i=\Bscr$ for all $i$, and take the source $S_i$ to be stationary and the channel $P_{Y_i|X_i}$ as stationary and memoryless, 
%
then under the assumptions of Theorem 10 in \cite{kostina2013lossy}, 
\begin{align*}
&\Ebb[\I{d(S,\Shat)>\dbf}]\geq \OPT({\SC})\geq \OPT(\rm DP)\\&\geq \Qbb\left( \frac{nC-kR(\dbf)+(|\Ascr |-\frac{1}{2})\log(n+1)}{\sqrt{n\Vbb+k\Vscr(\dbf)-L_2|\Ascr |}} \right)+o(n,k),
\end{align*}where $o(n,k)$ vanishes as $n,k$ go to infinity, $\Qbb$ is the standard Gaussian complimentary cumulative distribution function, $\Ascr$ is the channel input alphabet, 
$C$ is the channel capacity, 
$\Vbb={\rm Var}[i^{*}_{X;Y}(X^{*};Y^{*})]>0$, $i^{*}_{X;Y}(X^{*};Y^{*})=\log \frac{dP_{Y|X=x}}{dP_{Y^{*}}}(y)$, $X^{*} \in \Ascr$, $ Y^{*} \in \Bscr$ are the capacity achieving input and output random variables, respectively, $\Vscr(\dbf)={\rm Var}[j_S(s,\dbf)]$ and $L_2$ is a non-negative constant. This follows from employing the central limit theorem-based asymptotic analysis in~\cite{kostina2013lossy}. It is known from \cite{kostina2013lossy} that this converse is asymptotically tight.
\subsubsection{Lossy Source Coding}
For the finite blocklength lossy source coding problem, the following converse which improves on the Kostina-Verd\'{u} converse in \cite[Theorem 7]{kostina2012fixed}  follows from the converse in Corollary~\ref{thm:DPimplieskostina}.
\begin{corollary}\textit{(DP improves on Kostina-Verd\'{u} 
Lossy Source Coding Converse)} \label{thm:lossysourcecoding}
Consider the setting of Corollary~\ref{thm:DPimplieskostina} with $\Xscr=\Yscr=\{1,\hdots, M\}$, $M \in \Nbb$ and channel conditional distribution $P_{Y|X}(y|x)\equiv \I{x=y}$. 
Then, for any code, the following  lower bound follows from \eqref{eq:DPIimplieskostina}, 
\begin{align}
&\Ebb[\I{d(S,\Shat)>\dbf}]\geq \OPT({\SC})\geq \OPT(\rm DP) \nonumber \\& \quad\geq \sup_{\gamma} \biggl \lbrace -\exp(-\gamma)+ \Pbb[j_S(S,\dbf) \geq \gamma+ \log M]+\frac{1}{M}\times\nonumber\\&\hspace{-0.3cm}\sum_sP_S(s)\exp(j_S(s,\dbf)-\gamma)\Ibb{\{j_S(s,\dbf) < \log M+\gamma\}} \biggr \rbrace.\label{eq:lossysourcecoding}
\end{align} 
\end{corollary} To get to the bound in \eqref{eq:lossysourcecoding}, set $P_{\Ybar}(y)\equiv \frac{1}{M}$ in \eqref{eq:DPIimplieskostina}.
It is easy to see that the converse in \eqref{eq:lossysourcecoding} implies the asymptotically tight lossy source coding converse of Kostina-Verd\'{u} \cite[Theorem 7]{kostina2012fixed}. Consequently,
Corollary~\ref{thm:lossysourcecoding} shows that the LP relaxation is asymptotically tight for lossy (and lossless) source coding with the probability of excess distortion as the loss criterion. 
\subsubsection{Channel Coding}
For channel coding problems, $d(s,\shat)\equiv \I{s \neq \shat}$, $\dbf=0$ and $S$ is uniformly distributed on $\Sscr$. 
 Kostina and Verd\'{u} in \cite{kostina2015channels} showed that the lower bound derived in Theorem~\ref{thm:kostinaVerdu} also implies the channel coding converse proposed by Wolfowitz  \cite{wolfowitz1968notes} in the absence of cost constraints on channel input,
\begin{align}
&\Ebb [\Ibb\{S \neq \Shat\}] \geq \sup_{\gamma>0}\biggl \lbrace\sup_{P_{\Ybar}}\inf_x \Pbb[i_{X;\Ybar}(x;Y)\leq \log M-\gamma]\non\\& \qquad \qquad -\exp(-\gamma)\biggr \rbrace. \label{eq:wolfowitz}
\end{align} 
The following corollary gives a new converse for channel coding problem in the absence of cost constraints derived from \eqref{eq:DPIimplieskostina} which improves on the above converse of Wolfowitz. 
\begin{corollary}[DP improves on Wolfowitz's Converse]\label{cor:channelcoding}
Consider the setting of Corollary~\ref{thm:DPimplieskostina} with $\Sscr=\Sscrhat=\{1,2, \hdots, M\}$, $M \in \Nbb$, $P_S(s)\equiv \frac{1}{M}$. Let $d(s,\shat) \equiv \I{s \neq \shat}$ and $\dbf=0$. Then, for any code, 
the following lower bound on the minimum error probability holds,
\begin{align}
&\Ebb[\I{S \neq \Shat}]\geq \OPT({\SC})\geq \OPT(\rm DP)\nonumber\\& \geq \sup_{\gamma} \biggl \lbrace \sup_{P_{\Ybar}}\inf_x\biggl[\Pbb[i_{X;\Ybar}(x;Y) \leq \log M-\gamma]+\hspace{-0.05cm}M\exp(-\gamma)\times\non\\&\sum_{y\in \Yscr}P_{\Ybar}(y)\Ibb{\{i_{X;\Ybar}(x;y) > \log M-\gamma\}}\biggr]-\exp(-\gamma)\biggr \rbrace. \label{eq:channelcoding}
\end{align} 
Further, let $\Xscr=\prod_{i=1}^n \Xscr_i$, where $\Xscr_i=\Ascr$ for all $i$ and $\Ascr$ represents the channel input alphabet and $\Yscr=\prod_{i=1}^n \Yscr_i$, where $\Yscr_i=\Bscr$ for all $i$ and $\Bscr$ represents the channel output alphabet. Let $y_i \in \Yscr_i$ represent an element of $\Yscr_i$ and $x_i \in \Xscr_i$ represent an element of $\Xscr_i$.
If the channel is stationary and memoryless, \ie, $P_{Y|X}(y|x)\equiv\prod_{i=1}^n P_{Y_i|X_i}(y_i|x_i),$ for $n\in \Nbb$  and $P_{Y_i|X_i} \in \Pscr(\Bscr|\Ascr)$ is independent of $i$, and  $M = \exp(nR)$ for any $R$ greater than the capacity of $P_{Y_i|X_i}$, then as $n \rightarrow \infty,$ the RHS of \eqref{eq:channelcoding} and hence \OPT(\DP) and \OPT(\SC) tend to unity.
\end{corollary} 

In the above corollary, \eqref{eq:channelcoding} follows from Theorem~\ref{thm:DPimplieskostina}. Notice that the converse in \eqref{eq:channelcoding} has an additional non-negative term corresponding to $\Ibb\{i_{X;\Ybar}(x;y) > \log M-\gamma\}$ compared to Wolfowitz's converse in \eqref{eq:wolfowitz}. The second claim in the corollary is the channel coding strong converse, which follows from the lower bound in \eqref{eq:channelcoding} by applying the law of large numbers for non-identically distributed random variables as explained in \cite[Section IV]{kostina2015channels}. Also note that the `achievability' part  of Shannon's channel coding theorem directly shows that for $R$ strictly less than the channel capacity, $\OPT(\LP) \rightarrow 0$ as $n \rightarrow 0$ (this is because $0\leq \OPT(\LP) \leq \OPT(\SC)$ and $\OPT(\SC) \rightarrow 0$).
\subsubsection{Linear Programming Relaxation of Matthews using Non-Signaling Codes and the Polyanskiy-Poor-Verd\'{u} Converse} The authors were made aware of the work of Matthews  \cite{matthews2012linear} by an anonymous reviewer. Matthews in \cite{matthews2012linear} considers the finite blocklength channel coding  problem (with the same setting as in Corollary~\ref{cor:channelcoding}) and obtains a lower bound on the minimum probability of error, $\mathcal{E}(M)$, achieved by a channel code of size $M$, by relaxing the problem to an optimization problem over \textit{non-signaling codes}. The resulting relaxed problem is posed as the following linear program, 
 $$ \begin{small}\problemsmalla{NS}
	{ Q_{X|S},Q_{\widehat{S}|Y},W}
	{\displaystyle 1-\frac{1}{M}\sum_{s,x,y}P_{Y|X}(y|x)W(s,x,y,s)}
				 {\begin{array}{r@{\ }c@{\ }l}
 \sum_{x} W(z)-Q_{\widehat{S}|Y}(\shat|y)&=&0  
  \hspace{0.1cm}  \forall s,\shat,y\\ 
 \sum_{\shat}W(z)-Q_{X|S}(x|s)&=&0  \hspace{0.1cm} \forall x,s,y\\
\sum_{x,\shat}W(s,x,y,\shat)&=& 1 \hspace{0.15cm}\forall s,y ,\\
 Q_{X|S}(x|s) &\geq& 0  \hspace{0.15cm} \forall s,x\\
 Q_{\widehat{S}|Y}(\shat|y) &\geq& 0   \hspace{0.2cm} \forall \shat,y\\
 W(z)&\geq& 0  \hspace{0.2cm} \forall z.
	\end{array}} \end{small} $$ Matthews finds the optimal solution of NS by resorting to an equivalent linear program over symmetrized non-signaling codes and its dual program. Specifically, he shows,
	\begin{align*}\OPT({\rm NS})=\max_{z} \min_{x} \sum_{y}\biggl[\min\{z_y, P_{Y|X}(y|x)\}-\frac{z_y}{M}\biggr], 
	\end{align*} whereby,
\begin{align}
	\mathcal{E}(M)&\geq 
\max_{z} \min_{x} \sum_{y}\biggl[\min\{z_y, P_{Y|X}(y|x)\}-\frac{z_y}{M}\biggr].\label{eq:OPTNS}	\end{align}
	 Furthermore, he shows that the upper bound on $M$ that follows from \eqref{eq:OPTNS}, $M^{{\rm NS}}(\epsilon)$,  where $\mathcal{E}(M) \leq \epsilon$ is in fact the hypothesis testing based converse of Polyanskiy, Poor and Verd\'{u} \cite[Theorem 27]{polyanskiy2010channel}. Precisely,
	\begin{align}
	M(\epsilon)\leq M^{{\rm NS}}(\epsilon)=\lfloor M^{{\rm PPV}}(\epsilon)\rfloor, \label{eq:Mbound}
	\end{align}where $$M^{{\rm PPV}}(\epsilon)=\sup_{P_X \in \Pscr(\Xscr)}\inf_{Q_Y \in \Pscr(\Yscr)}\frac{1}{\beta_{1-\epsilon}(P_{XY},P_X \times Q_Y)},$$ with $P_{X,Y}(x,y)=P_{Y|X}(y|x)P_X(x)$ is the bound on $M$ obtained by Polyanskiy, Poor and Verd\'{u} in \cite[Theorem 27]{polyanskiy2010channel}. Note that if $P,Q \in \Pscr(\mathcal{R})$ where $\mathcal{R}$ is a finite set and identifying $P$ with the null hypothesis, $\beta_{1-\epsilon}(P,Q)$ represents the minimum type II error, $\sum_{r \in \Rscr} T(r)Q(r)$ achieved by statistical tests $T$ with the type I error not exceeding $\epsilon$, \ie, $\sum_{r \in \mathcal{R}}T(r)P(r)\geq 1-\epsilon$. Other works related to \cite{matthews2012linear} can be found in \cite{leung2015power}, \cite{matthews2014finite}.
	
It is easy to see that NS is in fact our simpler relaxation ${\rm LP'}$, thereby giving that $\OPT({\rm NS})=\OPT({\rm LP'})=\OPT({\rm DP'})$. 
Consequently, we have the following theorem.
\begin{theorem}[DP implies Polyanskiy-Poor-Verd\'{u} Converse]
Consider problem SC with the setting of Corollary~\ref{cor:channelcoding}. Consequently, for any code,
\begin{align*}
\Ebb[\Ibb\{S\neq \Shat\}]&\geq \OPT(\SC) \geq \OPT(\DP) \geq \OPT(\DP'),
\end{align*} and $$M(\epsilon)\leq M^{\DP'}(\epsilon) = \lfloor M^{\rm PPV}(\epsilon)\rfloor,$$
where $M^{\DP'}(\epsilon)$ is the upper bound on $M$ obtained from $\DP'$ by putting  $\OPT(\DP')\leq \OPT(\SC)\leq \epsilon.$ 
\end{theorem}
Although NS is the same as our ${\rm LP'}$, we note that our original relaxation LP is tighter than NS due to the presence of additional inequality constraints, $Q_{X|S}(x|s)+Q_{\widehat{S}|Y}(\shat|y)-W(z)\leq 1 $ for all $s,x,y,\shat$ and hence,
 has the potential to obtain  better bounds.

Notice that our improved converse for channel coding in \eqref{eq:channelcoding} follows from \eqref{eq:OPTNS}. To see this, lower bound the RHS of \eqref{eq:OPTNS} by taking $z_y=P_{\Ybar}(y)M\exp(-\gamma)$ and replace maximum over $z$ in the bound in \eqref{eq:OPTNS} with supremum over $\gamma$ and $P_{\bar{Y}}$. 

While Matthews's relaxation is obtained as a linear program by appealing to the abstract idea of non-signaling codes, our LP relaxation is obtained mechanistically by appealing to principles and techniques in optimization. As such, our approach can generate additional constraints as seen in our relaxation LP, which further tightens the relaxation. Moreover, our relaxation approach is extendable to network settings. 


\subsection{Further strengthening of these bounds}\label{sec:strengtheningofbounds} Our main message from the results obtained in the first part of this section is that they serve to demonstrate the linear programming based framework. That the gap between $\OPT(\SC)$ and $\OPT(\DP)$ is small (and vanishes asymptotically) shows that better converses could be found by thinking within the general framework of the abstract converse given by $\DP$. 

Indeed some avenues for strengthening the above converses are already suggested by the proofs of Theorems~\ref{thm:DPIimpliesgeneralkostina} and \ref{thm:lossysourcecoding}. Notice from the dual constraint (D2), that when $\mu\equiv 0$, the optimal value of $\gamma^b(y) $ is $ \min_{\shat} \sum_s \lambda^a(s,\shat,y)$. On the other hand, the $\gamma^b(y)$ constructed in \eqref{eq:DPimpliesgeneraldual} is, in general, less than this value. Setting $\gamma^b(y) \equiv \min_{\shat} \sum_s \lambda^a(s,\shat,y) $ would lead to an improved bound. Note that $\gamma^a(s)$ has been set to its optimal value in \eqref{eq:DPimpliesgeneraldual}. Setting $\gamma^b$ as indicated above and all other variables as in \eqref{eq:DPimpliesgeneraldual}, we obtain the following tighter lower bound on the minimum excess distortion probability of a lossy joint source-channel code than the converse in Theorem~\ref{thm:DPIimpliesgeneralkostina}.
\begin{theorem}\textit{(Joint Source-Channel Coding -- Further Improvement)}\label{thm:DPimpliestighterbounds}
Consider the lossy joint source-channel coding setting of Theorem~\ref{thm:DPIimpliesgeneralkostina}.
 For any code, we have the following lower bound on the excess distortion probability that improves on the converse in Theorem~\ref{thm:DPIimpliesgeneralkostina},
\begin{align}
&\Ebb[\Ibb{\{d(S,\Shat)>\dbf\}}]\geq \OPT(\SC) \geq \OPT(\DP)\nonumber\\&\geq \max_{\gamma,T}\sup_{\Ybar,V} \biggl \lbrace  \Ebb \biggl[\inf_x \biggl \lbrace \Pbb[j_S(S,\dbf)-i_{X;\bar{Y}|V}(x;Y|V)\geq \gamma |S]
\non\\&+\exp(j_S(S,\dbf)-\gamma)\sum_{y \in \Yscr}\sum_{t=1}^TP_{V|X}(t|x)P_{\Ybar|V}(y|t)\times \non \\&\Ibb{\{i_{X;\Ybar|V}(x;y|t)>j_S(S,\dbf)-\gamma\}}\biggr\rbrace\biggr]\non \\&-T \sup_{\shat}\Ebb \biggl[\exp(j_S(S,\dbf)-\gamma)\Ibb\{d(S,\shat) \leq \dbf\}\biggr]\biggr \rbrace, \label{eq:DPIgivestighterbounds}
\end{align} where $T$ is a positive integer, $V$ is a random variable that takes values on $\{1,2, \hdots, T\}$, $i_{X;\Ybar|V}(x;y|t)$ is as defined in \eqref{eq:generalinfdensity}, $\Pbb$ is generated by $P_{V|X=x}P_{Y|X=x,V}$ and $\Ebb$ in the RHS of \eqref{eq:DPIgivestighterbounds} is with respect to $P_S$. 
\end{theorem} 

The following further improvement on the converse in \eqref{eq:lossysourcecoding} for lossy source coding follows from \eqref{eq:DPIgivestighterbounds}.
\begin{corollary} \textit{(Lossy Source Coding -- Further Improvement)}\label{thm:lossysourcecodingimproved}
Consider the setting of Theorem~\ref{thm:DPimplieskostina} with $\Xscr=\Yscr=\{1,\hdots, M\}$, $M \in \Nbb$ and channel conditional distribution $P_{Y|X}(y|x)\equiv \I{x=y}$. 
Then, for any code, the following  lower bound follows from \eqref{eq:DPIgivestighterbounds},
\begin{align}
&\Ebb[\I{d(S,\Shat)>\dbf}]\geq \OPT({\SC})\geq \OPT(\rm DP) \nonumber \\& \quad\geq \sup_{\gamma} \biggl \lbrace \Pbb[j_S(S,\dbf) \geq \gamma+ \log M]+\frac{1}{M}\times\nonumber\\&\quad \sum_sP_S(s)\exp(j_S(s,\dbf)-\gamma)\Ibb{\{j_S(s,\dbf) < \log M+\gamma\}}\non\\&-\sup_{\shat}\frac{\exp(-\gamma)}{M}\sum_s P_S(s)\exp(j_S(s,\dbf))\Ibb\{d(s,\shat)\leq \dbf\} \biggr \rbrace.\label{eq:lossysourcecodingimproved}
\end{align}
\end{corollary}

For the channel coding problem a similar improvement is not possible. Indeed, with $P_S(s)\equiv \frac{1}{M}$, $d(s,\shat)=\Ibb\{s \neq \shat\}$ and $\dbf=0$, our improved construction $\gamma^b(y)=\inf_{\shat}\sum_s \lambda^a(s,\shat,y)=-\exp(-\gamma)P_{\bar{Y}}(y)$, which is as in \eqref{eq:DPIimpliesgeneralkostina}. Consequently, the tighter converse that follows from \eqref{eq:DPIgivestighterbounds} (with $T=1$) for channel coding coincides with \eqref{eq:channelcoding}.

Before we present further results, we numerically illustrate these bounds in the next section.
\section{Lossy Transmission of a Binary Memoryless Source (BMS) over a BSC}\label{sec:BMS-BSC}
In this section, we employ the converses in Theorem~\ref{thm:DPIimpliesgeneralkostina} and Theorem~\ref{thm:DPimpliestighterbounds} with $T=1$ to obtain lower bounds on the minimum excess distortion probability of transmitting a BMS over a BSC. We then consider the lossy source coding of a BMS with average bit-wise Hamming distance as the distortion measure, and apply the converses in Corollary~\ref{thm:lossysourcecoding}, Corollary~\ref{thm:lossysourcecodingimproved} and Kostina-Verd\'{u} converse in \cite[Theorem 7]{kostina2012fixed}. We numerically illustrate these bounds. 

We consider the following setting. Let $\Sscr=\Sscrhat=\{0,1\}^k$ and $\Xscr=\Yscr=\{0,1\}^n$. The probability distribution of the binary memoryless source is given as \begin{align}
P_S(s)&=\prod_{i=1}^k P_{S_i}(s_i),\quad \mbox{where}\non\\ P_{S_i}(s_i)&\equiv p \Ibb\{s_i=1\}+(1-p)\Ibb\{s_i=0\},\label{eq:BMS}\end{align} where $p\in [0,1]$. Let the distortion measure be $d(s,\shat)=\frac{1}{k}\sum_{i=1}^k \Ibb\{s_i \neq \shat_i\}$ and $\dbf\in [0,p)$ be the excess distortion level. The $d$-tilted information for this BMS evaluates to
\begin{align}
j_S(s,\dbf)&\equiv kH(p)-kH(\dbf)+(w_s-kp)\log_2\biggl(\frac{1-p}{p}\biggr),\label{eq:BMStilted}
\end{align}
where recall that $w_s$ represents the Hamming weight of $s \in \Sscr$. The binary memoryless symmetric channel is as given in \eqref{eq:BSC} for $\epsilon<\half$. The rate-distortion function of the source and the channel capacity are given as,
$$R_S(\dbf)=H(p)-H(\dbf),\quad C=1-H(\epsilon),$$ 
and the rate of transmission of the joint source-channel code is $r=\frac{k}{n}$.
\subsection{Joint Source-Channel Coding of a BMS over a BSC}
The following converse on the minimum excess distortion probability of a BMS over a BSC follows from \eqref{eq:DPIimplieskostina}. 

\begin{corollary}\textit{(Joint Source-Channel Coding Converse of \eqref{eq:DPIimplieskostina} for BMS-BSC)}\label{cor:BMSBSC1}
Consider problem SC with $\Sscr=\Sscrhat=\{0,1\}^k$, $\Xscr=\Yscr=\{0,1\}^n$, $P_{S}(s)$ is as given in \eqref{eq:BMS} and the channel is the BSC as given in \eqref{eq:BSC} with $\epsilon<0.5$.  Let the loss function be $\kappa(s,x,y,\shat)\equiv \Ibb \biggl\lbrace \frac{1}{k}\sum_{i=1}^k \Ibb\{s_i \neq \shat_i\}>\dbf \biggr \rbrace$, where $\dbf <p$. Then, for any code, the following lower bound follows from \eqref{eq:DPIimplieskostina},
\begin{align}
&\Ebb\biggl[ \Ibb \biggl\lbrace \frac{1}{k}\sum_{i=1}^k \Ibb\{s_i \neq \shat_i\}>\dbf \biggr \rbrace\biggr]\geq \OPT({\SC})\geq \OPT(\rm DP)\non\\&\hspace{-0.1cm}\geq \sup_{\gamma} \biggl \lbrace \sum_{b=0}^k \Comb{k}{b}p^b(1-p)^{k-b}\biggl[1-\hspace{-0.2cm}\sum_{a=0}^{ n\epsilon+\theta -1}\hspace{-0.3cm}\Comb{n}{a}\epsilon^a(1-\epsilon)^{n-a}  \non\\&\qquad +\epsilon^{n\epsilon+n\theta}(1-\epsilon)^{n-n\epsilon-\theta}\sum_{a=0}^{n\epsilon+\theta -1}\Comb{n}{a}\biggr]-2^{-\gamma}\biggr\rbrace,\label{eq:BUSBSC1} \\
&\mbox{where} \hspace{0.2cm}\theta=\frac{\gamma-kR_S(\dbf)+nC-(b-kp)\log_2\biggl(\frac{1-p}{p}\biggr)}{\log_2 \biggl(\frac{1-\epsilon}{\epsilon}\biggr)}. \label{eq:theta} \end{align}
\end{corollary}
\begin{proof}
To get to this converse, set $P_{\Ybar}(y)\equiv \frac{1}{2^n}$ and  substitute $j_S(s,\dbf)$ from \eqref{eq:BMStilted} and \begin{align*}
i_{X;\Ybar}(x;y)&\equiv n(\log_2 2-H(\epsilon))-(d_{x,y}-n\epsilon)\log_2 \biggl( \frac{1-\epsilon}{\epsilon}\biggr),
\end{align*} in the bound in \eqref{eq:DPIimplieskostina}. Recall that $d_{x,y}$ represents the Hamming distance between $x$ and $y$. A simple calculation then results in the required bound.
\end{proof}
In particular, if $p=0.5$, a few easy calculations reveal that \eqref{eq:BUSBSC1} results in, \begin{align}
&\Ebb\biggl[ \Ibb \biggl\lbrace \frac{1}{k}\sum_{i=1}^k \Ibb\{s_i \neq \shat_i\}>\dbf \biggr \rbrace\biggr]\geq \OPT({\SC})\geq \OPT(\rm DP)\non\\& \geq \sup_{\gamma} \biggl \lbrace 1-\sum_{a=0}^{r}\Comb{n}{a}\epsilon^a(1-\epsilon)^{n-a}+\epsilon^{r+1}(1-\epsilon)^{n-r-1}\times \non\\&\qquad \qquad \biggl[\sum_{a=0}^r \Comb{n}{a}-2^{n-k+kH(\dbf)}\biggr]\biggr \rbrace,\label{eq:BUSBSCimprovedkv}
\end{align}where $r= n\epsilon+\theta-1$, $\theta$ is as defined in \eqref{eq:theta}.
\begin{figure}
\begin{center}
\includegraphics[scale=0.62,clip=true, trim = 1.1in 3.85in 0.9in 3.75in]{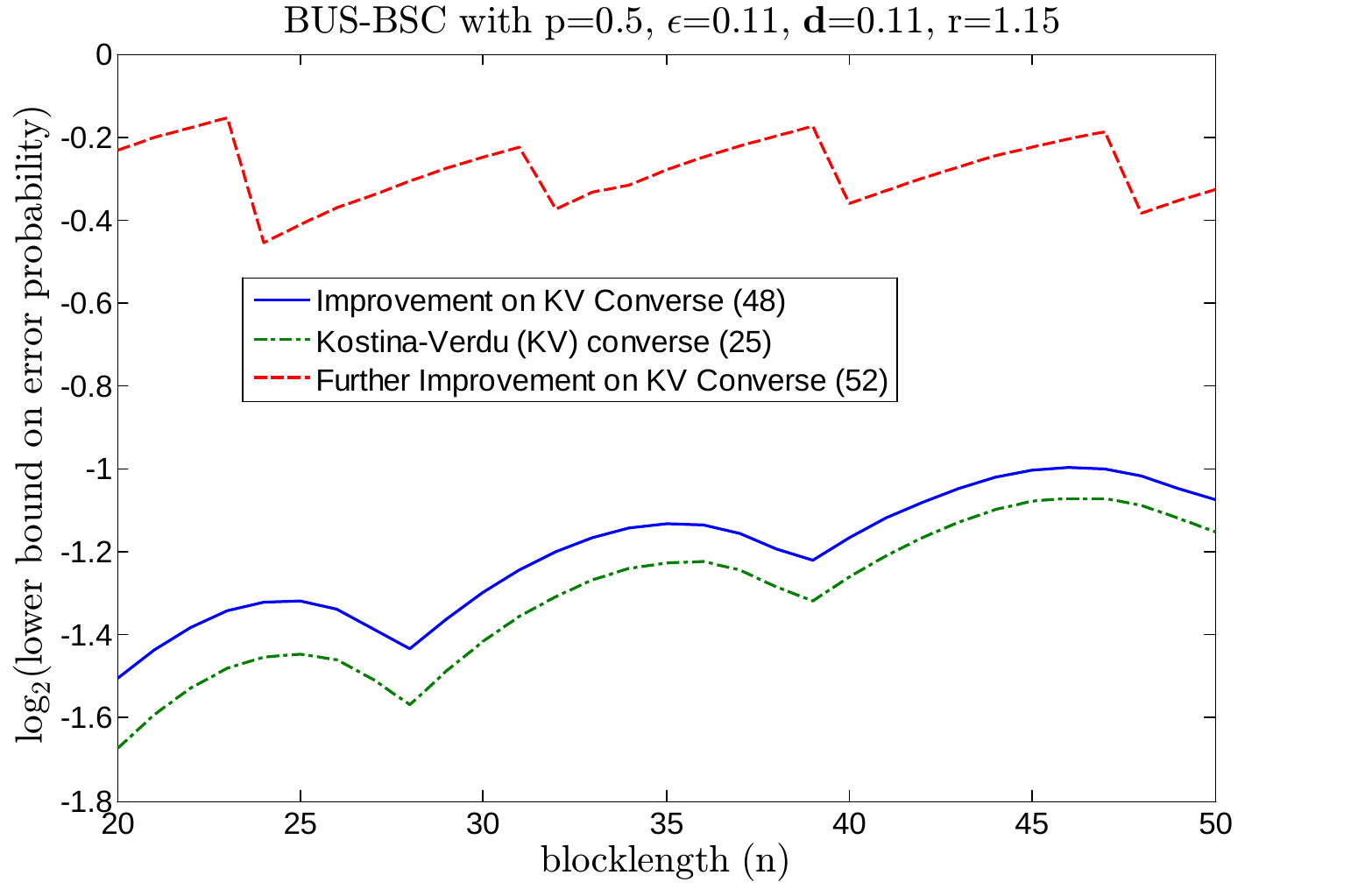} 
\caption{Comparison of the converses with blocklengths for transmitting a BUS ($p=0.5$) over a ${\rm BSC}(\epsilon=0.11)$,  $\dbf=0.11$ and rate, $r>\frac{\mathcal{C}_{BSC}}{R_S(\dbf)}$.
}\label{Fig:JSCCoutside}
\end{center} 
\vspace{-.8cm}
\end{figure}
\begin{figure}
\begin{center}
\includegraphics[scale=0.62,clip=true, trim = 1.35in 3.9in 0.5in 3.65in]{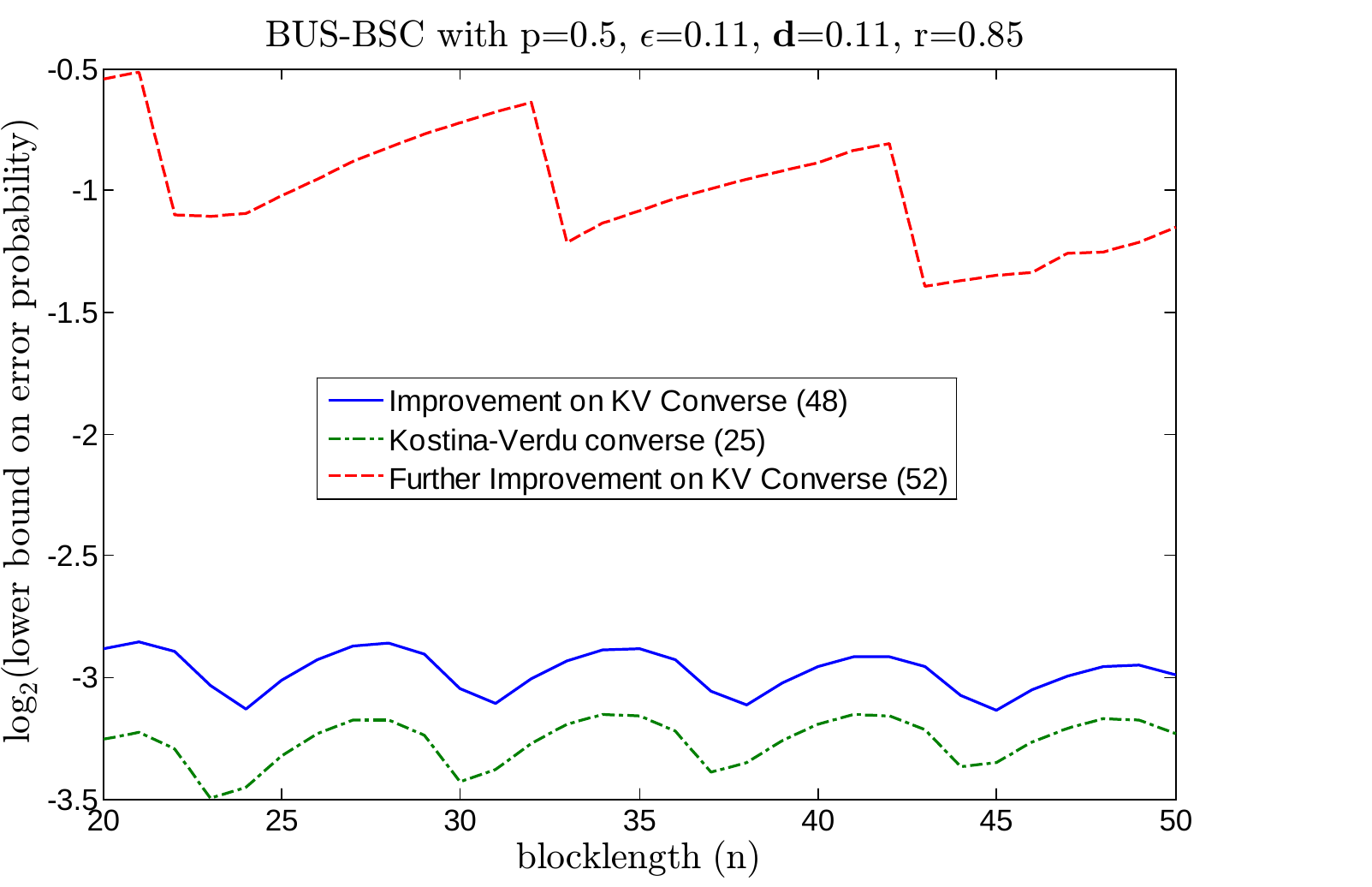} 
\caption{Comparison of the converses with blocklengths for transmitting a BUS ($p=0.5$) over a ${\rm BSC}(\epsilon=0.11)$,  $\dbf=0.11$ and rate, $r<\frac{\mathcal{C}_{BSC}}{R_S(\dbf)}$.
}\label{Fig:JSCCinside}
\end{center} 
\vspace{-.8cm}
\end{figure}
\begin{figure}
\begin{center}
\includegraphics[scale=0.62,clip=true, trim = 0.8in 4in 0.8in 3.6in]{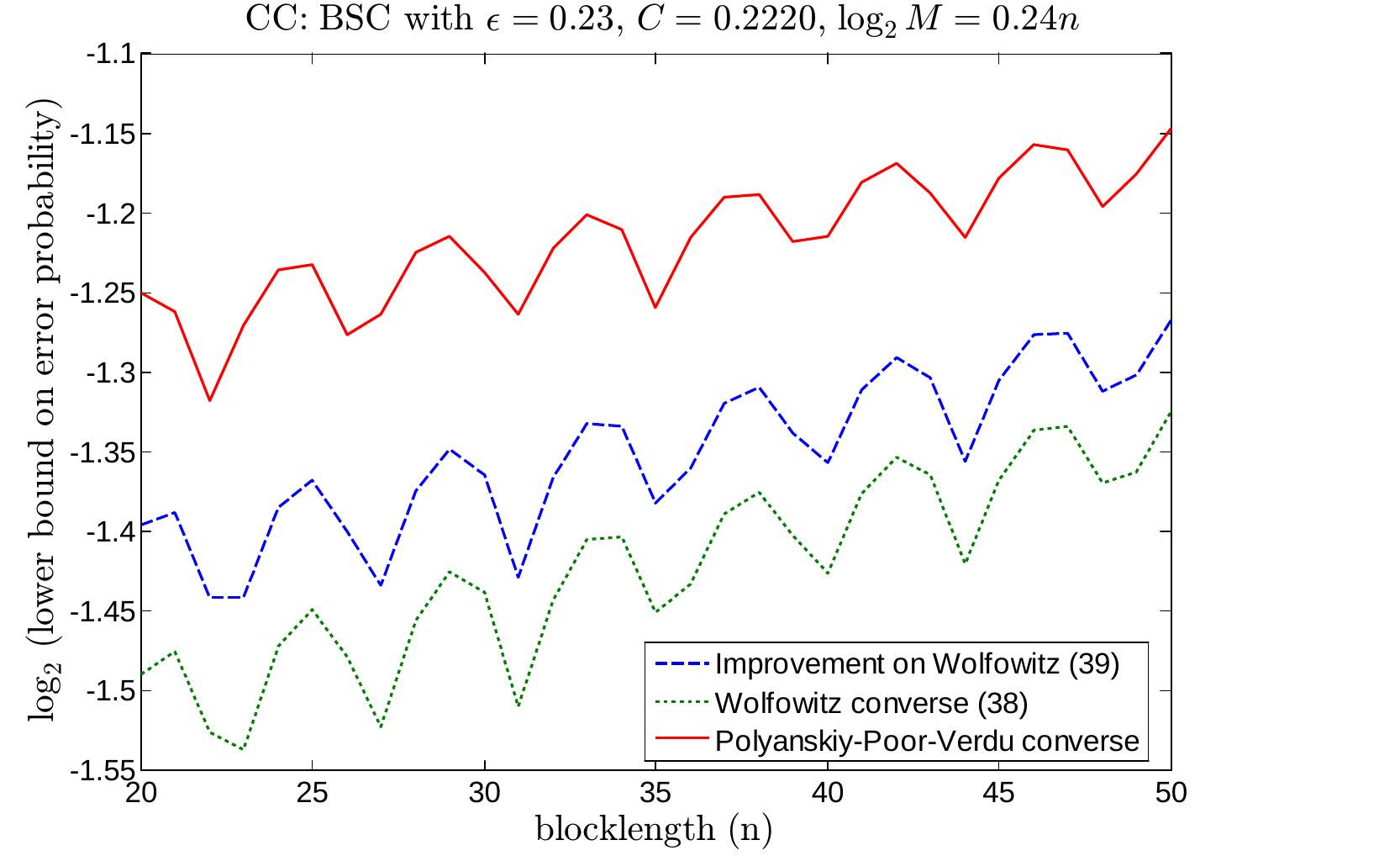} 
\caption{ Comparison of the converses for a ${\rm BSC}(\epsilon=0.23)$ with $M=2^{nR}$, where $R=0.24$ and capacity $C=0.2220$. 
}\label{Fig:channel}
\end{center} 
\vspace{-.8cm}
\end{figure}
\begin{figure}
\begin{center}
\includegraphics[scale=0.62,clip=true, trim = 1.35in 4.1in 1.0in 3.3in]{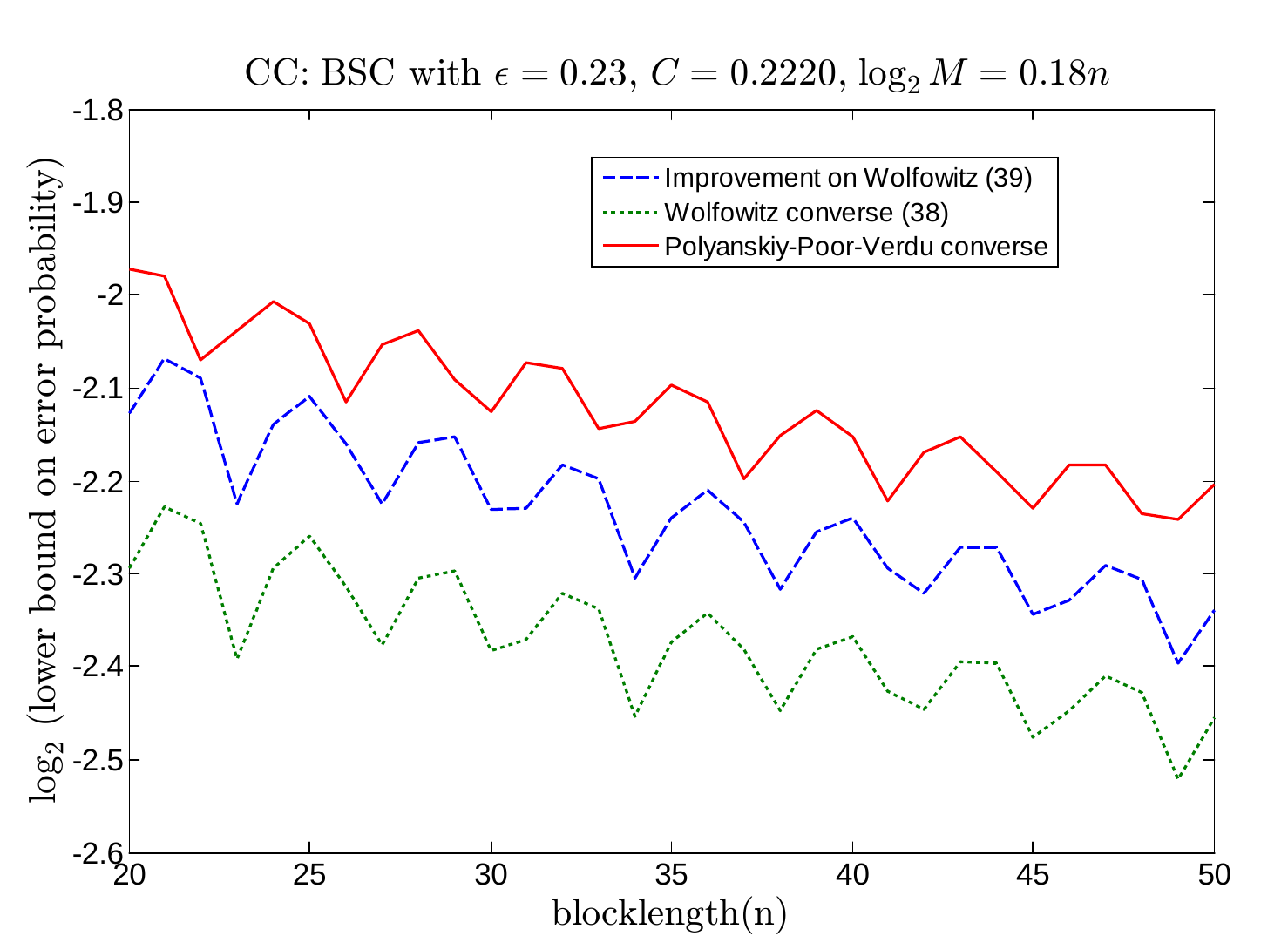} 
\caption{Comparison of the converses for a ${\rm BSC}(\epsilon=0.23)$ with $M=2^{nR}$, where $R=0.0.18$ and capacity $C=0.2220$.
}\label{Fig:channelinside}
\end{center} 
\vspace{-.8cm}
\end{figure}

We now particularize the bound in \eqref{eq:DPIgivestighterbounds} with $T=1$ to a binary memoryless source over a BSC.
\begin{corollary}\textit{(Joint Source-Channel Coding Converse of \eqref{eq:DPIgivestighterbounds} for BMS-BSC)}\label{cor:DPtightest}
Consider  the setting of Corollary~\ref{cor:BMSBSC1}. For any code, the following lower bound follows from \eqref{eq:DPIgivestighterbounds},
\begin{align}
&\Ebb\biggl[ \Ibb \biggl\lbrace \frac{1}{k}\sum_{i=1}^k \Ibb\{s_i \neq \shat_i\}>\dbf \biggr \rbrace\biggr]\geq \OPT({\SC})\geq \OPT(\rm DP)\non\\&\hspace{-0.1cm}\geq \sup_{\gamma} \biggl \lbrace \sum_{b=0}^k \Comb{k}{b}p^b(1-p)^{k-b}\biggl[1-\hspace{-0.2cm}\sum_{a=0}^{ n\epsilon+\theta -1}\hspace{-0.3cm}\Comb{n}{a}\epsilon^a(1-\epsilon)^{n-a}  \non\\&+\epsilon^{n\epsilon+n\theta}(1-\epsilon)^{n-n\epsilon-\theta}\sum_{a=0}^{n\epsilon+\theta -1}\Comb{n}{a}\biggr]-2^{-kH(\dbf)-\gamma}\sum_{a=0}^{k\dbf}\Comb{k}{a}\biggr\rbrace,\label{eq:BMSBSC2}
\end{align} where $\theta$ is as defined in Corollary~\ref{cor:BMSBSC1}.
\end{corollary}
The bound in \eqref{eq:BMSBSC2} follows from \eqref{eq:DPIgivestighterbounds} by setting $T=1$, $P_{\Ybar}(y)\equiv \frac{1}{2^n}$ and using $j_S(s,\dbf)$ from \eqref{eq:BMStilted} and $i_{X;\Ybar}(x;y)$ as defined in \eqref{eq:theta}.
In particular, when $p=0.5$, \eqref{eq:BMSBSC2} results in the following converse 
 which implies a converse obtained by Kostina-Verd\'{u} \cite[Theorem 13]{kostina2013lossy} using hypothesis testing.
\begin{corollary}\textit{(Joint Source-Channel Coding Converse of \eqref{eq:DPIgivestighterbounds} for BUS-BSC)}
Consider the setting of Corollary~\ref{cor:BMSBSC1} with $P_{S}(s)\equiv \frac{1}{\mid \Sscr \mid}\equiv \frac{1}{2^k}$.
 For any code, the following lower bound on the excess distortion probability follows from \eqref{eq:BMSBSC2},
\begin{align}
&\Ebb\biggl[ \Ibb \biggl\lbrace \frac{1}{k}\sum_{i=1}^k \Ibb\{s_i \neq \shat_i\}>\dbf \biggr \rbrace\biggr]\geq \OPT({\SC})\geq \OPT(\rm DP)\non\\& \geq \sup_{0\leq r\leq n} \biggl \lbrace 1-\sum_{a=0}^{r}\Comb{n}{a}\epsilon^a(1-\epsilon)^{n-a}+\epsilon^{r+1}(1-\epsilon)^{n-r-1}\times \non\\&\qquad \qquad \biggl[\sum_{a=0}^r \Comb{n}{a}-2^{n-k}\sum_{a=0}^{\lfloor k\dbf \rfloor}\Comb{k}{a}\biggr]\biggr \rbrace.
\label{eq:DPBUStightest} 
\end{align}
Further, if $\Ebb\biggl[ \Ibb \biggl\lbrace \frac{1}{k}\sum_{i=1}^k \Ibb\{s_i \neq \shat_i\}>\dbf \biggr \rbrace\biggr]\leq \delta$, $\delta \in (0,1)$, \begin{align*}
r^{*}=\max\biggl \lbrace r: \sum_{t=0}^r \Comb{n}{t}\epsilon^t(1-\epsilon)^{n-t}\leq 1-\delta\biggr \rbrace,
\end{align*}and $\lambda\in [0,1)$ is the solution to
\begin{align}
\sum_{t=0}^{r^{*}} \Comb{n}{t}\epsilon^t(1-\epsilon)^{n-t}+\lambda\epsilon^{r^{*}+1}(1-\epsilon)^{n-r^{*}-1}\Comb{n}{r^{*}+1}=1-\delta,\label{eq:lambdaconverse}
\end{align} then the lower bound in \eqref{eq:DPBUStightest} implies the hypothesis testing based converse of Kostina-Verd\'{u} \cite[Theorem 13]{kostina2013lossy}, namely,\begin{align}
\lambda \Comb{n}{r^{*}+1}+\sum_{k=0}^{r^{*}}\Comb{n}{j}\leq \sum_{j=0}^{\lfloor k\dbf \rfloor}\Comb{k}{j}2^{n-k}. \label{eq:BUSBSC3}
\end{align}
\end{corollary}
\begin{proof}
In the RHS of \eqref{eq:BMSBSC2}, put $p=0.5$ and rewrite $2^{-kH(\dbf)-\gamma}$ in terms of $\theta$ using \eqref{eq:theta}. A simple exercise then results in the following bound,
\begin{align*}
&\sup_{\gamma} \biggl \lbrace 1-\sum_{a=0}^{n\epsilon+\theta-1}\Comb{n}{a}\epsilon^a(1-\epsilon)^{n-a}+\epsilon^{n\epsilon+\theta}(1-\epsilon)^{n-n\epsilon-\theta}\times \non\\&\qquad \qquad \biggl[\sum_{a=0}^{n\epsilon+\theta-1} \Comb{n}{a}-2^{n-k}\sum_{a=0}^{\lfloor k\dbf \rfloor}\Comb{k}{a}\biggr]\biggr \rbrace, 
\end{align*} where $\theta$ is as defined in Corollary~\ref{cor:BMSBSC1}. Further substitute $r= n\epsilon+\theta-1$, and let $\gamma$ lie in the range, $\log_2 \biggl(\frac{1-\epsilon}{\epsilon}\biggr)+kR_S(\dbf)-n(1+\log_2(1-\epsilon))\leq\gamma\leq (n+1)\log_2 \biggl(\frac{1-\epsilon}{\epsilon}\biggr)+kR_S(\dbf)-n(1+\log_2(1-\epsilon))$. Consequently, we get the bound in \eqref{eq:DPBUStightest} with the supremum over $r \in [0,n]$.
To see that \eqref{eq:BUSBSC3} follows from \eqref{eq:DPBUStightest}, fix $r=r^{*}$ in \eqref{eq:DPBUStightest} and substitute from   \eqref{eq:lambdaconverse}.
\end{proof}
   
It can be seen from Fig~\ref{Fig:JSCCoutside} and Fig~\ref{Fig:JSCCinside} that the converse in \eqref{eq:BUSBSC3} outperforms the Kostina-Verd\'{u} converse \eqref{eq:kostinaVerdu2} and even our converse of \eqref{eq:BUSBSCimprovedkv} for small blocklengths. The difference, however, diminishes as blocklength increases. Also, notice that improvement of \eqref{eq:BUSBSCimprovedkv} over \eqref{eq:kostinaVerdu2} is significant for rates, $r <\frac{C}{R_S(\dbf)}$.

\subsection{Channel coding for a BSC}
We now come to channel coding for a BSC, numerically illustrated in Fig~\ref{Fig:channel} and Fig~\ref{Fig:channelinside}. 
Fig~\ref{Fig:channel} shows that for rates greater than capacity, our improved converse for channel coding in \eqref{eq:channelcoding} outperforms Wolfowtiz's converse \eqref{eq:wolfowitz}. For rates less than the capacity of the channel, this difference is significant for small blocklengths as can be seen in Fig~\ref{Fig:channelinside}. However, as blocklength increases, this appears to diminish. The hypothesis testing based converse of Polyanksiy, Poor and Verd\'{u} \cite{polyanskiy2010channel} outperforms the converses \eqref{eq:channelcoding} and \eqref{eq:wolfowitz}.
\subsection{Lossy Source Coding of a BMS}
In this section, we particularize the bound in Corollary~\ref{thm:lossysourcecoding} for a BMS with $\Sscr=\Sscrhat=\{0,1\}^k$  and  $P_S$ as defined in \eqref{eq:BMS} with average bit-wise Hamming distance as the distortion measure. The tilted information is as given in \eqref{eq:BMStilted}. The following lower bound follows from Corollary~\ref{thm:lossysourcecoding}.
\begin{corollary}[Source Coding Converse \eqref{eq:lossysourcecoding} for BMS]\label{cor:BMSlossy}
Consider problem SC with $\Sscr=\Sscrhat=\{0,1\}^k$, $\Xscr=\Yscr=\{1,\hdots,M\}$, $P_{S}(s)$ is as given in \eqref{eq:BMS} with bias $p$ and $P_{Y|X}(y|x)\equiv \Ibb\{x=y\}$.  Let $$\kappa(s,x,y,\shat)\equiv \Ibb \biggl\lbrace \frac{1}{k}\sum_{i=1}^k \Ibb\{s_i \neq \shat_i\}>\dbf \biggr \rbrace,$$ where $0\leq \dbf <p$. Then, for any code, the following lower bound follows from Corollary~\ref{thm:lossysourcecoding},
\begin{align}
&\Ebb\biggl[ \Ibb \biggl\lbrace \frac{1}{k}\sum_{i=1}^k \Ibb\{s_i \neq \shat_i\}>\dbf \biggr \rbrace\biggr]\geq \OPT({\SC})\geq 
\OPT(\rm DP)\non\\& \geq \sup_{\gamma}\biggl \lbrace 1-\sum_{j=0}^t\Comb{k}{j}p^j(1-p)^{k-j}+p^{t+1}
(1-p)^{k-t-1}\times\non\\&\qquad \qquad \biggl[\sum_{j=0}^t\Comb{k}{j}-M 2^{kH(\dbf)}\biggr]\biggr\rbrace,\label{eq:BMSimprovedKV}
\end{align}
where $t= kp-1+\frac{\log_2 M-k R_S(\dbf)+\gamma}{\log_2 \bigl(\frac{1-p}{p}\bigr)}$.
\end{corollary} The proof follows by a simple calculation employing \eqref{eq:BMStilted}. We skip the proof here.

As in the case of lossy joint source-channel coding, we now show that the tighter converse in Corollary~\ref{thm:lossysourcecodingimproved} also implies a converse of Kostina and Verd\'{u} for a BMS that was obtained by them by employing hypothesis testing \cite[Theorem 20]{kostina2012fixed}. 
\begin{corollary}[Source Coding Converse \eqref{eq:lossysourcecodingimproved} for BMS]
Consider the problem setup as in Corollary~\ref{cor:BMSlossy}. Consequently, the following lower bound follows from the converse in Corollary~\ref{thm:lossysourcecodingimproved},
\begin{align}
&\Ebb\biggl[ \Ibb \biggl\lbrace \frac{1}{k}\sum_{i=1}^k \Ibb\{s_i \neq \shat_i\}>\dbf \biggr \rbrace\biggr]\geq \OPT({\SC})\geq \OPT(\rm DP)\non\\& \geq \sup_{0 \leq t \leq n}\biggl \lbrace 1-\sum_{j=0}^t\Comb{k}{j}p^j(1-p)^{k-j}+p^{t+1}(1-p)^{k-t-1}\times\non\\&\qquad \qquad \biggl[\sum_{j=0}^t\Comb{k}{j}-M \sum_{j=0}^{\lfloor k\dbf \rfloor}\Comb{k}{j}\biggr]\biggr\rbrace. \label{eq:BMShypoth}
\end{align} Further, if $\Ebb\biggl[ \Ibb \biggl\lbrace \frac{1}{k}\sum_{i=1}^k \Ibb\{s_i \neq \shat_i\}>\dbf \biggr \rbrace\biggr]\leq \delta$ ,
 \begin{align*}
t^{*}=\max\biggl \lbrace t: \sum_{j=0}^t \Comb{k}{j}p^j(1-p)^{k-j}\leq 1-\delta\biggr \rbrace,
\end{align*}and $\lambda\in [0,1)$ is the solution to
$$
\sum_{j=0}^{t^{*}} \Comb{k}{j}p^j(1-p)^{k-j}+\lambda p^{t^{*}+1}(1-p)^{k-t^{*}-1}\Comb{k}{t^{*}+1}=1-\delta,$$
then the lower bound in \eqref{eq:BMShypoth} results in the converse of Kostina and Verd\'{u} \cite[Theorem 20]{kostina2012fixed},
\begin{align}
M \geq \frac{\sum_{j=0}^{t^{*}}\Comb{k}{j}+\lambda \Comb{k}{t^{*}+1}}{\sum_{j=0}^{\lfloor k\dbf \rfloor}\Comb{k}{j}}.\label{eq:BMShypothorig}
\end{align}
\end{corollary} The proof is similar to the proof of Corollary~\ref{cor:DPtightest}; we skip the proof here. 

Further, if $p=0.5$ in \eqref{eq:BMShypoth}, then, we get that
\begin{align}
&\Ebb\biggl[ \Ibb \biggl\lbrace \frac{1}{k}\sum_{i=1}^k \Ibb\{s_i \neq \shat_i\}>\dbf \biggr \rbrace\biggr]\geq \OPT({\SC})\geq \OPT(\rm DP)\non\\& \geq 1-2^{-k}M\sum_{j=0}^{\lfloor k\dbf \rfloor}\Comb{k}{j},
\end{align}which coincides with the hypothesis testing based converse of Kostina and Verd\'{u} \cite[Theorem 15]{kostina2012fixed}.
\begin{figure}
\begin{center}
\includegraphics[scale=0.62,clip=true, trim = 1.25in 3.8in 0.9in 3.8in]{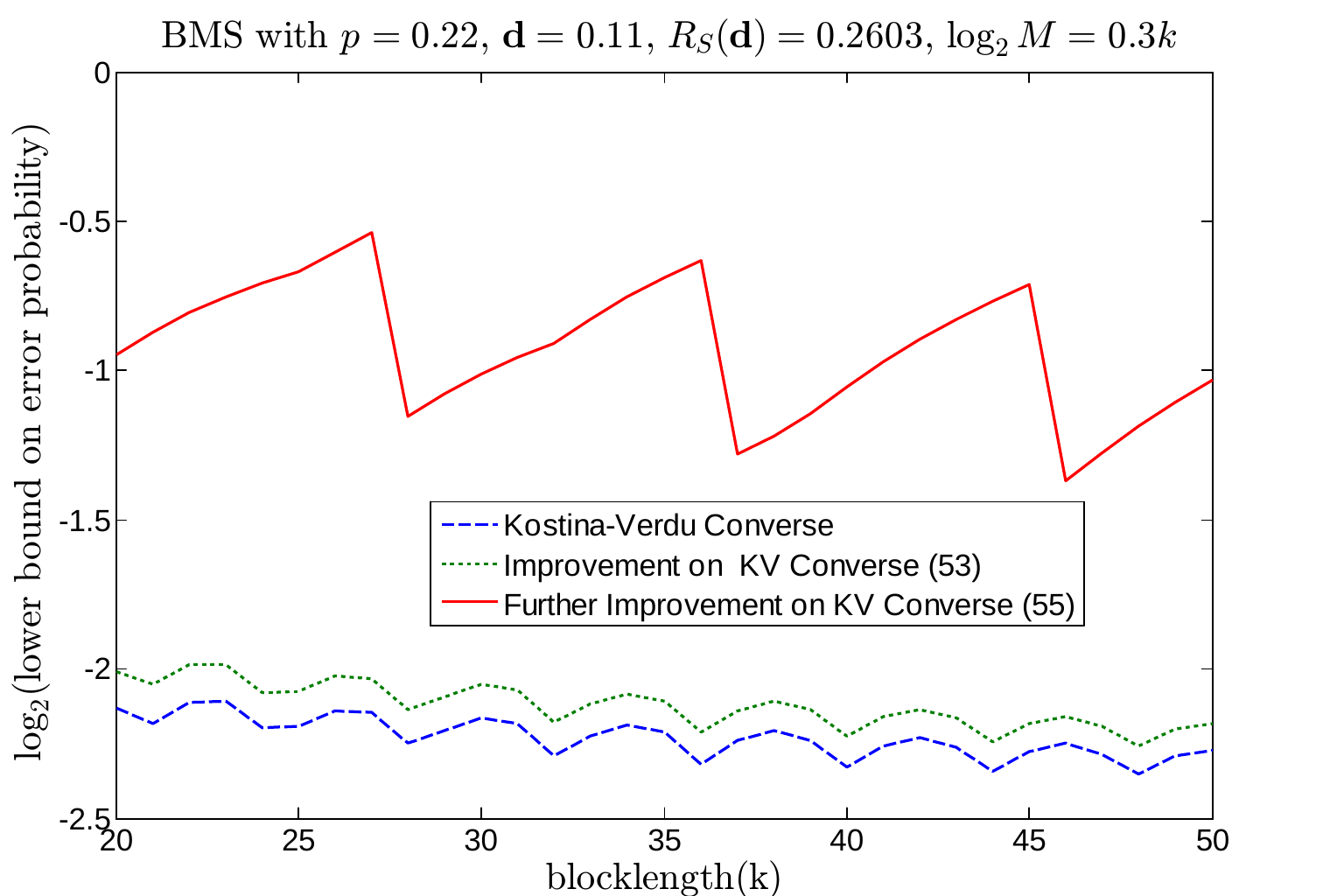} 
\caption{Comparison of performance of converses with blocklengths of a BMS with $p=0.22$, $\dbf=0.11$, $R_S(\dbf)=0.2603$ and $\log M=kR,$ $R>R_S(\dbf)$.
}\label{Fig:SSBMSout}
\end{center} 
\end{figure}
\begin{figure}
\begin{center}
\includegraphics[scale=0.62,clip=true, trim = 1.3in 3.8in 0.9in 3.7in]{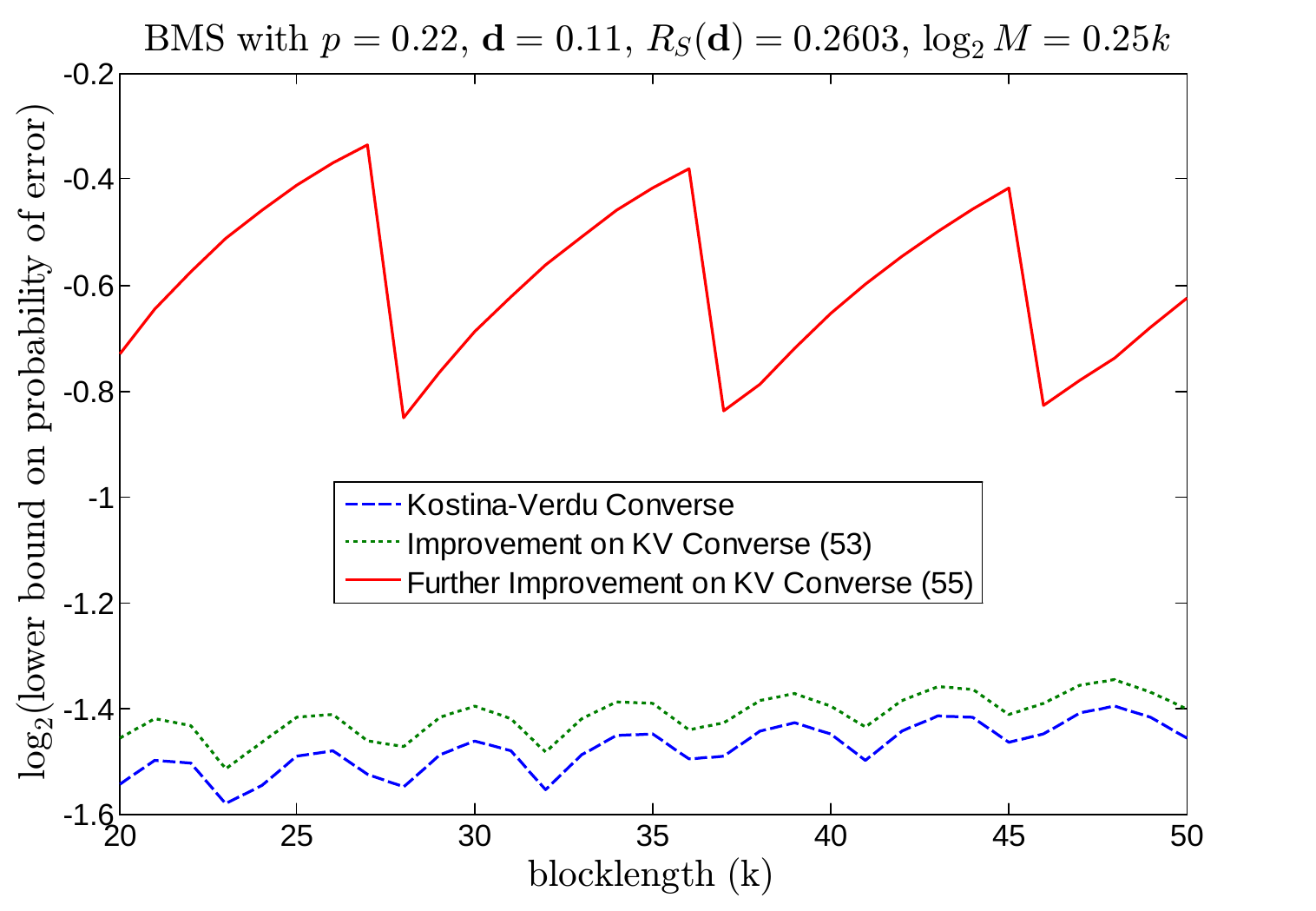} 
\caption{Comparison of performance of converses with blocklengths of a BMS with $p=0.22$, $\dbf=0.11$, $R_S(\dbf)=0.2603$ and $\log M=kR,$ $R<R_S(\dbf)$.
}\label{Fig:SSBMSins}
\end{center} 
\vspace{-.8cm}
\end{figure}

As can be seen from Fig~\ref{Fig:SSBMSins} and Fig~\ref{Fig:SSBMSout}, the converse \eqref{eq:BMSimprovedKV} outperforms Kostina-Verd\'{u} converse \cite[Theorem 7]{kostina2012fixed}.  The converse in \eqref{eq:BMShypothorig} outperforms both these converses by a large margin. 

Finally, the inter-relations between the various converses we have derived are explained in Fig~\ref{fig:interrelations}.
\subsection{To code or not to code}
Kostina and Verd\'{u} in \cite{kostina2012code} compare the minimum excess distortion achieved by joint source-channel codes with that achieved by symbol-by-symbol codes for fixed blocklength and excess distortion probability. It is shown that for transmitting a binary uniform source over a binary symmetric channel, uncoded transmission attains the minimum distortion among all rate 1 (rate, $r=\frac{k}{n}$) codes operating at  blocklength $n$ and  acheiving a fixed excess distortion probability. To show this, Kostina-Verd\'{u} resort to the lower bound in \eqref{eq:BUSBSC3}. Since our LP relaxation implies this lower bound, this  conclusion also follows from our LP relaxation. Indeed, our LP relaxation also shows that uncoded transmission is optimal for transmitting a $q$-ary uniform source over a $q$-ary symmetric channel with average symbol-wise Hamming distance as the loss criterion, as discussed in the next section.
\begin{figure*}
\tikzstyle{start} = [rectangle,minimum width=0.5cm, minimum height=0.3cm,text centered, draw=black, fill=white!]
\tikzstyle{arrow} = [thick,->,>=stealth]
\begin{small}
\begin{tikzpicture}[node distance=1cm]
\node (T1) [start] {$\OPT(\DP)$};
\node (T2)[start, below =of T1] {$\OPT({\rm DP'})$};
\node(T3) [start,right =of T2] {\small Further Improvement -- JSCC \eqref{eq:DPIgivestighterbounds}};
    \node(T4) [start,right =of T3] {\small Improvement on Gen. KV \eqref{eq:DPIimpliesgeneralkostina}};
    \node(T5) [start,right =of T4] {\small Generalized KV \eqref{eq:generalkostina}};
    \node (T6) [start,below =of T5] {\small KV  \eqref{eq:kostinaVerdu2}};
    \node(T7) [start,below =of T4] {\small Converse \eqref{eq:DPIimplieskostina}};
    \node (T8) [start,above =of T3] {\small Converse \eqref{eq:DPBUStightest} };
    \node (T9) [start,above =of T4] {\small Converse  \eqref{eq:BUSBSC3}};
    \node(T20)[start,above=of T5]{\small KV Hypothesis Test \cite[Theorem 4]{kostina2013lossy}};
    \node(T11)[start,below =of T7] {\small Improvement on Wolfowitz \eqref{eq:channelcoding}};
     \node(T19)[start, below= 2.53cm of T2]{\small PPV \cite[Theorem 27]{polyanskiy2010channel}};
   \node(T12)[start,below =of T6] {\small Wolfowitz \eqref{eq:wolfowitz}};
   \node(T13)[start,below = 4.06cm of T3]{\small Further Improvement -- SC \eqref{eq:lossysourcecodingimproved}};
    \node(T14)[start, below= of T11]{\small Improvement -- SC \eqref{eq:lossysourcecoding}};
    \node(T15)[start, below= of T12]{\small KV SC Converse \cite[Theorem 7]{kostina2012fixed}};
    \node(T16)[start, below= of T13]{\small Converse \eqref{eq:BMShypoth}};
    \node(T17)[start, below= of T14]{\small Converse \eqref{eq:BMShypothorig}};
    \node(T21)[start, below= of T15]{\small KV Hypothesis Test \cite[Theorem 8]{kostina2012fixed}};
\draw [arrow] (T1) -- (T2);
     \draw [arrow] (T2) -- node[above] {\small JSCC} (T3);
      \draw [arrow] (T3) -- node[above] {\small JSCC} (T4);
       \draw [arrow] (T4) -- node[above] {\small JSCC} (T5);
        \draw [arrow] (T5) -- node[right] {\small JSCC}(T6);
         \draw [arrow] (T7) -- node[above] {\small JSCC} (T6);
          \draw [arrow] (T4) -- node[right] {\small JSCC} (T7);
           \draw [arrow] (T3) -- node[left] {\small BUS-BSC} (T8);
           \draw [arrow] (T8) -- node[above] {\small BUS-BSC}  (T9);
          \draw [arrow] (T2) -- node [left] {\small CC} (T19);
            \draw [arrow] (T19) -- node [left] {\small CC} (T2);
              \draw [arrow] (T19) -- node [above right] {\small CC \qquad} (T11);
                \draw [arrow] (T11) -- node [left] {\small CC} (T7);
                  \draw [arrow] (T7) -- node [left] {\small CC} (T11);  
            \draw [arrow] (T11) -- node[above] {\small CC} (T12);
\path [draw,thick, ->,>=stealth]  (T3.south) -- ++(0,0cm)-- ++(0cm,0) -- node[left]{\rotatebox{90}{\small SC}} (T13.north);
             \draw [arrow] (T13) -- node[above] {\small SC} (T14);
             \draw [arrow] (T14) -- node[above] {\small SC} (T15);
              \path [draw,thick, ->,>=stealth]  (T6.east) -- ++(1.61,0cm)--++(0,-2.2cm)-|  node[below right]{\small SC} (T15);
             \draw [arrow] (T13) -- node[left] {\rotatebox{-90}\small {BMS}}   (T16);
             \draw [arrow] (T16) -- node[left] {\rotatebox{-90}\small {BMS}}   (T13);
             \draw [arrow] (T16) -- node[above] {\small BMS}  (T17);
             \draw [arrow] (T21) -- node[above] {\small BMS}  (T17);
             \draw [arrow] (T20) -- node[above= 0.2 cm] {\small BUS-BSC}(T9);
             \draw [arrow] (T12) -- node[right]{\small CC} (T6);
              \draw [arrow] (T6) -- node[right]{\small CC} (T12);
\end{tikzpicture} \end{small}
\caption{Implications of the converses derived. An arrow from $A \rightarrow B$ implies that $A \geq B$, $A \leftrightarrow  B$ implies $A=B$ and the heading above the arrow mentions the case in which the relation holds. 
 The abbreviations are: PPV  = 
 Polyanskiy-Poor-Verd\'{u}, 
 KV = Kostina-Verd\'{u}, JSSC = Joint Source-Channel Coding, SC = Source Coding and CC = Channel Coding.
}\label{fig:interrelations}
\end{figure*}
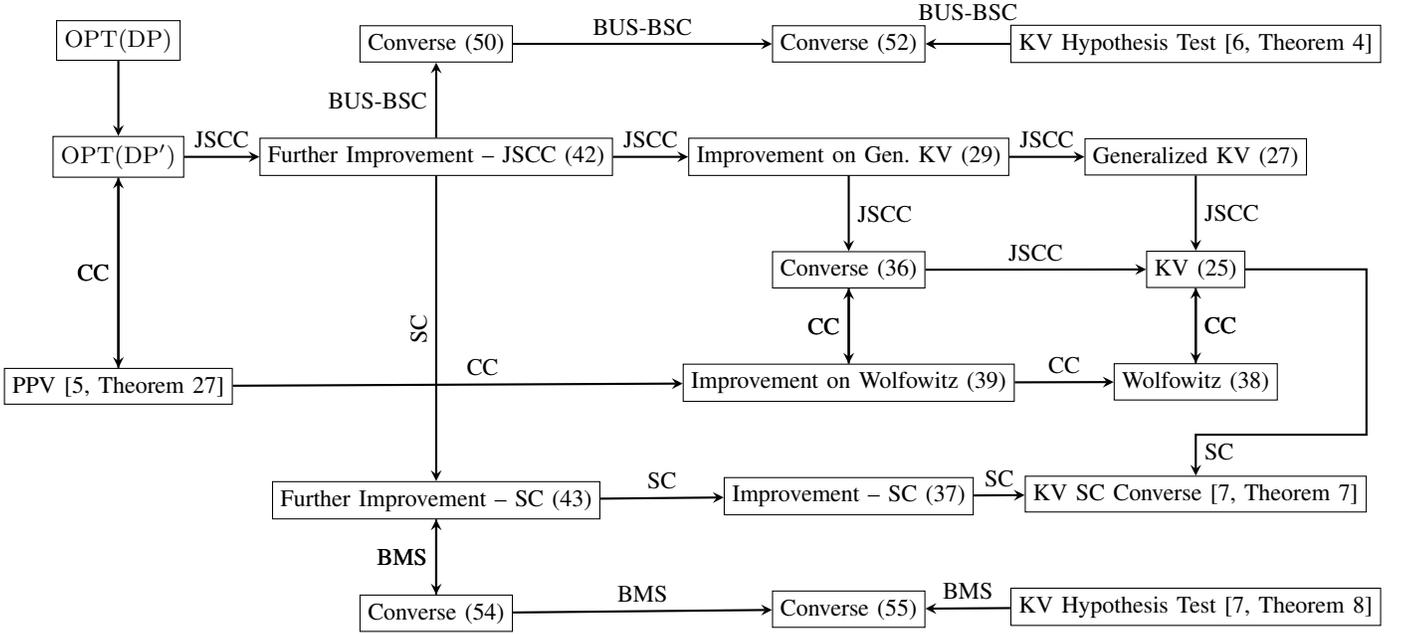

\section{Tight bound on the average symbol-wise Hamming distortion of $q$-ary uniform source over a $q$-ary symmetric channel} \label{sec:bsc}  

In this section, we obtain a general converse on the optimal cost of a finite blocklength joint source-channel code for a class of source-channel pairs and any distortion measure by constructing a feasible solution of DP. As a special case of this converse, in Corollary~\ref{cor:qaryUsSC}, we consider the problem of transmitting a $q$-ary uniform source over a $q$-ary symmetric channel (with crossover probability $\frac{\epsilon}{q-1}$, $\epsilon<1-\frac{1}{q}$) with average symbol-wise Hamming distance as the distortion measure. This source-channel pair is `probabilistically matched' in the sense of Gastpar \cite[Lemma 8]{gastpar03tocode} wherein it is shown that uncoded transmission achieves the optimal system performance in the sense of Pareto optimality of distortion and channel cost. Employing the general converse, we obtain the tight lower bound of $\epsilon$ for this problem for all blocklengths.

We consider the following class of source-channel pairs in this section. The source probability distribution takes the form $P_S(s)\equiv \bar{P}_S(\delta(s))$ where $\delta:\Sscr \rightarrow \Real$, $\bar{P}_S: \Real \rightarrow [0,1]$. The channel is given by the conditional distribution $P_{Y|X}(y|x)\equiv \bar{P}_{Y|X}(\bar{d}(x,y))$ which is defined as a function of a ``metric'' $\bar{d}:\Xscr\times \Yscr \rightarrow  [0,1]$ such that $\bar{P}_{Y|X}:[0,1] \rightarrow[0,1]$. 
These assumptions capture the cases where the probability of a block of source symbols and the probability of the channel error between a block of input and output symbols is a function of a simpler ``sufficient statistic'' such as $\delta(S)$ for the source and $\bar{d}(X,Y)$ for channel input and output. A binary memoryless source with a (i.i.d.) Bernoulli distribution corresponds to the case when $\delta(s) \equiv w_s$, the Hamming weight of $s$. The binary symmetric memoryless channel and binary memoryless erasure channel for input of blocklength $n$ are channels with $\bar{d}(x,y)\equiv \frac{1}{n}d_{x,y}$, the normalised Hamming distance (for the erasure channel, this must be defined by embedding $\Xscr$ as a subset of $\Yscr$). 
We take the loss function to be of the form $\kappa(s,x,y,\shat)\equiv d(s,\shat)$ where $d:\Sscr \times \Sscrhat \rightarrow [0,1]$ is a distortion measure. Notice that both $d$ and $\dbar$ have been normalized to take values in the unit interval. 

The following theorem gives a lower bound on the minimum expected cost of the above problem setup by constructing a feasible solution of DP.
\begin{theorem}\label{thm:generalconverse}
For problem SC, consider $S$ with probability distribution given by $P_S(s)\equiv \bar{P}_S(\delta(s))$, where $\delta: \Sscr \rightarrow \Real$ and $\bar{P}_S:\Real \rightarrow [0,1]$. The channel  conditional probability distribution is given as $P_{Y|X}(y|x) \equiv \bar{P}_{Y|X}(\bar{d}(x,y))$ where $\bar{d}:\Xscr \times \Yscr \rightarrow [0, 1]$ and $\bar{P}_{Y|X}:[0,1]\rightarrow [0,1]$. The loss function is $\kappa(s,x,y,\shat) \equiv d(s,\shat)$, where $d: \Sscr \times \Sscrhat \rightarrow [0,1]$ is a distortion measure. Then, there exists a   feasible solution of DP of the following form,
\begin{align}
\mu(s,x,y,\shat) &\equiv 0, \non\\
\lambda^a(s,\shat,y) &\equiv g_{\delta(s)}(d(s,\shat)), \non \\ 
\lambda^b(x,s,y) &\equiv -g_{\delta(s)} (\bar{d}(x,y))+\bar{d}(x,y)g_{\delta(s)}'(\bar{d}(x,y)), \label{eq:generalconversedual} \\
\gamma^a(s)&\equiv \min_x \sum_y \lambda^b(x,s,y), \non \\
\gamma^b(y)&\equiv \min_{\shat}\sum_s \lambda^a(s,\shat,y), \non
\end{align} where for each $t \in \Real$, $g_{t}: [0,1] \rightarrow \Real$ is a concave differentiable function satisfying
\begin{align} g_{t}'(m) \leq \bar{P}_S(t)\bar{P}_{Y|X}(m),\quad \forall t\in \Real,\ \forall m \in [0,1]\label{eq:derivativecondition}.\end{align} Consequently, for any code, 
\begin{align} &\Ebb\left[d(S,\Shat)\right] \geq \OPT(\rm{SC}) \geq \OPT(\DP)\nonumber\\ &\geq \sup_{g}\left[\sum_s\min_x \sum_y \left[-g_{\delta(s)}(\bar{d}(x,y))\hspace{-0.04cm}+\hspace{-0.04cm}\bar{d}(x,y)g_{\delta(s)}'(\bar{d}(x,y))\right]\right.\nonumber\\& \left. \qquad \hspace{0.2cm} +\sum_y \min_{\shat}\sum_s g_{\delta(s)}(d(s,\shat))\right],\label{eq:generalconverse}\end{align} where the supremum is taken over all concave differentiable functions $g_t, t\in \Real$ satisfying \eqref{eq:derivativecondition}.
\end{theorem}
\begin{proof} We first verify that the given set of dual variables in fact satisfy the constraints of DP. It can be directly seen that the constraints (D1) and (D2) of DP are satisfied.
We now check the feasibility of the dual variables with respect to (D3). The LHS of (D3),  $\lambda^a(s,\shat,y)+\lambda^b(x,s,y)$, 
evaluates to 
\begin{align*}
g_{\delta(s)}(d(s,\shat))-g_{\delta(s)}&(\bar{d}(x,y)) + \bar{d}(x,y)g_{\delta(s)}'(\bar{d}(x,y))\\
&\stackrel{(a)}{\leq} \left( d(s,\shat))-\bar{d}(x,y)\right)g_{\delta(s)}'(\bar{d}(x,y)) \\
 & \qquad \qquad \qquad + \bar{d}(x,y)g_{\delta(s)}'(\bar{d}(x,y))\\
&=d(s,\shat)g_{\delta(s)}'(\bar{d}(x,y))\\
& \leq d(s,\shat)\bar{P}_S(\delta(s))\bar{P}_{Y|X}(\bar{d}(x,y)),
\end{align*}which is the RHS of (D3). The inequality in $(a)$ results from the concavity of $g_{t}(\cdot)$ for each $t$. Hence, the dual variables satisfy (D3) and are thus feasible for the program DP. Consequently, for any concave differentiable function $g$, the lower bound to $\OPT(\SC)$ is given by the objective of DP, \ie, $\sum_s \gamma^a(s)+\sum_y \gamma^b(y)$. Hence, the supremum of the objective of DP over all concave differentiable functions $g_t, t\in \Real $ satisfying \eqref{eq:derivativecondition} gives the required bound.
\end{proof}

As a particular application of Theorem~\ref{thm:generalconverse} we derive a lower bound on the minimum expected average bit-wise Hamming distance of a binary source over a binary, symmetric and memoryless channel for arbitrary blocklengths $n$. 
\begin{corollary}\label{cor:BMS-BSC}
Consider the setting of Theorem~\ref{thm:generalconverse} with $\mathcal{X}=\mathcal{S}=\mathcal{Y}=\mathcal{\widehat{S}}=\Fbb_2^n$, for some $n \in \Nbb$. Let the source symbols be i.i.d. with distribution ${\rm Bern}(p)$ for some $p\in (0,1)$,  let the channel be binary, symmetric and memoryless as given in \eqref{eq:BSC}, and let the loss function be given as $\kappa(s,x,y,\shat)\equiv \frac{d_{s,\shat}}{n}$.
	Then 
 for any code,
\begin{align}
&\frac{1}{n}\Ebb\left[ \sum_{i=1}^n\I{S_i \neq \Shat_i}\right]\geq \OPT(\rm{SC}) \geq \OPT(\DP)\nonumber                                                                                                                                                                                                                                                                                                                                                                                                                                                                                                                                                                                                                                                                                                                                                                                                                                                                                                                                                                                                                                                                                                                                                                                                                                                                                                                                                                                                                                                                                                                                                                                                                                                                                                                                                   \\
&\geq \sup_g \left \lbrace \sum_{v=0}^n\binom{n}{v}\sum_{k=0}^n \binom{n}{k}\left [-g_{v}\left(\frac{k}{n}\right)+ \frac{k}{n}g_{v}'\left(\frac{k}{n}\right)\right ]+2^n \times\right. \nonumber\\& \left.  \left[ \min_{u \in \{0,\hdots,n\}} \hspace{-0.1cm} \sum_{k_1=0}^{n-u}\sum_{k_2=0}^{u} \hspace{-0.1cm}g_{k_1-k_2+u}\left(\frac{k_1+k_2}{n}\right)\binom{n-u}{k_1}\binom{u}{k_2}\right]\hspace{-0.1cm} \right \rbrace\label{eq:BMS-BSC},
\end{align}where the supremum is taken over all concave differentiable functions $g_t:[0,1] \rightarrow \Real, t\in [0,n]$ satisfying   \begin{align}g_{t}'(m) \leq p^{t}(1-p)^{n-t} \epsilon^{nm}(1-\epsilon)^{n-nm},\label{eq:concavefunction}
\end{align} for all $t \in [0,n],m \in [0,1]$.
\end{corollary}
\begin{proof}
In Theorem~\ref{thm:generalconverse}, let $\delta(s) \equiv w_s$, $\bar{P}_S(\sbar) \equiv p^{\sbar}(1-p)^{n-\sbar}$, $\bar{d}(x,y) = \frac{1}{n}d_{x,y}$ and  $\bar{P}_{Y|X}\left(m\right) \equiv\left(\epsilon^{m}(1-\epsilon)^{1-m}\right)^n$. Then, for any concave differentiable function $g$ satisfying \eqref{eq:concavefunction}, 
the dual variables in \eqref{eq:generalconversedual} are feasible for DP. It remains to show that \eqref{eq:BMS-BSC} follows from the lower bound obtained in \eqref{eq:generalconverse}. To obtain \eqref{eq:BMS-BSC}, we first evaluate the second term of \eqref{eq:generalconverse}. 

Consider any $\shat \in \Sscrhat$ with $w_{\shat}=u$ and any $s \in \Sscr$ such that $w_s=v$ and $d_{s,\shat}=k$, where $v,k,u\in \{0,\hdots,n\}$. Then, there exist integers $k_1,k_2$ ($k_1$ = number of zeros in $\shat$ changed to ones in $s$ and $k_2$ = number of ones in $\shat$ changed to zeros in $s$), where $0 \leq k_1 \leq n-u$, $0 \leq k_2 \leq u$ such that $k_1+k_2=k$ and $v=u+k_1-k_2.$ We evaluate the RHS of \eqref{eq:generalconverse}. First note that we have $$g_{w_s}\left(\frac{d_{s,\shat}}{n}\right)=g_{u+k_1-k_2}\left(\frac{k_1+k_2}{n}\right).$$ Now fix $u$. Note that for any fixed $u$, $k_1$ and $k_2$, there exist $\binom{n-u}{k_1}\binom{u}{k_2}$ number of $s$'s in $\Sscr$ with $w_s=u+k_1-k_2$ and $d_{s,\shat}=k_1+k_2.$  Thus, to evaluate $\sum_s g_{w_s}(\frac{d_{s,\shat}}{n})$, we now sum over all possible $w_s$ and $d_{s,\shat}$. Since there is a bijection between $(w_s,d_{s,\shat})$ and $(k_1,k_2)$, this amounts to summing over $k_1,k_2$. The second term in \eqref{eq:generalconverse} evaluates to
	\begin{align*}
	&\sum_{y}\min_{\shat}\sum_{s \in \Sscr} g_{w_s}\left(\frac{d_{s,\shat}}{n}\right)\\
	&=2^n \min_{u\in\{0,\hdots,n\}}\sum_{k_1=0}^{n-u}\sum_{k_2=0}^{u}g_{u+k_1-k_2}\hspace{-0.1cm}\left(\frac{k_1+k_2 }{n}\right)\hspace{-0.1cm}\binom{n-u}{k_1}\hspace{-0.1cm}\binom{u}{k_2}.
	\end{align*} 
For evaluating the first term in \eqref{eq:generalconverse}, note that the number of $y$'s in $\Yscr$ which are at a Hamming distance of $k$ from $x\in \Xscr$ (\ie, $d_{x,y}=k$) is given by $\binom{n}{k}$. Consequently, for any $x \in \Xscr$ and any $s \in \Sscr$ with $w_s=v$, \begin{align*} \sum_y -g_{v}\left(\frac{d_{x,y}}{n}\right)&+ \frac{d_{x,y}}{n}g_{v}'\left(\frac{d_{x,y}}{n}\right)\\&=\sum_{k=0}^n\binom{n}{k}\left[-g_{v}\left(\frac{k}{n}\right)+\frac{k}{n}g'_{v}\left(\frac{k}{n}\right)\right],\end{align*} which is independent of $x$. Since there exist $\binom{n}{v}$ number of $s$'s in $\Sscr$ with Hamming weight of $v$, summing the above term over $s \in \Sscr$ becomes equivalent to multiplying the above term by $\binom{n}{v}$ and summing over $v\in \{0,\hdots,n\}$, thereby resulting in the first term of the bound in \eqref{eq:BMS-BSC}. 
\end{proof}

Evidently, the RHS in \eqref{eq:generalconverse} or \eqref{eq:BMS-BSC} is not easy to evaluate. We now consider a special instance of Theorem~\ref{thm:generalconverse} where the suprema in \eqref{eq:generalconverse} and \eqref{eq:BMS-BSC} become easy to evaluate.
\begin{theorem}
Consider the setting of Theorem~\ref{thm:generalconverse} with $\Sscr=\Yscr=\Xscr=\Sscrhat=\Fbb_{q}^n$ with $P_{S}(s) \equiv \frac{1}{| \Sscr|}$ and $d=\dbar$. Further, let $d:\Xscr \times \Yscr \rightarrow \{0,\frac{1}{n},\frac{2}{n},\hdots,1\}$ be such that $|N_{k}(x)|$ is independent of $x \in \Fbb_q^n$ where $N_{k}(x):=\{y \in \Fbb_q^n \mid d(x,y)=\frac{k}{n}\}$, $k \in \{0,\hdots,n\}$ $x\in \Fbb_q^n$. Then, for any code,
\begin{align}
\Ebb\left[d(S,\Shat)\right]&\geq\OPT(\SC)\hspace{-0.1cm} \geq \OPT(\DP)\geq \sum_{k=0}^n\frac{k\Nbar_k \Pbar_{Y|X}\left(\frac{k}{n}\right)}{n},\label{eq:specialgeneralconverse}
\end{align}where $ \Nbar_k:=|N_k(x)|$.
\end{theorem}
\begin{proof} Since $P_S(s)\equiv \frac{1}{|\Sscr|}$, we let all the dual variables in Theorem~\ref{thm:generalconverse} to be independent of $\delta(s)$. Consequently, by Theorem~\ref{thm:generalconverse}, for any concave differentiable function $g$ satisfying, $g'(m) \leq \frac{ \Pbar_{Y|X}(m)}{\mid \Sscr \mid},$ for all $m \in [0,1]$, we have the following dual variables feasible for DP, $\mu(s,x,y,\shat) \equiv 0$, $\lambda^a(s,\shat,y) \equiv g(d(s,\shat))$, $\lambda^b(x,s,y) \equiv -g(d(x,y))+d(x,y)g'(d(x,y))$, $\gamma^a(s) \equiv \min_x \sum_y \lambda^b(x,s,y)$ and $\gamma^b(y) \equiv \min_{\shat}\sum_s \lambda^a(s,\shat,y)$.

We now evaluate the RHS of \eqref{eq:generalconverse}. Since $| N_{k}(x)|=\Nbar_k$ is independent of $x$, we get that for any $s,x \in \Fbb_q^n$, 
\begin{align*}
&\sum_y \left[-g(d(x,y))+d(x,y)g'(d(x,y))\right]\\&\qquad \qquad \qquad= \sum_{k=0}^n \Nbar_k\bigg[-g\left(\frac{k}{n}\right)+\frac{k}{n}g'\left(\frac{k}{n}\right)\bigg],
\end{align*} 
which is independent of $s,x \in \Fbb_q^n$. Consequently,
\begin{align*}& \sum_s\min_x \sum_y \left[-g(d(x,y))+d(x,y)g'(d(x,y))\right]\\
&\qquad\qquad \qquad=|\Sscr|\sum_{k=0}^n \Nbar_k \bigg[-g\left(\frac{k}{n}\right)+\frac{k}{n}g'\left(\frac{k}{n}\right)\bigg].\end{align*} 
Similarly, we get that 
  $$\sum_y \min_{\shat} \sum_sg(d(s,\shat))=|\Yscr|\sum_{k=0}^n \Nbar_k g\left(\frac{k}{n}\right).$$ Moreover, $|\Sscr|=| \Yscr|$. Then, from the lower bound in \eqref{eq:generalconverse}, we get that \begin{align*}
&\Ebb\left[d(S,\Shat)\right] \geq \OPT(\SC)\geq \OPT(\DP)\nonumber\\
 &\geq \sup_{g}\bigg[|\Sscr |  \sum_{k=0}^n \bigg[  \Nbar_k\hspace{-0.1cm}\left[-g\left(\frac{k}{n}\right)+\frac{k}{n}g'\left(\frac{k}{n}\right)\right]\hspace{-0.05cm}+ \hspace{-0.1cm}\Nbar_k g\left(\frac{k}{n}\right) \bigg] \bigg]\\
&= \sup_g \bigg[ |\Sscr| \sum_{k=0}^n\frac{k}{n}g'\left(\frac{k}{n}\right) \Nbar_k \bigg]\\
&\stackrel{(a)}{=}\sum_{k=0}^n\frac{k}{n}\Pbar_{Y|X}\left(\frac{k}{n}\right)\Nbar_k,
\end{align*}
where the equality in (a) follows since supremum over all concave differentiable functions imply that taking $g'(x)\equiv \frac{\Pbar_{Y|X}(x)}{| \Sscr |}$ gives the strongest bound.
\end{proof}

We now consider the problem of transmitting a $q$-ary uniform source over a $q$-ary symmetric channel which can be addressed as a special case of the above result. In this case, we have $\mathcal{X}=\mathcal{S}=\mathcal{Y}=\mathcal{\widehat{S}}=\Fbb_q^n$. The source distribution $P_S(s)=\frac{1}{\mid \Sscr\mid}$ and the channel is symmetric and memoryless, given by \begin{align}
&\hspace{-0.4cm}P_{Y|X}(y|x)=\prod_{i=1}^n P_{Y_i|X_i}(y_i|x_i),\qquad \mbox{where}\non\\&\hspace{-0.4cm}P_{Y_i|X_i}(y_i|x_i)\equiv \frac{\epsilon}{q-1}\Ibb\{y_i\neq x_i\}+ (1-\epsilon) \Ibb\{x_i=y_i\},   \label{eq:qarychannel}
\end{align}where $\epsilon < 1-\frac{1}{q}$. Let the cost function be $\kappa(s,x,y,\shat)\equiv \frac{1}{n}\sum_{i=1}^n \Ibb\{s_i \neq\shat_i\}\equiv \frac{d_{s,\shat}}{n}$. Then, from the lower bound in \eqref{eq:specialgeneralconverse}, we have the following tight lower bound.
\begin{corollary}\label{cor:qaryUsSC}
Consider problem SC with $\mathcal{X}=\mathcal{S}=\mathcal{Y}=\mathcal{\widehat{S}}=\Fbb_q^n$, where $n \in \mathbb{N }$.
Let $P_S(s)\equiv \frac{1}{|\Sscr|}$, the channel be $q$-ary, symmetric and memoryless as given in \eqref{eq:qarychannel} which can be expressed as $\Pbar_{Y|X}(\frac{d_{x,y}}{n})\equiv\left(\bigl(\frac{\epsilon}{q-1}\bigr)^{\frac{d_{x,y}}{n}}(1-\epsilon)^{1-\frac{d_{x,y}}{n}}\right)^n$ and $\kappa(s,x,y,\shat)\equiv \frac{d_{s,\shat}}{n}$. Then, for any $n,q\in \Nbb$ and $0<\epsilon<1-\frac{1}{q}$,
 \[ \epsilon \geq \OPT(\rm{SC}) \geq \OPT(\LP)=\OPT(\DP)\geq \epsilon.\] Consequently, LP is a tight relaxation of SC for any $n \in \Nbb$ and $\epsilon \in (0,1-\frac{1}{q})$.
\end{corollary}
\begin{proof}
To get the upper bound on the minimum probability of error, we consider $\SC$ with $Q_{X|S}(x|s)\equiv \Ibb\{x=s\}$ and $Q_{\Shat|Y}(\shat|y)\equiv \Ibb\{\shat=y\}$ (\ie uncoded transmission). It can be easily seen that the corresponding cost of $\SC$ is $\epsilon$. For the lower bound, using $\bar{N}_k=\Comb{n}{k}(q-1)^k$ in \eqref{eq:specialgeneralconverse}, we have $$\sum_{k=0}^n \frac{k}{n} \Comb{n}{k}(q-1)^k \biggl(\frac{\epsilon}{q-1}\biggr)^{k}(1-\epsilon)^{n-k}=\epsilon.$$ This completes the proof.
\end{proof}
In the problem setup in Corollary~\ref{cor:qaryUsSC}, if $q=2$, we have the problem of transmitting a binary uniform source over binary symmetric channel.
 It is well known that for a BUS-BSC system, uncoded transmission of signals achieve the minimum expected average bit-wise distortion of $\epsilon$ for all blocklengths. 

This concludes our main results in the point-to-point setting. We have shown that for many cases of interest, the linear programming relaxation yields tight bounds and that it leads to new converses. In the following section we consider an extension to a networked setting and derive an improvement on the converse of Zhou \etal \cite{zhou2016successive} for successively refinable source-distortion measure triplets.

\section{LP Relaxation of a Networked Problem} \label{sec:networked} 
The casting of a point-to-point communication system as an equivalent optimization problem over joint probability distributions can also be extended to the case of a network. The model we consider here is similar but not identical to the discrete memoryless multicast network in~\cite[Chapter 18]{el2011network}. We cast it as an optimization problem and derive its LP relaxation.

Consider a network on a directed acyclic graph with $N$ nodes. Suppose a random variable $Y_i$ is the information available at node $i$ and $X_i$ is the random variable to be produced according to an unknown mapping, $f_i$, \ie $X_i=f_i(Y_i)$. The spaces of $Y_i,X_i$ are $\Yscr_i$ and $\Xscr_i$ respectively and these are assumed to be fixed and finite. $Y_i$ may itself be a block of symbols; in that case we are assuming that the \textit{entire} block is available at node $i$ to produce $X_i$. Assume that information available at node $1$ is the ``source'', say $S=:Y_1$ with a given distribution. A subset of the nodes, say $\Oscr \subseteq \{1,\hdots,N\}$, are output nodes. The output at these nodes is denoted $X_j:=\Sscrhat_j, j\in \Oscr.$ 
The transformation between $X_1,\hdots,X_N$ and $Y_1,\hdots,Y_N$ is given according to a known kernel, abstractly represented as 
\[P_{Y_1,\hdots,Y_N|X_1,\hdots,X_N}.\]
This kernel may be further decomposed based on the structure of the graph. 
We assume that the network has no feedback. Specifically, if there is a directed path from node $i$ to node $j$, then $Y_i$ does not depend on the choice of $f_j.$ 
Hence, the joint distribution of all variables in the network factors as in a graphical model:
\[Q_{X_1,\hdots,X_N,Y_1,\hdots,Y_N} \equiv P_{Y_1,\hdots,Y_N|X_1,\hdots,X_N} \prod_{i=1}^N Q_{X_i|Y_i}.\]

Assume there is a loss function $\kappa(X_1, \hdots, X_N, Y_1,\hdots, Y_N)$ whose expectation we want to minimize over the functions $f_1,\hdots,f_n.$ The laws of $X_1,\hdots,X_N$ and $Y_1,\hdots, Y_N$ are fixed once $f_1,\hdots,f_N$ are fixed and the expectation is with respect to the resulting joint law.  

Let $z$ denote the tuple $(x_1,\hdots,x_N,y_1,\hdots,y_N) \in \Zscr := \prod_{i=1}^N \Xscr_i \times \prod_{i=1}^N \Yscr_i$. Once again, instead of optimizing over functions $f_1,\hdots, f_N$ we will optimize over kernels $Q_{X_i|Y_i},i=1,\hdots,N.$ This results in the following formulation.
$$
\problemsmalla{SCN}
	{Q,Q_{X_i|Y_i}, i=1,\hdots,N}
	{\displaystyle \mathbb{E}[\kappa(X_1,\hdots,X_N,Y_i,\hdots,Y_N)]}
				 {\begin{array}{r@{\ }c@{\ }l}
				 				 Q(z)= P_{Y_1,\hdots,Y_N|X_1,\hdots,X_N}\hspace{-1mm} & & \hspace{-5mm}\prod_{i=1}^N Q_{X_i|Y_i}(z), \forall z \in \Zscr,\\
\sum_{x_i} Q_{X_i|Y_i}(x_i|y_i) &=&1, \ \ \forall y_i \in \Yscr_i, \forall i,\\
Q_{X_i|Y_i}(x_i|y_i) &\geq & 0,  \, \; \, \hspace{-1mm}\forall x_i \in \Xscr_i, y_i \in \Yscr_i, \forall i.
	\end{array}}
	$$ 
The above optimization problem 	is equivalent to a multilinear polynomial optimization problem in the variables $Q_{X_i|Y_i}, i=1,\hdots,N$, 
with separable constraints. It is easy to show that a solution of this problem also lies at an extreme point of the feasible region, which in this case, corresponds to deterministic mappings $f_1,\hdots,f_N$ such that $Q_{X_i|Y_i}(x_i|y_i) \equiv \I{x_i=f_i(y_i)}.$ Hence SCN is equivalent to the minimization of $\Ebb[\kappa(X_1,\hdots,X_N,Y_1,\hdots,Y_N)]$ over $f_1,\hdots,f_N.$

When expressed abstractly the feasible region of SCN,
\begin{figure}
\centering
\includegraphics[scale=.3, clip=true, trim=.3in 4.5in .1in 3in]{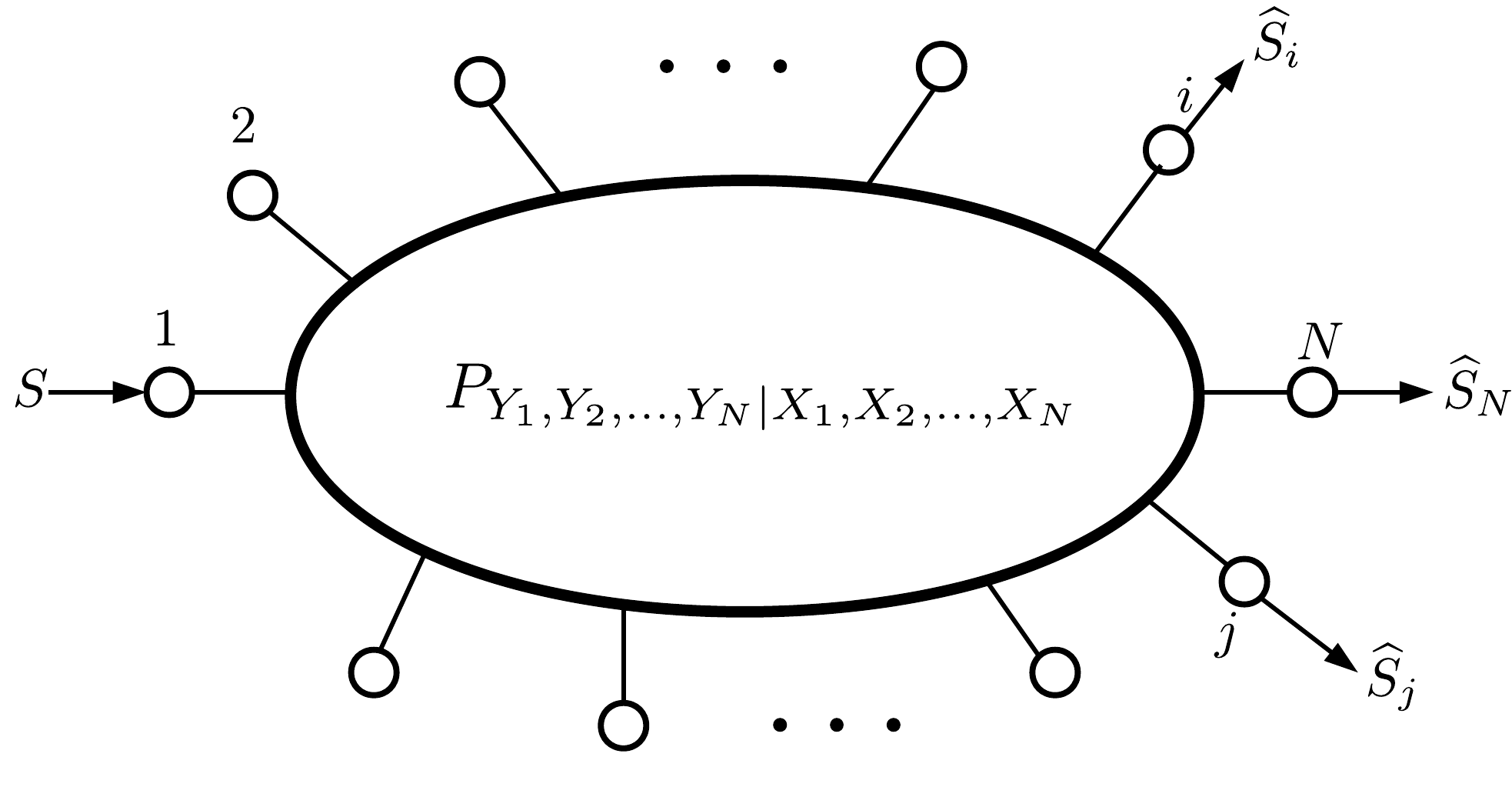}
\caption{A networked setting with source at node 1, some destination nodes and topology abstracted by the kernel $P_{Y_1,\hdots,Y_N|X_1,\hdots,X_N}.$}
\end{figure}
$\FEA(\SCN)$, takes the form,
\begin{align*}\FEA(\SCN)&=\bigg\{\hspace{-0.1cm}(v,u_1,\hdots,u_N )\hspace{-0.1cm}\mid \hspace{-0.1cm}u_i \in \hspace{-0.2cm}\prod_{y_i\in \Yscr_i} \hspace{-0.05cm}\Pscr(\Xscr_i), \hspace{0.035cm}\forall i=1,\hdots,N,\\ &\qquad     \mbox{and} \ v=\pi(u_1,u_2,\hdots,u_N)\bigg\},
\end{align*}
where $u_i = (u_i(x_i,y_i))_{x_i\in \Xscr_i,y_i \in \Yscr_i}$ denotes $Q_{X_i|Y_i}$ and  $\pi$ is a polynomial multilinear function of $(u_1,\hdots u_N)$ of degree  $N$.
Thus $u_i$ is constrained to lie in the $|\Yscr_i|$ fold product $\Pscr(\Xscr_i) \times \cdots \times \Pscr(\Xscr_i).$
To explain the formulation, recall that in the point-to-point setting (problem $\SC$), the variable $Q_{X|S}(\cdot|s)$ is constrained to be a probability distribution in $\Pscr(\Xscr)$ for each $s \in \Sscr.$ Thus $Q_{X|S}$ is constrained to be in the $|\Sscr|$-fold product $\Pscr(\Xscr) \times \cdots \times \Pscr(\Xscr).$ Similarly, $Q_{\Shat|Y}$ is constrained to be in the $|\Yscr|$-fold product $\Pscr(\Sscrhat) \times \cdots \times \Pscr(\Sscrhat).$ Thus, $\SCN$ is a multilinear polynomial optimization where each variable is constrained to lie in the Cartesian product of probability simplices.


To obtain a convex relaxation as in the case of $\SC$, we adopt the following lift-and-project approach. The approach we adopt here is a modification of the approach in \cite{sherali1992global} for polynomial optimization problems. 
Let $\bar{N}=\{1,\hdots,N\}$. Then, for any nonempty $
\mathcal{I} \subseteq \Nbar$, $\ybf = \ybf_{\Iscr} \in \prod_{i \in \Iscr } \Yscr_i$ 
and $\xbf =\xbf_{\Iscr} \in \prod_{i \in \Iscr} \Xscr_i$,  we introduce
new variables $U_{\Iscr}(\xbf, \ybf)$ to denote the product \begin{align}
\prod_{i \in \mathcal{I}} u_{i}(x_i,y_i) \mapsto U_{\Iscr}(\xbf, \ybf),  
\label{eq:newvariables}\end{align} 
where $\ybf=(y_i)_{i\in \Iscr}$ and $\xbf=(x_i)_{i\in \Iscr}.$ 
Now we obtain a set of valid inequalities of maximum degree $N$ in the space of these new variables. One set of valid inequalities is obtained by considering bounds on variables. 
We consider distinct products from combinations of upper and lower bounds, given as \begin{align}\prod_{i \in \mathcal{I}_1}(u_i(x_i,y_i)-0) \prod_{j \in \mathcal{I}_2}(1-u_{j}(x_j,y_j)) \geq 0, \label{eq:validinequality}
\end{align}for each $y_i \in \Yscr_i$, $x_i \in \Xscr_i$, $y_j\in \Yscr_j$, $x_j\in \Xscr_j$, and  $ \Iscr_1, \Iscr_2 \subseteq \bar{N}$ such that $ \mathcal{I}_1 \cap \mathcal{I}_2 =\emptyset$. Replacing the product terms with the newly introduced variables in \eqref{eq:newvariables}, we get valid inequalities in the new space of additional variables.  Notice that by construction, only products that are replaceable by variables given in \eqref{eq:newvariables} arise in \eqref{eq:validinequality}. Notice that nonnegativity of the variables $U=\{U_{\Iscr}\}_{I \subseteq \bar{N}\backslash \emptyset}$ follows from \eqref{eq:validinequality} by considering the $\Iscr=\Iscr_1$ and $\Iscr_2 =\emptyset$ in \eqref{eq:validinequality}.

We further obtain additional equality constraints as follows. For any node $i\in \bar{N}$, and any $y_i \in \Yscr_i$, we multiply the equality constraints corresponding to $u_i(\cdot,y_i)$ with the product of $u_j(x_j,y_j)$ taken over $j\in \Jscr \subseteq \bar{N}\backslash \{i\}$, for every possible $y_j\in \Yscr_j$ and $x_j \in \Xscr_j$, 
\begin{align} \prod_{j \in \Jscr}u_{j} (x_j,y_j)\sum_{x_i\in \Xscr} u_i(x_i,y_i)=\prod_{j \in \Jscr}u_{j} (x_j,y_j).\label{eq:equalityconstraintsgeneration}\end{align} By replacing the product terms with new variables defined in \eqref{eq:newvariables}, we get the following equality constraints,
\begin{align}
\sum_{x_i \in \Xscr_i}U_{\Jscr \cup \{i\}}(\xbf',\ybf')=U_{\Jscr}(\xbf,\ybf), \label{eq:equalityconstraints} 
\end{align} 
$\forall i \in \bar{N}, \ybf \in \prod_{j \in \Jscr} \Yscr_j , \xbf \in \prod_{j \in \Jscr} \Xscr_j$, and $y_i \in \Yscr_i$ with $\ybf'=(\ybf,y_i)$,  $\xbf'=(\xbf,x_i)$. Finally, we have the ``boundary condition",
\begin{equation}
U_{\{i\}}(x_i,y_i) \equiv  Q_{X_i|Y_i}(x_i|y_i), \quad \forall i\in \bar{N}.\label{eq:bdry} 
\end{equation}
Thus the LP relaxation of SCN is given by,
$$
\problemsmalla{LPN}
	{Q_{X_i|Y_i}, i=1,\hdots,N,U}
	{\displaystyle \sum_z \kappa(z)P_{Y_1,\hdots,Y_N|X_1,\hdots,X_N} U_{\bar{N}}(z)}
				 {\begin{array}{r@{\ }c@{\ }l}
				 {\rm Eq} \ \eqref{eq:validinequality},\ \eqref{eq:equalityconstraints} &\aur& \eqref{eq:bdry} \ {\rm hold}, \\ 				 			
\sum_{x_i} Q_{X_i|Y_i}(x_i|y_i) &=&1, \ \ \forall y_i \in \Yscr_i, \forall i,\\
Q_{X_i|Y_i}(x_i|y_i) &\geq & 0,  \ \; \, \forall x_i \in \Xscr_i, y_i \in \Yscr_i, \forall i.
	\end{array}}
	$$ 
Once again, any feasible point of dual of LPN gives a lower bound on SCN.

It is evident that the number of constraints in LPN (and hence number of variables in its dual) is large. Constraint  \eqref{eq:equalityconstraints} results in $N(2^{N-1}-1)$ sets of constraints and there are $3^N-1$ sets of constraints from \eqref{eq:validinequality}. This is indicative of the immense complexity of the underlying geometry of networked problems, even for moderately sized networks.

In the following section, we consider the successive refinement source coding problem with a joint excess distortion probability as the loss criterion. We show that the LP based framework applied to this networked setting, in fact improves on the converse of Zhou \etal \cite{zhou2016successive} for successively refinable source-distortion measure triplets.
\subsection{Successive Refinement Source Coding Problem}\label{subsec:successiveref}
\begin{figure}
\centering
\includegraphics[scale=.4, clip=true, trim=.5in 5.4in .1in 3in]{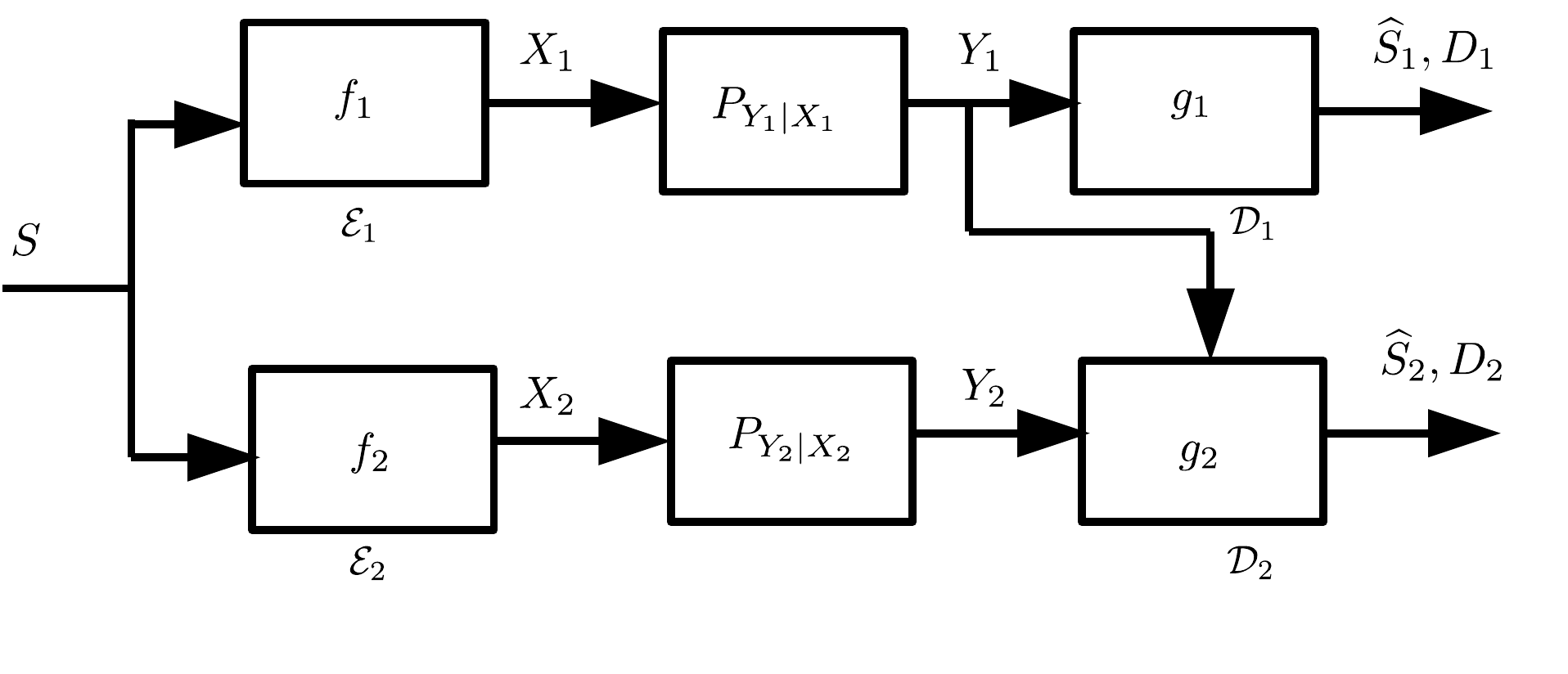}
\caption{The successive refinement source coding problem where $P_{Y_1|X_1}(y_1|x_1)\equiv \Ibb\{y_1=x_1\}$ and $P_{Y_2|X_2}(y_2|x_2)\equiv \Ibb\{y_2=x_2\}$.  }  
\label{fig:SR}
\end{figure}
Consider the successive refinement source coding problem (Fig~\ref{fig:SR}) for successively refinable  source-distortion measure triplets as defined in \cite{zhou2016successive}. Let $S,X_1,X_2,Y_1,Y_2,\Shat_1,\Shat_2$ represent random variables taking values in finite spaces, $\Sscr,\Xscr_1,\Xscr_2,\Yscr_1,\Yscr_2,\Sscrhat_1,\Sscrhat_2$, respectively.
The source message $S$ distributed according to $P_S$ is encoded by two separate encoders according to $f_1:\Sscr \rightarrow \Xscr_1$ and $f_2:\Sscr \rightarrow \Xscr_2$ to get the signals $X_1,X_2$ respectively. $X_1$ is sent through a channel with conditional distribution $P_{Y_1|X_1}$ with $Y_1$ as the output signal. Similarly, $X_2$ is sent through a channel with conditional distribution $P_{Y_2|X_2}$ to get the output signal $Y_2$. There are two decoders, $g_1:\Yscr_1 \rightarrow \Sscrhat_1$ which aim to recover the message $S$ from $Y_1$ under a distortion measure $d_1:\Sscr \times \Sscrhat_1 \rightarrow [0,\infty)$ and distortion level $D_1$, and $g_2:\Yscr_1 \times \Yscr_2 \rightarrow \Sscrhat_2$ which aims to recover $S$ from signals $Y_1$ and $Y_2$ under a distortion measure $d_2:\Sscr \times \Sscrhat_2 \rightarrow [0,\infty)$ and distortion level $D_2$. The output messages of these decoders are $\Shat_1$ and $\Shat_2$. Notice that for the successive refinement source coding problem, we assume that $P_{Y_1|X_1}(y_1|x_1)\equiv \Ibb\{y_1=x_1\}$ and $P_{Y_2|X_2}(y_2|x_2)\equiv \Ibb\{y_2=x_2\}$.
The loss criterion for the successive refinement source coding problem is taken to be the joint excess distortion probability, defined as,
$$\Ebb\bigl[ \Ibb\{d_1(S,\Shat_1)>D_1 \hspace{0.1cm}\mbox{or}\hspace{0.1cm} d_2(S,\Shat_2)>D_2\}\bigr].$$
We assume that the triplet $(S,D_1,D_2)$ is successively refinable in the sense of Zhou \etal \cite{zhou2016successive}. In this case the corresponding tilted informations are defined as,
\begin{align*}
j_{\Shat_i}(s,D_i)=\log \frac{1}{\Ebb[\exp(\lambda^{*}_iD_i-\lambda^{*}_i d_i(s,\Shat_i^{*}))]},\quad i=1,2,
\end{align*} where the expectation is with respect to the unconditional distribution $P_{\Shat_i}^{*}$ on $\Shat_i$ which achieves the infimum in the rate-distortion function, \begin{align}R_S^{(i)}(D_i)=\inf_{P_{\Shat_i|S}:\Ebb[d_i(S,\Shat_i)]\leq D_i}I(S;\Shat_i),\quad i=1,2 ,
\end{align} and $\lambda^{*}_i=-R_S^{(i)'}(D_i)$. Further, as in \eqref{eq:dtiltedproperty}, the tilted information satisfies,
\begin{align}
\Ebb[\exp(j_{\Shat_i}(S,D_i)+\lambda^{*}_iD_i-\lambda^{*}_id_i(S,\shat_i))] \leq 1, \quad \forall \shat_i \in \Sscrhat_i,  \label{eq:LtiltedSR}
\end{align} where $ i=1,2$ and the expectation is with respect to $P_S$.

 Just as in the point-to-point setting, the successive refinement source coding problem can be posed as the following optimization problem,
$$
\problemsmalla{SR}
	{f_1,f_2,g_1,g_2}
	{\displaystyle \mathbb{E}[\Ibb\{d_1(S,\Shat_1)>D_1 \hspace{0.1cm}\mbox{or}\hspace{0.1cm} d_2(S,\Shat_2)>D_2\}]}
				 {\begin{array}{r@{\ }c@{\ }l}
				 				 X_1&=&f_1(S),\qquad 
				 				 X_2=f_2(S),\\
								 \Shat_1&=&g_1(Y_1),\qquad
								\Shat_2=g_2(Y_1,Y_2).
	\end{array}}
	$$
	Consider the following joint probability distribution,
	\begin{align}Q(z)&\equiv P_{S}(s) Q_{X_1|S}(x_1|s) Q_{X_2|S}(x_2|s) P_{Y_1|X_1}(y_1|x_1)\times \non\\ &  P_{Y_2|X_2}(y_2|x_2)Q_{\widehat{S}_1|Y_1}(\shat_1|y_1) Q_{\Shat_2|Y_1,Y_2}(\shat_2|y_1,y_2) \label{eq:joint}.
	\end{align} Employing \eqref{eq:joint}, SR can be posed equivalently as the following optimization problem over randomized codes,
	 $$
\problemsmalla{SR}
	{Q,Q_{X_1|S},Q_{X_2|S},Q_{\Shat_1|Y_1}Q_{\widehat{S}_2|Y_1,Y_2}}
	{\displaystyle \sum_{z}\kappa(z) Q(z)}
				 {\begin{array}{r@{\ }c@{\ }l}
				 Q(z)&\mbox{is}& \hspace{0.1cm} \mbox{as given in \eqref{eq:joint}},\\
				 				 \sum_{x_1} Q_{X_1|S}(x_1|s)&=& 1 \quad \forall s \in \mathcal{S},\\
				 				  \sum_{x_2} Q_{X_2|S}(x_2|s)&=& 1 \quad \forall s \in \mathcal{S},\\
				 				   \sum_{\shat_1} Q_{\Shat_1|Y_1}(\shat_1|y_1)&=& 1 \quad \forall y_1 \in \Yscr_1,\\
				 				    \sum_{\shat_2} Q_{\Shat_2|Y_1,Y_2}(\shat_2|y_1,y_2)&=& 1 \quad \forall y_1 \in \Yscr_1,y_2\in \Yscr_2\\
 Q_{X_1|S_1}(x_1|s_1) &\geq& 0\quad \forall s_1 \in \mathcal{S}_1,x_1 \in \mathcal{X}_1,\\
 Q_{X_2|S_2}(x_2|s_2) &\geq& 0\quad \forall s_2 \in \mathcal{S}_2,x_2 \in \mathcal{X}_2,\\
  Q_{\widehat{S}_1|Y_1}(\shat_1|y_1) &\geq& 0 \quad \forall \shat_1 \in \mathcal{\widehat{S}}_1,y_1 \in \mathcal{Y}_1,\\
  Q_{\widehat{S}_2|Y_1,Y_2}(\shat_2|y_1,y_2) &\geq& 0 \quad \forall \shat_2 ,y_1 ,y_2,
	\end{array}}
	$$ where $\Zscr:=\Sscr \times \Xscr_1 \times \Xscr_2 \times \Yscr_1 \times \Yscr_2 \times \Sscrhat_1 \times \Sscrhat_2$ and $z \in \Zscr$,  $\kappa(z)\equiv \Ibb\{d_1(s,\shat_1)>D_1 \hspace{0.2cm}\mbox{or}\hspace{0.2cm}d_2(s,\shat_2)>D_2\}$.
	To obtain the LP relaxation of the problem SR, we resort to the approach described in the first part of this section. However, since we have two encoders and two decoders, the number of additional valid constraints the approach generates is too large. Thus, for the sake of analytical ease, we omit the inequality constraints generated according to \eqref{eq:validinequality} and include only few of the equality constraints generated from \eqref{eq:equalityconstraints}. 
	
	%
Let $\Nbar=\{\mathcal{E}_1,\mathcal{E}_2,\mathcal{D}_1,\mathcal{D}_2\}$ represent the set of nodes where $\mathcal{E}_1,\mathcal{E}_2,\mathcal{D}_1,\mathcal{D}_2$ are as shown in Fig~\ref{fig:SR} with the corresponding probability distributions, $u_{\mathcal{E}_1}=Q_{X_1|S}$, $u_{\mathcal{E}_2}=Q_{X_2|S}$, $u_{\mathcal{D}_1}=Q_{\Shat_1|Y_1}$ and $u_{\mathcal{D}_2}=Q_{\Shat_2|Y_1,Y_2}$. To obtain a converse for the successive refinement problem, it is enough to consider the constraints generated as explained below. In \eqref{eq:equalityconstraintsgeneration},  for $i=\mathcal{E}_1$, take $\mathcal{J}=\{\mathcal{E}_2,\mathcal{D}_1,\mathcal{D}_2\}, \Jscr = \{\mathcal{E}_2\}$ and $\Jscr = \{\mathcal{D}_2\}$ to generate new valid inequalities. Similarly,   for $i=\mathcal{E}_2$, take $\mathcal{J}=\{\mathcal{E}_1,\mathcal{D}_1,\mathcal{D}_2\}$ and $\Jscr=\{\mathcal{D}_1\}$; for $i=\mathcal{D}_1$, take $\mathcal{J}=\{\mathcal{E}_1,\mathcal{E}_2\}$, $\Jscr = \{\mathcal{E}_1,\mathcal{D}_2\}$ and finally for $i=\mathcal{D}_2$, take $\mathcal{J}=\{\mathcal{E}_1,\mathcal{E}_2,\mathcal{D}_1\}$ and $\{\mathcal{E}_2,\mathcal{D}_1\}$.
	Linearizing these constraints by introducing new variables and adding them to SR results in the LP relaxation LPSR of problem SR. LPSR and its dual DPSR are included in Appendix~\ref{app:SR}.
	
	The following theorem shows that LPSR in fact improves on the converse of Zhou \etal \cite[Lemma 9]{zhou2016successive} for successively refinable source-distortion measure triplets. The resulting improved converse in fact generalizes our improvement on the Kostina-Verd\'{u} converse for lossy source coding in Corollary~\ref{thm:lossysourcecoding} to the network setting. The proof is included in Appendix~\ref{app:SR}.
\begin{theorem}\label{thm:SR}
Consider problem SR where $S$ has the distribution $P_S$ and $\Xscr_1=\Yscr_1=\{1,\hdots,M_1\}$, $\Xscr_2=\Yscr_1=\{1,\hdots,M_2\}$. The channel conditional probabilities are $P_{Y_1|X_1}(y_1|x_1)\equiv \Ibb\{(y_1=x_1\}$ and $P_{Y_2|X_2}(y_2|x_2)\equiv \Ibb\{(y_2=x_2\}$.  The loss function is given as $\kappa(s,x,y,\shat)\equiv \Ibb\{d_1(s,\shat_1)>D_1 \hspace{0.2cm}\mbox{or}\hspace{0.2cm}d_2(s,\shat_2)>D_2\}$ where $d_1:\Sscr \times \Sscrhat_1 \rightarrow [0,\infty)$ and $d_2:\Sscr \times \Sscrhat_2 \rightarrow [0,\infty)$ are the distortion measures and  $D_1,D_2 \in [0,\infty)$ are the distortion levels. Let $(S,D_1,D_2)$ be successively refinable.
Then, for any code, the joint excess distortion probability satisfies, 
\begin{align}
&\Ebb[\Ibb\{d_1(S,\Shat_1)>D_1 \hspace{0.2cm}\mbox{or}\hspace{0.2cm}d_2(S,\Shat_2)>D_2\}]\geq \OPT(\rm{DPSR})\nonumber\\ & \geq  \sup_{\gamma_1,\gamma_2} \biggl \lbrace -\exp(-\gamma_1)+ \Pbb\left[j_{\Shat_1}(S,D_1)\geq \log M_1+ \gamma_1\hspace{0.1cm} \mbox{or} \nonumber\right.\\&\left.j_{\Shat_2}(S,D_2)\geq \log M_2+ \gamma_2\right]-\exp(-\gamma_2)+\sum_{s}\frac{P_S(s)}{2M_1}\times \non \\&\biggl  [\biggl(\exp(j_{\Shat_1}(s,D_1)-\gamma_1)+ \frac{1}{M_2}\exp(j_{\Shat_2}(s,D_2)-\gamma_2) \biggr)\times\non\\& \Ibb\{j_{\Shat_1}(s,D_1)\hspace{-0.1cm}< \hspace{-0.1cm}\log M_1+ \gamma_1, j_{\Shat_2}(s,D_2)< \hspace{-0.1cm}\log M_1M_2+ \gamma_2\}\biggr]\biggr \rbrace. \label{eq:zhouimprov} 
\end{align}
\end{theorem}
The improvement over \cite[Lemma 9]{zhou2016successive} is on account of the nonnegativity of the term in the second square bracket in \eqref{eq:zhouimprov}.

We end with a few observations on using this approach for networked problems. The LP relaxation of networked problem has in general a large number of constraints. Consequently, the number of variables in the dual of the relaxation is also large. This implies, both, greater flexibility in choosing values for dual variables to yield converses, and greater difficulty in ascertaining specific values for these dual variables.
\section{Conclusion}\label{conclusion}
This paper has presented a linear programming based approach to derive converses on finite blocklength joint source-channel problems. The finite blocklength joint source-channel coding was posed as an optimization problem over joint probability distributions. The resulting problem is nonconvex for which we presented a convex relaxation using LP relaxation.  Lower bounds on the optimal cost of the finite blocklength problem were obtained by constructing dual feasible points. We show that this approach recovers and improves on the known converses of Kostina and Verd\'{u} and implies the converse of Polyanskiy, Poor and Verd\'{u}. Further, we also derive a new general converse for finite blocklength joint source-channel coding which shows  that the LP relaxation gives tight lower bounds to the minimum expected average symbol-wise Hamming distortion of a $q$-ary uniform source over a $q$-ary symmetric channel for all blocklengths. We also discuss the extension of the relaxation to networked settings and show that the relaxation in fact improves on the converse of Zhou \etal \cite{zhou2016successive} for successively refinable source-distortion measure triplets.

\section*{Acknowledgments}
The authors thank the two anonymous reviewers for their comments, particularly a reviewer who made us aware of the work of Matthews~\cite{matthews2012linear}. Addressing these comments has led to a substantial improvement in this paper over its previous version.
\section{Appendices}\label{sec:appendices}
\appendices

\section{Proof of Theorems in Section~\ref{sec:optimisiation}}\label{app:sec3}
\begin{proofarg}[of Theorem~\ref{thm:extremepointsSC}]
We first show that if  $Q_{X|S}$, $Q_{\Shat|Y}$ are deterministic, then $(Q,Q_{X|S},Q_{\Shat|Y})$ $ \in \ext(\FEA(\SC))$.  Assume the contrary, \ie, $(Q, Q_{X|S},Q_{\Shat|Y}) \not \in \ext(\FEA(\SC))$.
  Then, there exist distinct vectors $(Q^1, Q_{X|S}^1,Q_{\Shat|Y}^1)$, $(Q^2, Q_{X|S}^2,Q_{\Shat|Y}^2)$  and $\alpha \in (0,1)$ such that
\begin{small}
 \begin{align*}
 &(Q(z), Q_{X|S}(x|s),Q_{\Shat|Y}(\shat|y)) =\alpha(Q^1(z), Q_{X|S}^1(x|s),Q_{\Shat|Y}^1(\shat|y)) \\&\quad +(1-\alpha)(Q^2(z), Q_{X|S}^2(x|s),Q_{\Shat|Y}^2(\shat|y)), \quad  \forall z.
 \end{align*} \end{small}
However, since $Q_{X|S}$, $Q_{\Shat|Y}$ are deterministic, they are extreme points of the space of marginal distributions $\Pscr(X|S)$ and $\Pscr(\Sscrhat| \Yscr)$ respectively.  Consequently, $Q_{X|S}^1=Q_{X|S}^2=Q_{X|S}$ and $Q_{\Shat|Y}^1=Q_{\Shat|Y}^2=Q_{\Shat|Y}$. Now, consider $Q(z)=\alpha Q^1(z)+(1-\alpha)Q^2(z)$, for all $z$, which evaluates to
\begin{align}&Q_{X|S}(x|s)Q_{\Shat|Y}(\shat|y)=\alpha Q^1_{X|S}(x|s)Q^1_{\Shat|Y}(\shat|y)\nonumber\\& \qquad  \qquad +(1-\alpha)Q^2_{X|S}(x|s)Q^2_{\Shat|Y}(\shat|y) \quad \forall s,x,y,\shat.\label{eq:deterministic}
\end{align} However, since $Q_{X|S}^1=Q_{X|S}^2=Q_{X|S}$ and $Q_{\Shat|Y}^1=Q_{\Shat|Y}^2=Q_{\Shat|Y}$, \eqref{eq:deterministic}
  implies $Q^{1}(z)=Q^{2}(z)=Q(z),$ for all $z$ and thereby contradicts the assumption that $(Q^1, Q_{X|S}^1,Q_{\Shat|Y}^1) \neq(Q^2, Q_{X|S}^2,Q_{\Shat|Y}^2)$. Hence $(Q, Q_{X|S},Q_{\Shat|Y}) \in \ext(\FEA(\SC))$.

We now prove the converse, \ie,
 if $(Q, Q_{X|S},Q_{\Shat|Y}) \in \ext(\FEA(\SC))$, then $Q_{X|S}$, $Q_{\Shat|Y}$ are deterministic. Assume the contrary. Let there exist atleast one, say, $s=s^{*}$ (or $y=y^{*}$) such that $Q_{X|S}(x|s^{*})$ (or, $Q_{\Shat|Y}(\shat|y^{*})$) is not deterministic and $(Q, Q_{X|S},Q_{\Shat|Y}) \in \ext(\FEA(\SC))$.
Since $Q_{X|S}(x|s^{*})$ is not deterministic, there exist $\Qbar_{X|S}(x|s^{*})$, $\Qhat_{X|S}(x|s^{*}) \in \Pscr(\Xscr|\Sscr)$ and $\alpha \in (0,1)$ such that $$Q_{X|S}(x|s^{*})= \alpha \Qbar_{X|S}(x|s^{*})+(1- \alpha) \Qhat_{X|S}(x|s^{*}),\quad \forall x.$$ Multiplying both sides by $P_{S}(s^{*})P_{Y|X}(y|x)Q_{\Shat|Y}(\shat|y)$  we get,
$$
Q(s^{*},x,y,\shat) \equiv \alpha \Qbar(s^{*},x,y,\shat)+(1-\alpha)\Qhat(s^{*},x,y,\shat).
$$ 
For $s \neq s^{*}$, we take $\Qbar_{X|S}(x|s)=\Qhat_{X|S}(x|s)=Q_{X|S}(x|s),$ for all $x$. Consequently, for $s\neq s^{*}$,
$$\Qbar(s ,x,y,\shat)= \Qhat(s ,x,y,\shat)=Q(s,x,y,\shat), $$ for all $x,y,\shat$. Thus, we get the following new vector which can be written as a convex combination of two different vectors,
 \begin{small} \begin{eqnarray*}
 \left( \begin{array}{c}
Q_{X|S}(x|s^{*})\\
Q_{X|S}(x|s \neq s^{*})\\
Q_{\Shat|Y}(\shat|y)\\
Q(s^{*},x,y,\shat)\\
Q(s \neq s^{*},x,y,\shat)\\
 \end{array} \right) \equiv \alpha & \left( \begin{array}{c}
\Qbar_{X|S}(x|s^{*})\\
Q_{X|S}(x|s \neq s^{*})\\
Q_{\Shat|Y}(\shat|y)\\
\Qbar(s^{*},x,y,\shat)\\
Q(s \neq s^{*},x,y,\shat)\\
 \end{array} \right)
 + \\ \qquad (1-\alpha) &\left( \begin{array}{c}
\Qhat_{X|S}(x|s^{*})\\
Q_{X|S}(x|s \neq s^{*})\\
Q_{\Shat|Y}(\shat|y)\\
\Qhat(s^{*},x,y,\shat)\\
Q(s \neq s^{*},x,y,\shat)\\ \end{array} \right).
 \end{eqnarray*} \end{small} 
This implies that $(Q, Q_{X|S},Q_{\Shat|Y}) \not \in \ext(\FEA(\SC))$ when there exists atleast one $s=s^{*}$ corresponding to which $Q_{X|S}(x|s^{*})$ is not deterministic.
  \end{proofarg}
\begin{proofarg}[of Theorem~\ref{thm:lpisrelaxation}]Consider $(Q_{X|S},Q_{\widehat{S}|Y},Q) \in \FEA({\SC})$ and let $W(z)=Q_{X|S}Q_{\widehat{S}|Y}(z)$ for all $z$. Then, $Q(z)\equiv P_SP_{Y|X} W(z)$.
Since $(Q_{X|S},Q_{\widehat{S}|Y},Q) \in \FEA({\SC})$, 
the constraints of SC also hold for LP. 
Clearly, $W(s,x,y,\shat) \geq 0$ for all $s,x,y,\shat$. Also, $\sum_x W(z)= \sum_x Q_{X|S}(x|s)Q_{\widehat{S}|Y}(\shat|y)= Q_{\widehat{S}|Y}(\shat|y)$ for all $s,\shat,y$. Similarly, $\sum_{\shat} W(z)= Q_{X|S}(x|s)$ for all $x,s,y$.
It is now sufficient to show that $-1+Q_{X|S}(x|s)+Q_{\widehat{S}|Y}(\shat|y)-W(z)\leq 0 $ for all $z$.
We have, \begin{align*}
&-1+Q_{X|S}(x|s)+Q_{\widehat{S}|Y}(\shat|y)-W(z)\\ =& -1+Q_{X|S}(x|s)+Q_{\widehat{S}|Y}(\shat|y)-Q_{X|S}(x|s)Q_{\widehat{S}|Y}(\shat|y)\\
=&(1-Q_{\widehat{S}|Y}(\shat|y))(Q_{X|S}(x|s)-1) \leq 0.
\end{align*}
Thus, if $(Q_{X|S},Q_{\widehat{S}|Y},Q) \in \FEA({\SC})$, then 
$(Q_{X|S},Q_{\widehat{S}|Y},W) \in \FEA(\LP)$ where  $W(z) \equiv Q_{X|S} Q_{\widehat{S}|Y}(z)$. This proves the claim.
%
\end{proofarg}

\begin{theorem}\label{thm:fDPI}
Consider the following optimization problem $\rm{DPI_{\fbf}}$, $$ \problemsmalla{$\rm{DPI_{\fbf}}$}
	{ a(\shat|s)}
	{\displaystyle \sum_{s,\shat}d(s,\shat)P_S(s)a(\shat|s)}
				 {\begin{array}{r@{\ }c@{\ }l}
				 				 \sum_{\shat} a(\shat|s)&=& 1 \hspace{0.05cm} : \eta(s) \qquad   \forall s\\
I_{\fbf}(aP_S)-C_{\fbf}&\leq &0 \hspace{0.05cm} :\lambda
 \\ 
%
 a(\shat|s) &\geq& 0 \hspace{0.05cm} : \psi(\shat|s) \quad \hspace{0.05cm} \forall \shat,s,\\
	\end{array}}  $$ where $I_{\fbf}(aP_S)$ is the mutual information of $S,\Shat$ under the distribution $a(\shat|s)P_S(s)$ and $\eta,\lambda,\psi$ are the Lagrange multipliers. Then, if at optimality we have $\lambda>0$, a minimizer $a^{*}$ of $\rm{DPI_{\fbf}}$ satisfies
	\[R_{\fbf}(d^{*})= I_{\fbf}(a^{*}P_S) =C_{\fbf},\]
where $d^{*}=\sum_{s,\shat}d(s,\shat)P_S(s)a^{*}(\shat|s),$ and $R_{\fbf}$ and $C_{\fbf}$ are defined in \eqref{eq:rf} and \eqref{eq:cf}, respectively.	 
\end{theorem}
\begin{proof}
We first note that since the objective of the $\rm {DPI_{\fbf}}$ can be expressed using $a$ alone, we drop variable $Q$ from the original formulation of  $\rm {DPI_{\fbf}}$. 
Let $a^{*}$ be a minimizer of $\rm{DPI_{\fbf}}$. Then, there exist Lagrange multipliers $\eta,\lambda,\psi$ satisfying the following KKT conditions:
\begin{align}
-P_S(s)d(s,\shat)&=\eta(s)+\lambda \frac{\partial I_{\fbf}(aP_S)}{\partial a}\biggr \rvert_{a^{*}}-\psi(\shat|s), \non \\
\lambda &\geq 0, \quad \lambda (I_{\fbf}(a^{*}P_S)-C_{\fbf})=0,\label{eq:lagrangianDPI}\\
\psi(\shat|s) &\geq 0,\quad \psi(\shat|s)a^{*}(\shat|s)\equiv 0.\nonumber 
\end{align} Now, consider the following optimization problem: 
	\begin{align*} ({\rm RD}) \qquad \min_{a(\shat|s) \in \Pscr(\Shat|S):\Ebb[d(S,\Shat)] \leq \dtilde} I_{\fbf}(aP_S). \end{align*} Notice that by the property of $\fbf$-mutual information, for fixed $P_S$, $I_{\fbf}(aP_S)$ is a convex function of $a$ whereby $\DPI_\fbf$ and RD are both convex optimization problems. Comparing the KKT conditions of these problems, it is easy to see that when $\lambda\in (0,+\infty)$, a minimizer $a^{*}$ of $\rm{DPI_{\fbf}}$ also solves RD with $\widetilde{d}=d^*$. 
	Thus, a minimizer  $a^{*}$ of $\rm{DPI_{\fbf}}$ satisfies,
\begin{align*}
	R_{\fbf}(d^{*})&=I_{\fbf}(a^{*}P_S),\\
	&\stackrel{(a)}{=}C_{\fbf}\hspace{0.2cm} \mbox{when $a^{*}$ solves $\rm{DPI_{\fbf}}$},	
	\end{align*}where equality in $(a)$ follows from \eqref{eq:lagrangianDPI} as $\lambda>0$.
\end{proof}
\begin{proofarg}[of Lemma~\ref{lem:extremepointinclusion}] Let $(Q_{X|S}^{*},Q_{\widehat{S}|Y}^{*},Q^{*}) \in \ext(\FEA({\SC}))$ and $W^{*}(z)\equiv Q_{X|S}^{*}  Q_{\widehat{S}|Y}^{*} (z)$.
We know from Theorem~\ref{thm:extremepointsSC} that $Q_{X|S}^{*},Q_{\widehat{S}|Y}^{*} $ are deterministic. This  implies that $W^{*}$ is deterministic too. \ie,    \begin{align*} W^{*
}(z)=Q_{X|S}^{*}Q_{\widehat{S}|Y}^{*}(z) = \I{x=f(s)\ {\rm and }\ \shat=g(y) } \quad \forall z.
\end{align*} 
 Assume that the vector $( Q_{X|S}^{*},Q_{\widehat{S}|Y}^{*}, W^{*})$ does not constitute an extreme point of LP. Then, it can be written as a convex combination of two distinct vectors $(Q^{1}_{X|S}, Q^{1}_{\widehat{S}|Y},W^{1} ) \in \FEA({\rm LP})$ and $(Q^{2}_{X|S},Q^{2}_{\widehat{S}|Y},W^{2})\in \FEA({\rm LP})$ as follows. 
 $$
 \left( \begin{array}{c}
Q_{X|S}^{*}\\
Q_{\widehat{S}|Y}^{*}\\
W^{*}
 \end{array} \right) = \alpha  \left( \begin{array}{c}
Q^{1}_{X|S}\\
Q^{1}_{\widehat{S}|Y}\\
W^{1}
 \end{array} \right)
 + (1-\alpha) \left( \begin{array}{c}
Q^{2}_{X|S}\\
Q^{2}_{\widehat{S}|Y}\\
W^{2}
 \end{array} \right),
$$
\noindent for some $\alpha \in (0,1).$ From the LP constraints, it follows that $0\leq W^1,W^2 \leq 1$. 
But since $Q_{X|S}^{*},Q_{\widehat{S}|Y}^{*} $ and $W^{*}$ are deterministic,  it is clear that they cannot be written as the above convex combination of any other vectors. Hence,
$(Q_{X|S}^{*},Q_{\widehat{S}|Y}^{*},W^{*})=(Q^{1}_{X|S},Q^{1}_{\widehat{S}|Y},W^{1})=(Q^{2}_{X|S},Q^{2}_{\widehat{S}|Y},W^{2})$
 which is a contradiction.  Thus, we have the result.
\end{proofarg}
\begin{proofarg}[of Proposition~\ref{prop:positiveoptimal}]
Since $d(\cdot,\cdot)\geq 0$, it follows that $\OPT(\LP)=\OPT(\DP) \geq 0$. To show that the optimal objective value of DP is strictly positive, it suffices to show by the strong duality of linear programming that the optimal value of LP cannot be zero. 
Assume the contrary. Let $\OPT(\LP)=0$. Recall that we have the LP objective function expressed as \begin{align*}
\sum_{s,x,y,\shat}P_S(s)P_{Y|X}(y|x)d(s,\shat)W(s,x,y,\shat). 
\end{align*}
However, since $P_S(s)P_{Y|X}(y|x) >0,$ for all $s \in \Sscr$, $x \in \Xscr$, $y \in \Yscr$ and for any fixed $x \in \Xscr$, $y \in \Yscr$, $d(s,\shat)=0$ if and only if $s=\shat$, we have that $\OPT(\LP)=0$ only if for all $x \in \Xscr$ and all $y \in \Yscr$, and all $s\neq \shat$, $W(s,x,y,\shat)=0$.
%
%
Now, consider the LP constraint, \begin{align}\sum_x W(s,x,y,\shat)=Q_{\Shat|Y}(\shat|y) \quad \forall s,\shat,y
\label{eq:contradiction}.
\end{align} For any $y \in \Yscr$ and any $\shat \in \Sscrhat$, putting $s \neq \shat$ in \eqref{eq:contradiction}  implies that  $Q_{\Shat|Y}(\shat|y)=0.$ Since this holds for all $\shat \in \Sscrhat$, this contradicts the LP constraint that for $y \in \Yscr$, $\sum_{\shat}Q_{\Shat|Y}(\shat|y)=1$. The result follows.
\end{proofarg} 
\section{Successive Refinement Source-Coding Problem}\label{app:SR}
We first present the linear programming relaxation of SR. To obtain the relaxation, we define the following new set of variables:
\begin{align*}
A_0(z)&\equiv Q_{X_1|S} Q_{X_2|S}Q_{\Shat_1|Y_1} Q_{\Shat_2|Y_1,Y_2}(z),\\A_1(z_1)&\equiv Q_{X_1|S} Q_{X_2|S} Q_{\Shat_1|Y_1}(z_1),\\
A_2(z_2)&\equiv Q_{X_1|S}(x_1|s)Q_{X_2|S}(x_2|s)\\
A_3(z_3)&\equiv Q_{X_1|S}(x_1|s)Q_{\Shat_2|Y_1,Y_2}(\shat_2|y_1,y_2),\\
A_4(z_4)&\equiv Q_{X_2|S}(x_2|s)Q_{\Shat_1|Y_1}(\shat_1|y_1),\\
A_{5}(z_5)&\equiv Q_{X_1|S} Q_{\Shat_2|Y_1,Y_2} Q_{\Shat_1|Y_1}(z_5),\\
A_{6}(z_{6})&\equiv Q_{X_2|S} Q_{\Shat_2|Y_1,Y_2} Q_{\Shat_1|Y_1}(z_{6}),\\
\end{align*}where $z:=(s,x_1,x_2,y_1,y_2,\shat_1,\shat_2)$, $z_1:=(s,x_1,x_2,\shat_1,y_1)$, $z_2:=(x_1,x_2,s)$,  $z_3:=(x_1,s,\shat_2,y_1,y_2)$, $z_4:=(x_2,s,\shat_1,y_1)$, $z_5:=(s,x_1,\shat_1,y_1,\shat_2,y_2)$, $z_{6}:=(s,x_2,\shat_1,y_1,\shat_2,y_2)$.

The LP relaxation of SR is then obtained as,
$$
\problemsmalla{LPSR}
	{T}
	{\displaystyle\sum_z \kappa(z) P_SP_{Y_1|X_1}P_{Y_2|X_2} A_0(z)}
				 {\begin{array}{r@{\ }c@{\ }l}
				 			T &\in& \Gamma,
	\end{array}}
	$$where $T$ represents the collection of variables $(Q_{X_1|S},Q_{X_2|S},Q_{\Shat_1|Y_1},Q_{\Shat_2,Y_2},A_0,A_1 ,\hdots,A_{6})$ and 
	\begin{align*}
	&\Gamma:=\biggl \lbrace T \geq 0 \mid  \sum_{x_2} Q_{X_2|S}(x_2|s)\equiv 1, \hspace{0.1cm} \sum_{\shat_1} Q_{\Shat_1|Y_1}(\shat_1|y_1)\equiv 1 ,\\&\quad \sum_{\shat_2} Q_{\Shat_2|Y_1,Y_2}(\shat_2|y_1,y_2)\equiv 1, \hspace{0.1cm} \sum_{x_1} Q_{X_1|S}(x_1|s)\equiv 1,\hspace{0.1cm}\\&\quad \sum_{x_1}A_0(z)\equiv A_{6}(z_{6}), \hspace{0.1cm} \sum_{x_2}A_0(z)\equiv A_{5}(z_5)\\&\quad \sum_{\shat_2}A_0(z)\equiv A_1(z_1),\hspace{0.1cm}\sum_{\shat_1}A_1(z_1)\equiv A_2(z_2)\\&\quad\sum_{x_1}A_2(z_2)\equiv Q_{X_2|S}(x_2|s), \\ 
	&\quad \sum_{x_1}A_3(z_3)=Q_{\Shat_2|Y_1Y_2}(\shat_2|y_1,y_2),\\&\quad \sum_{x_2}A_4(z_4)\equiv Q_{\Shat_1|Y_1}(\shat_1|y_1),\hspace{0.1cm} \sum_{\shat_1}A_{5}(z_5)\equiv  A_3(z_3), \\&\quad \sum_{\shat_2}A_{6}(z_{6})\equiv A_4(z_4)\biggr \rbrace,
	\end{align*}  where the constraints in $\Gamma$ are obtained as explained in Section~\ref{subsec:successiveref}.
	Let  $\eta_1(s)$, $\eta_2(y_1)$, $\eta_3(y_1,y_2)$, $\eta^4(s)$, $\lambda^a(z_{10})$, $\lambda^b(z_9)$,   $\lambda^c(s,x_1,x_2,y_1,y_2,\shat_1)$,  $\mu(s,x_1,x_2,y_1)$, $\delta^1(x_2,s)$,  $\delta^2(s,\shat_2,y_1,y_2)$,  $\delta^3(s,\shat_1,y_1)$,  $\theta(x_1,s,y_1,\shat_2,y_2)$ and $\gamma(x_2,s,\shat_1,y_1,y_2)$ be the Lagrange multipliers corresponding to the constraints of $\Gamma$ in that order. Let $F:=(\eta^1,\eta^2,\eta^3,\eta^4,\lambda^a,\lambda^b,\lambda^c,\mu,\delta^1,\delta^2,\delta^3,\theta,\gamma)$ represent the collection of all these Lagrange multipliers.

	With $\Pi(z)\equiv P_S(s)P_{Y_1|X_1}(y_1|x_1)P_{Y_2|X_2}(y_2|x_2)$ $\Ibb\{d_1(s,\shat_1)>D_1$ or $d_2(s,\shat_2)>D_2\}$, the dual program of LPSR is given as DPSR.
\begin{small}
	\begin{figure*}[!t]
\normalsize
\begin{equation}
{\rm DPSR}\quad \qquad \max_{F} \quad \qquad \sum_s\eta^1(s)+ \sum_{y_1}\eta^2(y_1)+\sum_{y_1,y_2}\eta^3(y_1,y_2)\nonumber
\end{equation}
\begin{align*}
\eta^1(s)-\delta^1(x_2,s)&\leq 0, \qquad(D1)\\  
\eta^2(y_1)-\sum_s \delta^3(s,\shat_1,y_1)&\leq 0, \qquad(D2)\\
 \eta^3(y_1,y_2)-\sum_s \delta^2(s,\shat_2,y_1,y_2)&\leq 0, \qquad(D3)\\ \lambda^a(z_{10})+\lambda^b(z_9)+\lambda^c(s,x_1,x_2,y_1,y_2,\shat_1)&\leq \Pi(z),\hspace{0.15cm} (D4)\\
	-\sum_{y_2}\lambda^c(s,x_1,x_2,y_1,y_2,\shat_1)+\mu(s,x_1,x_2,y_1)&\leq 0, \qquad(D5),\\
	-\sum_{y_1}\mu(s,x_1,x_2,y_1)+\delta^1(x_2,s)&\leq 0,\qquad (D6)\\
	-\theta(x_1,s,y_1,\shat_2,y_2)+\delta^2(s,\shat_2,y_1,y_2)&\leq 0,\qquad (D7)\\
	-\sum_{y_2}\gamma(x_2,s,\shat_1,y_1,y_2)+\delta^3(s,\shat_1,y_1) &\leq 0,\qquad(D8)\\
	-\lambda^b(z_9)+\theta(x_1,s,y_1,\shat_2,y_2)&\leq 0, \qquad(D9)\\
	-\lambda^a(z_{10})+\gamma(x_2,s,\shat_1,y_1,y_2)&\leq 0, \qquad(D10)\\
	\eta^4(s)&\leq 0,\qquad(D11)
\end{align*}
\hrulefill
\end{figure*}\end{small}
For notational convenience, we use $P(s)$ to represent $P_S(s)$, $P(y_i|x_i)$ to represent $P_{Y_i|X_i}(y_i|x_i)$, $i=1,2$.

	\begin{proofarg}[ of Theorem~\ref{thm:SR}]
	For the proof, consider the following values of dual variables of DPSR,
	\begin{align*}
	\hspace{-0.2cm}\lambda^c(s,x_1,x_2,y_1,y_2,\shat_1)&\equiv P(s)P(y_1|x_1)P(y_2|x_2)\times\\&\hspace{-2.5cm}\biggl[ \Ibb\{\exp(j_{\Shat_1}(s,D_1)-\gamma_1)\geq M_1 \hspace{0.2cm}\mbox{or}\\&\hspace{-2.5cm} \exp(j_{\Shat_2}(s,D_2)-\gamma_2)\geq M_1M_2\}+ \frac{1}{2M_1}\times\\&\hspace{-2.5cm}\biggl(\exp(j_{\Shat_1}(s,D_1)-\gamma_1)+ \frac{1}{M_2}\exp(j_{\Shat_2}(s,D_2)-\gamma_2) \biggr)\times\\& \hspace{-3cm}\Ibb\biggl \lbrace\frac{\exp(j_{\Shat_1}(s,D_1)-\gamma_1)}{M_1}< 1, \frac{\exp(j_{\Shat_2}(s,D_2)-\gamma_2)}{ M_1M_2}<1\biggr \rbrace\biggr], 
	\end{align*}
	\begin{align*}
	\lambda^a(s,x_2,\shat_1,\shat_2,y_1,y_2)&\equiv -\frac{P(s)}{M_1}\sum_{x_1}P(y_1|x_1)P(y_2|x_2)\times\\&\hspace{-3cm}\exp(j_{\Shat_1}(s,D_1)-\gamma_1) \Ibb\{d_1(s,\shat_1)\leq D_1,d_2(s,\shat_2)\leq D_2\}, \\
	\lambda^b(s,x_1,\shat_1,\shat_2,y_1,y_2)&\equiv -\frac{P(s)}{M_1M_2}\sum_{x_1,x_2}P(y_1|x_1)P(y_2|x_2)\times\\& \hspace{-3cm}\exp(j_{\Shat_2}(s,D_2)-\gamma_2)\Ibb\{d_1(s,\shat_1)\leq D_1,d_2(s,\shat_2)\leq D_2\}, \\
	\gamma(x_2,s,\shat_1,y_1,y_2)&\equiv -\frac{P(s)}{M_1}\sum_{x_1}P(y_1|x_1)P(y_2|x_2)\times\\&\hspace{-1cm}\exp(j_{\Shat_1}(s,D_1)-\gamma_1) \Ibb\{d_1(s,\shat_1)\leq D_1\},\\
	\theta(x_1,s,y_1,\shat_2,y_2)&\equiv -\frac{P_S(s)}{M_1M_2}\sum_{x_1,x_2}P(y_1|x_1)P(y_2|x_2)\times\\&\exp(j_{\Shat_2}(s,D_2)-\gamma_2)\Ibb\{d_2(s,\shat_2)\leq D_2\},\\
	\delta^3(s,\shat_1,y_1)&\equiv \sum_{y_2}\gamma(x_2,s,\shat_1,y_1,y_2),\\
	\delta^2(s,\shat_2,y_1,y_2)&\equiv \theta(x_1,s,y_1,\shat_2,y_2),\\
	\mu(s,x_1,x_2,y_1)&\equiv \sum_{y_2}\lambda^c(s,x_1,x_2,y_1,y_2,\shat_1),\\
	\delta^1(x_2,s)&\equiv \sum_{y_1}\mu(s,x_1,x_2,y_1),\\
	\eta^1(s)&\equiv \delta^1(x_2,s),\\
	\eta^2(y_1)&\equiv -\frac{\exp(-\gamma_1)}{M_1}\sum_{x_1}P(y_1|x_1),\\
	\eta^3(y_1,y_2)&\equiv -\frac{\exp(-\gamma_2)}{M_1M_2}\sum_{x_1,x_2}P(y_1|x_1)P(y_2|x_2),
	\end{align*} and $\eta^4(s) \equiv 0$, where $\gamma_1,\gamma_2$ are scalars to be chosen later. Notice that $\lambda^c$ is independent of $\shat_1$, $\sum_{y_2}\lambda^c(\cdot)$ is independent of $x_2$,  $\sum_{y_1} \mu(\cdot)$ is a function of only $s$, and finally, $\sum_{y_2} \gamma(\cdot) $ is independent of $x_2$.
	
	We now check the feasibility of these dual variables with respect to the constraints. It can be easily seen that (D1) is satisfied trivially. To verify (D2), we have, for any $\shat_1\in\Sscrhat_1$, \begin{align*}&\sum_s\delta^3(s,\shat_1,y_1)=\sum_{s,y_2}-\frac{P_S(s)}{M_1}\sum_{x_1}P(y_1|x_1)P(y_2|x_2)\times\\&\exp(j_{\Shat_1}(s,D_1)-\gamma_1) \Ibb\{d_1(s,\shat_1)\leq D_1\}\\&=-\sum_{x_1}P(y_1|x_1)\frac{\exp(-\gamma_1)}{M_1}\sum_s P_S(s)\exp(j_{\Shat_1}(s,D_1))\times\\& \qquad \qquad\Ibb\{d_1(s,\shat_1)\leq D_1\}\\&\stackrel{(a)}{\geq}-\sum_{x_1}P(y_1|x_1)\frac{\exp(-\gamma_1)}{M_1},
	\end{align*}which is $\eta^2(y_1)$, thereby satisfying (D2). The inequality in (a) follows from $\Ibb{\{d_1(s,\shat_1)\leq D_1\}} \leq \exp(\lambda^{*}_1(D_1-d_1(s,\shat_1))),$  $\lambda^{1*}=-R_S^{(1)'}(D_1)>0$ and using \eqref{eq:LtiltedSR}. The feasibility of the dual variables with respect to (D3) can also be verified in a similar manner.
	
To verify the feasibility of dual variables with respect to (D4), we consider the following cases.\\
Case 1: 	$\Ibb\{d_1(s,\shat_1)>D_1$ or $d_2(s,\shat_2)>D_2\}=1$.\\
In this case $\lambda^a,\lambda^b$ are zero. Consequently, LHS of (D4) becomes $\lambda^c$. We further consider the following subcases.\\
Case 1a: $\Ibb\{\exp(j_{\Shat_1}(s,D_1)-\gamma_1)\geq M_1\hspace{0.2cm}\mbox{or}\hspace{0.2cm} \exp(j_{\Shat_2}(s,D_2)-\gamma_2)\geq M_1M_2\}=1$.\\
In this case, LHS of (D4) becomes $P(s)P(y_1|x_1)P(y_2|x_2)$ which is the RHS of (D4).
\\Case 1b: $\Ibb\{\exp(j_{\Shat_1}(s,D_1)-\gamma_1)\geq M_1\hspace{0.2cm}\mbox{or}\hspace{0.2cm} \exp(j_{\Shat_2}(s,D_2)-\gamma_2)\geq M_1M_2\}=0$.\\
In this case, $M_1> \exp(j_{\Shat_1}(s,D_1)-\gamma_1)$ and $M_1M_2> \exp(j_{\Shat_2}(s,D_2)-\gamma_2)$ and the LHS of (D4) evaluates to
\begin{align*}
&\frac{P_S(s)}{2}P(y_1|x_1)P(y_2|x_2)\biggl(\frac{1}{M_1}\exp(j_{\Shat_1}(s,D_1)-\gamma_1)\\&+ \frac{1}{M_1M_2}\exp(j_{\Shat_2}(s,D_2)-\gamma_2) \biggr),
\end{align*}which is less than $P_S(s)P(y_1|x_1)P(y_2|x_2)$, the RHS of (D4).\\
Case 2: $\Ibb\{d_1(s,\shat_1)>D_1$ or $d_2(s,\shat_2)>D_2\}=0$.\\
In this case, RHS of (D4) is zero. We again consider the following sub-cases.\\
Case 2a: $\Ibb\{\exp(j_{\Shat_1}(s,D_1)-\gamma_1)\geq M_1\hspace{0.2cm}\mbox{or}\hspace{0.2cm} \exp(j_{\Shat_2}(s,D_2)-\gamma_2)\geq M_1M_2\}=1$.\\
In this case, we upper bound $\lambda^c$ using that $\Ibb\{\exp(j_{\Shat_1}(s,D_1)-\gamma_1)\geq M_1\hspace{0.2cm}\mbox{or}\hspace{0.2cm} \exp(j_{\Shat_2}(s,D_2)-\gamma_2)\geq M_1M_2\} \leq $ $\Ibb\{\exp(j_{\Shat_1}(s,D_1)-\gamma_1)\geq M_1\}$ $+\Ibb\{ \exp(j_{\Shat_2}(s,D_2)-\gamma_2)\geq M_1M_2\}$. Consequently, we get the following upper bound on $\lambda^c$,
\begin{align*}
&P(s)P(y_1|x_1)P(y_2|x_2)\biggl[\Ibb\{\exp(j_{\Shat_1}(s,D_1)-\gamma_1)\geq M_1\}\\&+\Ibb\{ \exp(j_{\Shat_2}(s,D_2)-\gamma_2)\geq M_1M_2\}\biggr]\end{align*}\begin{align*}&\stackrel{(a)}{\leq}P(s)P(y_1|x_1)P(y_2|x_2)\biggl[\frac{\exp(j_{\Shat_1}(s,D_1)-\gamma_1)}{ M_1}\\&+\frac{\exp(j_{\Shat_2}(s,D_2)-\gamma_2)}{ M_1M_2}\biggr],
\end{align*}where the inequality follows by upper bounding the indicator function, $\Ibb\{\exp(j_{\Shat_1}(s,D_1)-\gamma_1)\frac{1}{ M_1}\geq 1\} \leq \exp(j_{\Shat_1}(s,D_1)-\gamma_1)\frac{1}{ M_1}$ and $\Ibb\{\exp(j_{\Shat_2}(s,D_2)-\gamma_2)\frac{1}{ M_1M_2}\geq 1\} \leq \exp(j_{\Shat_2}(s,D_2)-\gamma_2)\frac{1}{ M_1M_2}.$
Employing this, LHS of (D4) is upper bounded as,
\begin{align*}
&P_S(s)\biggl[P(y_1|x_1)P(y_2|x_2)\biggl[\frac{\exp(j_{\Shat_1}(s,D_1)-\gamma_1)}{ M_1}\\&+\frac{\exp(j_{\Shat_2}(s,D_2)-\gamma_2)}{ M_1M_2}\biggr]-\frac{\exp(j_{\Shat_1}(s,D_1)-\gamma_1)}{M_1}P(y_2|x_2)\times\\&\sum_{x_1}P(y_1|x_1) -\frac{\exp(j_{\Shat_2}(s,D_2)-\gamma_2)}{M_1M_2}\sum_{x_2}P(y_1|x_1)P(y_2|x_2)\biggr],
\end{align*}which is non-positive, thereby satisfying (D4).\\
Case 2b: $\Ibb\{\exp(j_{\Shat_1}(s,D_1)-\gamma_1)\geq M_1\hspace{0.2cm}\mbox{or}\hspace{0.2cm} \exp(j_{\Shat_2}(s,D_2)-\gamma_2)\geq M_1M_2\}=0$.\\
In this case, it can be easily verified that the LHS of (D4) is non-positive, thereby satisfying (D4).
Thus, the considered variables satisfy (D4).
The dual constraints (D5), (D6), (D7), (D8) and (D11) are trivially satisfied.

To check feasibility with respect to (D9), we have for any $s,x_1,\shat_1,\shat_2,y_1,y_2,$ \begin{align*}&\lambda^b(s,x_1,\shat_1,\shat_2,y_1,y_2)= -\frac{P_S(s)}{M_1M_2}\sum_{x_1,x_2}P(y_1|x_1)P(y_2|x_2)\times\\& \exp(j_{\Shat_2}(s,D_2)-\gamma_2)\Ibb\{d_1(s,\shat_1)\leq D_1,d_2(s,\shat_2)\leq D_2\},\\
&\geq -\frac{P_S(s)}{M_1M_2}\sum_{x_1,x_2}P(y_1|x_1)P(y_2|x_2)\exp(j_{\Shat_2}(s,D_2)-\gamma_2)\times\\& \qquad \Ibb\{d_2(s,\shat_2)\leq D_2\},
\end{align*}which is equal to $\theta(x_1,s,y_1,\shat_2,y_2)$, thereby satisfying (D9). The last inequality follows since $\Ibb\{d_1(s,\shat_1)\leq D_1,d_2(s,\shat_2)\leq D_2\}\leq$  $\Ibb\{d_2(s,\shat_2)\leq D_2\}$. The feasibility of the dual variables with respect to (D10) can be verified in a similar manner.
Thus, the considered dual variables satisfy all the dual constraints and is thus feasible for DPSR.

Consequently, taking the dual cost as a lower bound on $\OPT({\rm SR})$ and taking the supremum over $\gamma_1$, $\gamma_2$ to get the best bound, results in the required bound.
	\end{proofarg}
\bibliographystyle{IEEEtran}
\bibliography{ref,apsbib}
\end{document}